\pdfoutput=1
\documentclass[11pt,a4paper]{article}

\usepackage[T1]{fontenc}
\usepackage[utf8]{inputenc}
\usepackage{lmodern,textcomp}
\usepackage{microtype}
\usepackage[UKenglish]{babel}
\usepackage[margin=1in]{geometry}
\usepackage[mathscr]{eucal}
\usepackage[tbtags,fleqn]{amsmath}
\usepackage{amsthm}
\usepackage{array,enumerate,comment}
\usepackage{etoolbox}
\usepackage{amsmath,amssymb}
\usepackage{graphicx}
\usepackage{tikz}
\usetikzlibrary{arrows.meta,calc,backgrounds}

\usepackage{bm}
\usepackage{mathabx}
\usepackage{booktabs}
\usepackage{thmtools}
\usepackage{thm-restate}

\newcommand{\andrei}[1]{\textcolor{blue}{AD: #1}}
\newcommand{\nop}[1]{}

\newcommand{\R}{\mathbb{R}}

\newcommand{\D}{\mathcal{D}}

\newcommand{\vect}[1]{\mathbf{#1}}
\newcommand{\norm}[1]{\left\lVert#1\right\rVert}
\newcommand{\set}[1]{\left\{#1\right\}}

\newcommand{\defeq}{\stackrel{\text{def}}{=}}

\newcommand{\op}{\mathbin{\mathrm{op}}}
\usepackage[unicode,hidelinks]{hyperref}
\usepackage[capitalise,noabbrev]{cleveref}

\theoremstyle{plain}
\newtheorem{theorem}{Theorem}
\newtheorem{lemma}[theorem]{Lemma}
\newtheorem{corollary}[theorem]{Corollary}

\newtheorem{definition}[theorem]{Definition}
\theoremstyle{definition}
\newtheorem{example}[theorem]{Example}
\theoremstyle{plain}

\title{Numerical Stability of Linear Algebra Operations over Relational Databases}
\author{%
  \href{https://orcid.org/0009-0000-9308-1169}{Andrei Draghici}\\
  University of Zurich\\
  {\small\texttt{andrei.draghici@uzh.ch}}
  \and
  \href{https://orcid.org/0009-0000-6740-8850}{Yuchen He}\\
  University of Zurich\\
  {\small\texttt{yuchen.he@uzh.ch}}
  \and
  \href{https://orcid.org/0000-0002-4682-7068}{Dan Olteanu}\\
  University of Zurich\\
  {\small\texttt{dan.olteanu@uzh.ch}}
}
\date{}

\hypersetup{
  pdftitle={Numerical Stability of Linear Algebra Operations over Relational Databases},
  pdfauthor={Andrei Draghici, Yuchen He, Dan Olteanu},
  pdfkeywords={backward stability, condition number, linear algebra over relational data}
}

\AtEndEnvironment{abstract}{%
  \par\smallskip\noindent
  \textbf{2012 ACM Subject Classification}
  Information systems $\rightarrow$ Database management system engines;
  Mathematics of computing $\rightarrow$ Numerical analysis
  \par\smallskip\noindent
}

\begin{document}

\maketitle

\begin{abstract}
A large body of work in the database literature develops efficient algorithms for linear algebra and machine learning over matrices defined by relational joins, yet the numerical stability of such computations has so far received no attention. This is a practical concern: a join matrix can be much larger than the input database, and the repeated copies of input values it contains compound the floating-point errors incurred by the numerical operations performed over it.

This paper initiates a formal investigation of numerical stability for linear algebra over database joins. We first show that backward stability, the standard yardstick of numerical stability, loses its effectiveness in this setting: join matrices form a structured subspace of the ambient matrix space, so a perturbation explaining a computed result need not correspond to any perturbed input database. This failure already occurs for operations as simple as matrix-vector multiplication.

To overcome this limitation, we introduce projected backward stability, a generalization of backward stability  from the computation of one function to that of a composition of two functions, and establish its connection to classical backward stability. In our database setting, the two functions are the join query and the numerical operation. 

We further introduce the database condition number as the square root of the ratio of maximal to minimal number of copies of input data values into the join matrix, and show that it quantifies how a perturbation of the join matrix is amplified into a perturbation of the input database, independently of the computation used. The database condition number coincides with the classical condition number of the expansion matrix that replicates input values into the join matrix.

Finally, we show that the QR decomposition and the Singular Value Decomposition are projected backward stable and that the FiGaRo algorithm, which pushes these decompositions past $\alpha$-acyclic joins, admits an error bound that can be asymptotically smaller than the size of the join matrix.

\end{abstract}

\section{Introduction}
\label{sec:intro}

In this paper we start an investigation into the numerical stability of linear algebra operations over  matrices defined by database joins. A significant body of research in the database literature, e.g.,~\cite{Kumar2015,Schleich:SIGMOD:2016,Olteanu:VLDB-Keynote:2020,Khamis2020,chen2017linear,FIGARO:VLDBJ:2023}, put forward efficient algorithms for computing linear algebra operations and also for training a variety of machine learning models over such matrices, albeit no consideration was given to their numerical stability. The risk of numerical instability is exacerbated by the multiple copying of input data in the join result that compounds the effect of numerical errors that come with such operations.

Numerical stability is not merely a concern for traditional scientific computing; it directly affects the reliability of modern data systems. Database engines increasingly support numerical analytics, statistical estimation, and machine-learning workloads, where relational joins define large matrices over which downstream computations are performed. The need to integrate complex mathematical computations into data management systems has driven a rich line of research exploring the connection between relational algebra and linear algebra. To bridge the semantic gap between these domains, frameworks like LaRa have been proposed as minimalist unified algebras~\cite{LARA-Pablo}. Foundational frameworks like MATLANG have been introduced to formalize the expressive power of matrix operations and their exact relationship to relational query languages~\cite{MATLANG,expressive-MATLANG}. This theoretical bridge yields direct algorithmic benefits, such as mapping fragments of linear algebra to tractable conjunctive queries to enable constant-delay enumeration and efficient updates for matrix computations~\cite{MATLANG-enum}. Building on this unified perspective, Galley applies modern relational query optimization techniques to compiling sparse tensor programs into highly efficient imperative kernels~\cite{Galley-Kyle}, LaraDB demonstrates that relational engines can efficiently execute linear algebra tasks~\cite{LARA-Dan-Suciu}, and the IFAQ~\cite{IFAQ:CGO:2020} and SDQL~\cite{SDQL:OOPSLA:2022} frameworks allow users to specify, optimize, and efficiently execute workloads that mix linear algebra and relational algebra operations. Fast matrix multiplication is integrated into query plans to lower the complexity of query evaluation~\cite{FMM-submodular}.

However, as database systems take on complex numerical workloads, precision becomes a critical challenge. Prior work has shown that even the computation of simple aggregates can suffer from severe floating-point error: various implementations in modern database systems lose precision due to catastrophic cancellation, while floating-point aggregation can become non-reproducible because different execution orders produce different numerical results~\cite{kamat2016closer,mueller2018reproducible}. Computing the marginal probability of a query in large probabilistic databases is known to incur large round-off errors~\cite{KochO08}. Similar issues arise in modern machine-learning software, where mathematically equivalent implementations may have very different numerical-stability properties and generally inaccurate results~\cite{kloberdanz2022deepstability}. This concern has become particularly acute in the modern era of artificial intelligence. Recent investigations into fundamental deep learning operators reveal significant vulnerabilities to finite-precision effects~\cite{chen2026automatednumericalstabilityanalysis}. Notably, the transformer architecture has been demonstrated to suffer from numerical instability~\cite{budzinskiy2026numericalstabilityanalysislarge}, driving the need for rigorous analyses of both forward and backward stability during model training~\cite{kan2026stabilitytransformerslayernormalization}. These examples suggest that formal correctness of exact arithmetic over large data (such as query results) is insufficient in practice: we need formal guarantees controlling how finite-precision errors propagate through the computation. This motivates our study of numerical stability in computations with matrices defined by database queries.

\begin{example}
    The least squares problem $\operatorname{argmin}_\mathbf{x} \|\mathbf{A}\mathbf{x} - \mathbf{b}\|_2^2$ is central to polynomial regression and a wide range of machine learning tasks. In our setting, the matrix $\mathbf{A}$ is defined by a database query: its rows represent tuples in the query result and its columns represent features. The vector $\mathbf{b}$ is also a column in the query result. The solution $\mathbf{x}$ gives the fitted model parameters~\cite{Schleich:SIGMOD:2016}. 
    Since practical algorithms operate in finite precision arithmetic, the computed result is generally not the exact solution of the mathematical problem~\cite{Higham2002, Demmel97}. \nop{A fundamental question is therefore not only whether an algorithm returns a small residual, but also whether the computation itself is numerically trustworthy.} A key question is how far off the computed result is from the mathematical result.
\end{example}


\begin{figure}[t]
    \centering
    \begin{tikzpicture}[
      dot/.style={circle,fill,inner sep=2pt},
      lab/.style={font=\small},
      >={Stealth[length=3mm]}
    ]
    \node[font=\bfseries\large] at (1.2,2.8) {Input space};
    \node[font=\bfseries\large] at (6.25,2.8) {Output space};
    
    \coordinate (x) at (1.0,2.00);
    \coordinate (xp) at (1.0, 0.2);
    \coordinate (y) at (5.55, 2.00);
    \coordinate (yh) at (5.55, 0.2);
    
    \node[dot] at (x) {};
    \node[dot] at (xp) {};
    \node[dot] at (y) {};
    \node[dot] at (yh) {};
    
    \node[lab,left=1pt] at (x) {$x$};
    \node[lab,left=2pt] at (xp) {$x+\Delta x$};
    \node[lab,right=2pt] at (y) {$y=f(x)$};
    \node[lab,right=1pt] at (yh) {$\hat y=f(x+\Delta x)$};
    
    \draw[->] (x) --  node[midway, above, sloped, fill=white, inner sep=1pt, font=\small] {$f$} (y);
    \draw[->] (xp) -- node[midway, above, sloped, fill=white, inner sep=1pt, font=\small] {$f$} (yh);
    \draw[->,dashed,line width=1pt] (x) -- node[midway, above, sloped, fill=white, inner sep=1pt, font=\small] {computation of $f$} 
    (yh);

    \draw[dotted, line width=0.9pt] (x) --
    node[left, align=center] {backward\\perturbation}
    (xp);
    \draw[dotted, line width=0.9pt] (y) -- node[right] {$\text{forward error}$} (yh);
    
    
    \end{tikzpicture}
    \caption{Illustration of backward stability for a computation of a function $f$. The computation output $\hat{y}$  equals the exact value of $f$ on a slightly perturbed input $x + \Delta x$.
    The smallest possible perturbation $|\Delta x|$ is the backward error, while $|\hat{y}-y|$ is the forward error.}
    \label{fig:forward-backward}
\end{figure}
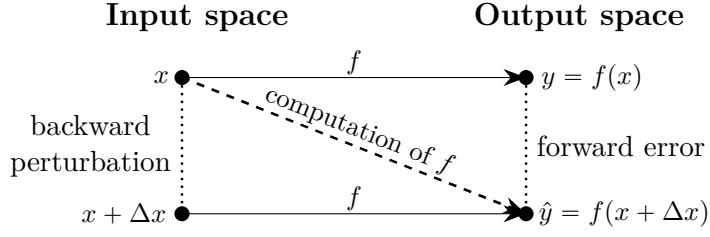

Backward stability is the main measure of numerical stability. A computation of a function $f$ is {\em backward stable} if its output $\hat{y}$ equals the exact value of $f(x+\Delta x)$ on a slightly perturbed input $x + \Delta x$~\cite{Higham2002, Demmel97}, as illustrated in Fig.~\ref{fig:forward-backward}.
For instance, $f$ can be the composition of the join query defining the matrix and the numerical operation computed over the join matrix, whereas $x$ and $\Delta x$ are databases over the same schema and join keys for every join attribute and that may differ in the values for the data attributes. 

\nop{
\begin{example}
Baseline implementations of least-squares regression use the classical normal equation, $\mathbf{A}^\top \mathbf{A}\mathbf{x} = \mathbf{A}^\top \mathbf{b}$, which requires explicitly forming the normal matrix  $\mathbf{A}^\top \mathbf{A}$~\cite{dahiya2018sketching}. Although mathematically equivalent to the original least squares problem under standard rank assumptions, this transformation can significantly worsen numerical behavior and is not, in general, a backward stable way to solve least squares problems~\cite{Bjorck1996, Higham2002, GolubVanLoan2013}. More robust alternatives avoid forming $\mathbf{A}^\top \mathbf{A}$ explicitly: QR decomposition yields a backward stable method for least squares \cite{Bjorck1996, Demmel97}, and singular value decomposition (SVD) provides even stronger guarantees, particularly in ill-conditioned or rank-deficient settings \cite{Bjorck1996}. \nop{However, existing analyses typically bound stability in terms of the dimensions or condition number of the explicitly formed join matrix. }
\end{example}
 }

To our knowledge, {\em this paper is the first to formally analyze the numerical stability of operations over matrices defined by join queries directly with respect to perturbations to the input database.} 
We consider matrices defined by natural join queries. Such join matrices can be much larger than the input database as they may have many copies of input values. These copies compound the effect of computational errors that come with numerical operations. We start with two simple operations, the Frobenius norm and matrix-vector multiplication, and then move to the QR decomposition and Singular Value Decomposition of the join matrix. Our setting calls for a refinement of Fig.~\ref{fig:forward-backward}, as we now have a composite function $g \circ f$, where $f$ denotes the join query and $g$ denotes the subsequent matrix operation. The numerical stability of composed operations has been very recently recognized as an important challenge across disciplines: the numerical analysis community has investigated conditions under which forward stable algorithms compose stably~\cite{10.1093/imanum/drad026}, while the programming languages (PL) community has developed formal methods to check programs for numerical instability~\cite{10.1145/3729394} and to statically derive backward error bounds for composed operations~\cite{10.1145/3729324}. However, our database setting introduces a special type of composition that prior work does not address. Because the intermediate output (the join matrix) is structurally constrained to a specific subspace, we show that the classical notion of stability becomes insufficient. In our setting, we address two fundamental questions: 
\begin{enumerate}
    \item Does the backward stability of the join and numerical operation computations translate to the backward stability of a computation of their composition?
    \item Can tighter error bounds be derived based on the size of the input database and the query structure rather than on the large size of the join matrix?
\end{enumerate}

Our answer to the first question is that, even though (i) any join computation is trivially backward stable as it does not involve arithmetic computations on the data and (ii) there are backward stable algorithms for matrix multiplication, QR, and SVD~\cite{Higham2002}, {\em there may be no backward stable computation of the composition of the join and the numerical operation}. This is because a join matrix $\mathbf{A}$ belongs to a specific subspace of structured matrices, rather than to the full ambient matrix space. Consequently, a small perturbation $\Delta\mathbf{A}$ to $\mathbf{A}$ that explains the computed result may not lie in the join-induced subspace. Thus, there may not exist a perturbed input database for which the join query precisely defines the matrix $\mathbf{A} + \Delta\mathbf{A}$, i.e., the corresponding perturbed input to $g$ may not lie in the range of the join query $f$. \nop{Consequently, there may not exist a perturbed input to $f$ whose image under $f$ coincides with the perturbed input to $g$.}

\begin{figure}[t]
    \centering
    \begin{tikzpicture}[
      dot/.style={circle,fill,inner sep=2pt},
      lab/.style={font=\small},
      >={Stealth[length=3mm]}
    ]
    \node[font=\bfseries\small] at (-0.9, 2.5) {Input space $\mathcal{X}$};
    \node[font=\bfseries\small] at (4.5, 2.5) {Intermediate space $\mathcal{Y}$};
    \node[font=\bfseries\small] at (8.6, 2.5) {Output space $\mathcal{Z}$};
    
    \coordinate (x) at (-0.5,1.65);
    \coordinate (xp) at (-0.5, 0);
    \coordinate (y) at (3.5, 1.65);
    \coordinate (y0t) at (3.5, 0);
    \coordinate (yt) at (3.5, -1);
    \coordinate (z) at (8.6, 1.65);
    \coordinate (zh) at (8.6, -1);

    \node[dot] at (x) {};
    \node[dot] at (xp) {};
    \node[dot] at (y) {};
    \node[dot] at (y0t) {};
    \node[dot] at (yt) {};
    \node[dot] at (z) {};
    \node[dot] at (zh) {};

    \node[lab,left=1pt] at (x) {$x$};
    \node[lab,left=2pt] at (xp) {$x+\Delta x$};
    \node[lab,above=2pt] at (y) {$y=f(x)$};
    \node[lab,anchor=north east]at (y0t) {$\tilde y_0=f(x+\Delta x)\in\mathcal{Y}_0$};
    \node[lab,below=1pt] at (yt) {$\tilde y=\tilde y_0+\Delta y_0^\perp\in\mathcal{Y}$};
    \node[lab,above=1pt] at (z) {$z=g(y)=g(f(x))$};
    \node[lab,below=1pt] at (zh) {$\hat{z}=g(\tilde{y})=g(f(x+\Delta x)+\Delta y_0^\perp)$};

    \draw[->] (x) --  node[midway, above, sloped, fill=white, inner sep=1pt, font=\small] {$f$} (y);
    \draw[->] (y) --  node[midway, above, sloped, fill=white, inner sep=1pt, font=\small] {$g$} (z);
    \draw[->] (xp) -- node[midway, above, sloped, fill=white, inner sep=1pt, font=\small] {$f$} (y0t);
    \draw[->] (yt) -- node[midway, above, sloped, fill=white, inner sep=1pt, font=\small] {$g$} (zh);
    \draw[->,dashed,line width=1pt] (x) -- node[midway, above, sloped, fill=white, inner sep=1pt, font=\small] {computation of $g\circ f$} 
    (zh);

    \draw[dotted, line width=0.9pt] (x) --
    node[left, align=center] {projected\\backward\\perturbation}
    (xp);    \draw[dotted, line width=0.9pt] (y0t) -- (yt);
    \draw[dotted, line width=0.9pt] (z) -- (zh);

    
    \end{tikzpicture}
    \caption{Illustration of projected backward stability for an algorithm that implements the composition $g \circ f$ of two functions, where the first function is $f: \mathcal{X}\mapsto \mathcal{Y}_0$ and the second function $g: \mathcal{Y}\mapsto \mathcal{Z}$ where $\mathcal{Y}\supseteq\mathcal{Y}_0$. The input perturbation $\Delta x$ is small relative to $x$ and the projection error $\Delta y_0^\perp$ is also small relative to $f(x+\Delta x)$.}
    \label{fig:proj-backward}
\end{figure}
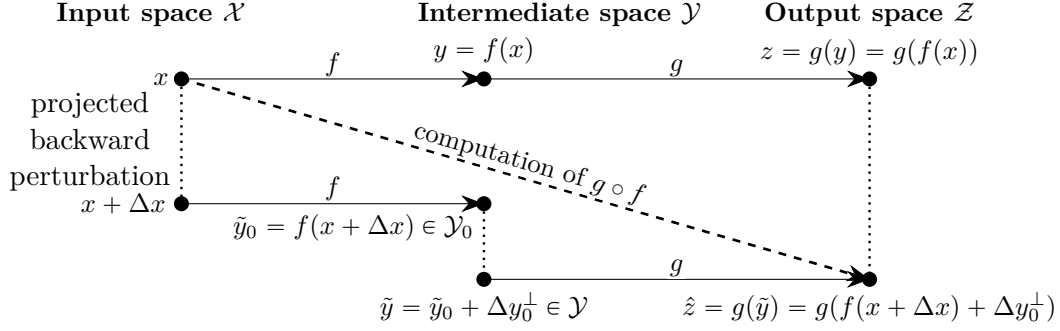

To resolve the domain mismatch between the full matrix space and the join-induced subspace, we introduce the notion of \emph{projected backward stability}, illustrated in Fig.~\ref{fig:proj-backward}. Projected backward stability generalizes classical backward stability to the composition $g \circ f$ of functions $f$ and $g$, where $f$ restricts its outputs to a structured subspace $\mathcal{Y}_0$. 
As in general it is not the case that the perturbed input $\tilde{y}$ to $g$ lies in $\mathcal{Y}_0$, we decompose it into $\tilde{y}=\tilde{y}_0+\Delta y_0^\perp$, where $\tilde{y}_0$ is the projection of $\tilde{y}$ onto $\mathcal{Y}_0$, so that $\tilde{y}_0$ is the output of $f$ for some perturbed input $x+\Delta x$, and $\Delta y_0^\perp$ is a small perturbation orthogonal to the subspace $\mathcal{Y}_0$. That is, $\tilde{y}_0=P_{\mathcal{Y}_0}(\tilde{y})=\underset{\tilde{y}_0\in\mathcal{Y}_0}{\operatorname{argmin}}
\,\left\|\tilde{y}-\tilde{y}_0\right\|_{\mathrm{F}}$, and $\Delta y_0^\perp=\tilde{y}-\tilde{y}_0$. This also means that $\tilde{y}_0$ is the closest to $\tilde{y}$ according to the Frobenius norm $\|\cdot\|_F$.
The projected backward perturbation $\Delta x$ is the difference between the true input $x$ to $f$ and a nearby input $x+\Delta x$ whose image under $f$ matches this projection. 

We argue that projected backward stability is the appropriate numerical stability concept in our setting: it accounts for the (join dependency) constraints while still permitting the reconstruction of a perturbed input database close to the original, in direct analogy to classical backward stability. In particular: $f$ is the join query; $\mathcal{X}$ is the set of databases; $\mathcal{Y}_0$ is the set of join matrices; $g$ is the numerical operation, e.g., matrix-vector multiplication or SVD; and $\mathcal{Y}$ is the set of all matrices. Then, the perturbation $\Delta y_0^\perp$ is the minimal change to a perturbed matrix to turn it into a join matrix for a given join query. Finally, the projected backward perturbation $\Delta x$ is the change to the input database that explains mathematically the result of $g$'s computation. The computation of $g\circ f$ is projected backward stable with respect to the input database if $\Delta x$ is small relative to $x$ and $\Delta y_0^\perp$ is small relative to $\tilde{y}_0$. We show that there is a backward stable computation for $g$ with respect to its input matrix if and only if there is a projected backward stable computation for $g\circ f$ with respect to the input database for every join query $f$. We can also derive the backward error from the projected backward error and vice-versa.

\nop{we study the projected backward stability of operations with respect to the input database $\mathbf{D}$. Consider the following setup: $f$ denotes a query that takes in the input database and outputs a data matrix $\mathbf{A}$, and $g$ denotes a matrix operation. The computation of $f\circ g$ is considered projected backward stable with respect to the input database, if the computed output of $g$ corresponds to an input matrix that can be expressed as $\tilde{\mathbf{A}}+\Delta\mathbf{A}^\perp$, where $\tilde{\mathbf{A}}$ is a matrix in the subspace defined by the query, and a small $\Delta{\mathbf{A}}^\perp$ orthogonal to that subspace (which implies $\tilde{\mathbf{A}}$ is the projection of $\tilde{\mathbf{A}}+\Delta\mathbf{A}$ on that subspace), and $\tilde{\mathbf{A}}$ is the exact output matrix of the same query on a nearby input database numerically close to $\mathbf{D}$.}

Towards an answer to the second question, we show how to bound the projected backward error using a new parameter that we call the \emph{database condition number}. This is a natural analogue to the fundamental notion of condition number\footnote{The condition number is defined as the value of the asymptotic worst-case relative change in output for a relative change in input. A problem with a low condition number (close to 1) is said to be well-conditioned, while a problem with a high condition number is said to be ill-conditioned. A problem is ill-conditioned if a small change in the input leads to a large change in the output. A (projected) backward stable computation can be expected to accurately solve well-conditioned problems.} of matrices, which measures how much a problem's output changes when the input is slightly perturbed~\cite{Higham2002}. The (database) condition number is a property of the problem and not of a particular computation. The database condition number of a relation is defined as the square root of the ratio of the maximal and minimal counts  of an input data value in the join matrix. 
We show that the perturbation in the join matrix yields a perturbation in the input database that is amplified by the database condition number, completely independent from the choice of the computation. This is in addition to the computation-dependent numerical error. In particular, an error upper bound $c$ in the join matrix translates to an error upper bound $\kappa_{\text{join}}\cdot c$ in the input database, where $\kappa_{\text{join}}$ denotes the database condition number. There is a natural relationship of the database condition number and the classical condition number: The former is precisely the condition number of the expansion matrix that transforms the input data columns into the corresponding columns in the join matrix. This transformation creates a number of copies of the input values in the join matrix as dictated by the join query. The rows in this expansion matrix thus select and replicate entries from the input relation.

For the QR and SVD decompositions of the join matrix, we answer the second question in the positive: We give an upper bound on the error of the FiGaRo algorithm~\cite{FIGARO:VLDBJ:2023} that depends on the size of the input database and not on the size of the join matrix.

\medskip

To sum up, the main contributions of this paper are as follows:
\begin{enumerate}
    \item[1.] We initiate a formal investigation on numerical stability of linear algebra operations over database joins. We show that in this setting, the widely accepted yardstick of numerical stability, namely backward stability, loses its effectiveness already for simple operations like the matrix-vector multiplication (Sec.~\ref{sec:bs}). 
    \item[2.] We introduce a natural generalization of backward stability called projected backward stability (Sec.~\ref{sec:projected-bs}) and we show its connection to the classical backward stability (Sec.~\ref{sec:pbs-bs-connection}).
    \item[3.] We introduce the notion of database condition number that amplifies the error bound of the computation over the join matrix into the error in the input database. We also exhibit the connection between this number and the classical condition number (Sec.~\ref{sec:join-condition-number}).
    \item[4.] We show that the QR and SVD decompositions of the join matrix are projected backward stable in our database setting and their backward error bounds depend on the number of computational steps performed over the join matrix. We further show that the \textsc{FiGaRo} algorithm~\cite{FIGARO:VLDBJ:2023}, which pushes the computation of the upper triangular matrix of the QR decomposition past $\alpha$-acyclic joins, can have an asymptotically tighter error bound (Sec.~\ref{sec:qr-svd}). 
    
    \nop{
    \item[3.] We show that for component-wise backward stable operations, such as computation of the Frobenius norm, the resulting error bound on the input relations is independent of the join structure. In contrast, for QR decomposition (which guarantees only column-normwise backward stability) and SVD (which guarantees only Frobenius-normwise backward stability), the backward error bound with respect to the input relations depends on the join condition number of the database.
    \item[3.]  [pending proof] We propose a numerically stable method to explicitly express the matrix $Q$ induced by the FIGARO algorithm~\cite{FIGARO:VLDBJ:2023} with orthogonality guarantee, using only tuple counts that do not depend on the data columns of any input relations. With this $Q$ and projected backward stability, the perturbed input relations corresponding to the computed QR or SVD factors can be explicitly reconstructed and remain close to the original input relations.
    }
\end{enumerate}

\section{Preliminaries}
\label{sec:preliminaries}

In this section, we recall standard relational algebra and linear algebra terminology and matrices defined by database joins, as used in prior work~\cite{FIGARO:VLDBJ:2023}.

\smallskip

\noindent{\bf Matrices.} We denote matrices by bold upper-case letters and column vectors by bold lower-case letters. For a real number $x \in \mathbb{R}$, $|x|$ denotes its absolute value. By $[n]$ we denote the set $\{1,\dots,n\}$. For a matrix $\vect{A}\in\mathbb{R}^{m \times n}$ and indices $i\in[m]$, $j\in[n]$, we use $\vect{A}[i,j]$, $\vect{A}[i,*]$, and $\vect{A}[*,j]$ to denote its entry at row $i$ and column $j$, its $i$-th row, and respectively its $j$-th column. The Frobenius norm of $\vect{A}$ is $\norm{\vect{A}}_F \defeq \sqrt{\sum_{i=1}^{m}\sum_{j=1}^{n} \vect{A}[i,j]^2}\;$. For a vector $\vect{a}\in\mathbb{R}^{m}$, $\vect a[i]$ denotes its $i$-th component and its Euclidean norm is $\norm{\vect{a}}_2 \defeq \sqrt{\sum_{i=1}^m \vect{a}[i]^2}\;$. 
For matrices $\vect A,\vect B\in\R^{m\times n}$, $|\vect A|\leq |\vect B|$ means that $|\vect A[i,j]|\leq |\vect B[i,j]|$ for all $i\in[m],j\in[n]$.  

\smallskip

\noindent{\bf Databases and Queries}
A database consists of a set of relations $S_1,\dots,S_r$. Each relation $S_i$ has schema $(\mathbf{X}_i, \mathbf{Y}_i)$, where $\mathbf{X}_i$ is the tuple of its join attributes and $\mathbf{Y}_i$ is the tuple of its data attributes. Each data attribute takes values in $\R$. For relations $S_i$ and $S_j$, we write $\mathbf{X}_{ij}$ for the tuple of join attributes shared by both relations. 
We consider full conjunctive (or join) queries expressed as $Q = S_1(\mathbf{X}_1, \mathbf{Y}_1)\bowtie\dots\bowtie S_r(\mathbf{X}_r, \mathbf{Y}_r)$, where two relations $S_i$ and $S_j$ equi-join on $\mathbf{X}_{ij}\neq\emptyset$. A query is $\alpha$-acyclic if and only if it admits a join tree~\cite{DBfoundations-book}. \nop{A database is fully reduced with respect to $Q$ if every tuple of an input relation contributes to at least one tuple in the result of $Q$; a database can be fully reduced in two passes over the database for any $\alpha$-cyclic query~\cite{DBfoundations-book}.}

\smallskip

\noindent{\bf From Databases to Matrices}
A relation $S_i$ with schema $(\mathbf{X}_i,\mathbf{Y}_i)$ defines a row-labeled, column-labeled real-valued matrix $\mathbf{S}_i$, where the rows of $\mathbf{S}_i$ are indexed by the tuples of $S_i$ and the columns are indexed by the attributes in $\mathbf{Y}_i$.
For every tuple $s_i=(\mathbf{x}_i,\mathbf{y}_i)\in S_i$, \nop{we define $\mathbf{S}_i[s_i] = \mathbf{y}_i$. Hence} the join keys $\mathbf{x}_i$ are not entries of $\mathbf{S}_i$; they appear only in the row label. Two tuples with the same join key but different data are distinct rows. We denote by $\vect S_i[*,\vect Y]$ the matrix formed by the columns in $\vect Y$ of $\vect S_i$.
A block-diagonal matrix $\vect A_{\text{block}}$ with blocks $\vect A_1,\ldots,\vect A_b$ is compactly denoted by the tuple $(\vect A_1,\ldots,\vect A_b)$. The matrix $\vect D = (\vect S_1,\ldots,\vect S_r)$ defined by the database $\mathcal{D}$ is block-diagonal, with the blocks $\vect S_1,\ldots,\vect S_r$. The matrix defined by the result of a join query over a database is called the {\em join matrix}, denoted by $\mathbf{A}$. We fix an arbitrary order of the row labels of a matrix; row permutations do not affect the algebraic quantities considered in this paper.

\nop{For every tuple $s_i=(\mathbf{x}_i,\mathbf{y}_i)\in S_i$, we define $\mathbf{S}_i[s_i,*] = \mathbf{y}_i$.} 

\nop{
For every tuple $s_i=(\mathbf{x}_i,\mathbf{y}_i)\in S_i$, we define $\mathbf{S}_i[s_i] = \mathbf{y}_i$. Hence the join keys $\mathbf{x}_i$ are not entries of $\mathbf{S}_i$; they appear only in the row identifiers. Two tuples with the same join key but different data are distinct rows.
}



Given a relation $S$, we denote by $\Delta S$ a relation with the same schema and the same join keys, but with possibly different data values. The sum $\vect S+\Delta \vect S$ is the entry-wise addition of the data values. For a database $\D=(S_1,\dots,S_r)$, we denote $\Delta \D$ by $(\Delta S_1,\dots,\Delta S_r)$ and $\vect D+\Delta \vect D$ by $(\vect S_1+\Delta \vect S_1,\dots,\vect S_r+\Delta \vect S_r)$.
The \emph{join space} of $\D$ is the set of all join matrices corresponding to a perturbed database $\D +\Delta \D$ of $\D$. Such perturbed databases induce a join space that is a subspace of the matrices in $\R^{m\times n}$ (Lemma~\ref{lemma:join_consistent_subspace}). \nop{The Frobenius norm of $\vect D$ is defined as: $\|\vect D\|_F=\sqrt{\sum_{i=1}^r\|\mathbf{S}_i\|_F^2}$.}

\begin{example}
Consider the join $Q
=
S_1(Y_1,X_{12})
\bowtie
S_2(X_{12},Y_2,X_{23})
\bowtie
S_3(X_{23},Y_3)$ and the input relations and the join matrix shown below:
{\normalsize
\setlength{\arraycolsep}{3pt}
\renewcommand{\arraystretch}{1.1}
\[ \begin{array}{@{}c@{\hspace{0.5em}}c@{\hspace{0.5em}}c@{\hspace{3.7em}}c@{}}
\begin{array}{c}
S_1 \\[2pt]
\begin{array}{c|c}
X_{12} & Y_1 \\
\hline
k_{11} & a_1 \\
k_{12} & a_2 \\
\end{array}
\end{array}
&
\begin{array}{c}
S_2 \\[2pt]
\begin{array}{c|c|c}
X_{12} & X_{23} & Y_2 \\
\hline
k_{11} & k_{21} & b_1 \\
k_{12} & k_{21} & b_2 \\
\end{array}
\end{array}
&
\begin{array}{c}
S_3 \\[2pt]
\begin{array}{c|c}
X_{23} & Y_3 \\
\hline
k_{21} & c_1 \\
\end{array}
\end{array}
&
\begin{array}{c}
\mathbf{A} \\[2pt]
\begin{array}{c|ccc}
\text{row label in query result} & Y_1 & Y_2 & Y_3 \\
\hline
(k_{11},k_{21},a_1,b_1,c_1) & a_1 & b_1 & c_1 \\
(k_{12},k_{21},a_2,b_2,c_1) & a_2 & b_2 & c_1 
\end{array}
\end{array}
\end{array} \]
}

The keys of the join attributes $X_{12}$ and $X_{23}$ are only retained in the row labels of $\mathbf{A}$; the actual matrix entries are values for the data attributes $Y_1,Y_2,Y_3$.
\end{example}

\begin{definition}[Our DB Setting] \label{def:eval-setting}

\begin{itemize}
    \item The input database $\mathcal{D}=(S_1,\dots,S_r)$ defines the row-labeled, column-labeled matrices $\vect{D} = (\vect{S}_1, \dots, \vect{S}_r)$. We denote by $d$ the total number of rows in  $\vect{D}$ and by $\ell$ ($n$) the total number of values (columns) in $\vect{S}_1,\ldots,\vect S_r$.
    \item A perturbation to $\vect{D} = (\vect{S}_1, \dots, \vect{S}_r)$ is denoted by $\Delta\vect{D} = (\Delta\vect{S}_1, \dots, \Delta\vect{S}_r)$. 
    \item A bound $|\Delta\vect{D}| \leq c\cdot |\vect{D}|$ is a shorthand notation for $|\Delta\vect{S}_i|\leq c\cdot |\vect{S}_i|$ for $i\in[r]$. 
    \item The join query is $Q = S_1 \bowtie \dots \bowtie S_r$ and has no self-joins.
    \item The query result $Q(\mathcal{D})$ defines the join matrix $\vect{A} \in \mathbb{R}^{m\times n}$ that has $m$ rows and $n$ columns.
\end{itemize}
\end{definition}

\smallskip

\noindent{\bf Numerical Stability.}
In numerical computations, it is essential to quantify the effects of rounding errors that arise due to finite-precision arithmetic. Let $\operatorname{fl}(x)$ denote the floating-point number obtained by rounding the real number $x$ to the nearest floating-point number representable by the computer. The \emph{unit round-off}, denoted by $u$, measures the relative error incurred in representing a real number in floating-point arithmetic. Then, 
$\operatorname{fl}(x \op y) = (x \op y) (1 + \delta), \text{ where } |\delta| \leq u, \text{ for some } \delta\in\R,\op\in\set{+,-,\times,/}$.
A convenient and common notation for bounding the accumulation of rounding errors is the quantity $\gamma_n \defeq \frac{nu}{1 - nu}$. When $nu < 1$, it follows from standard rounding models that a product of $n$ floating-point factors satisfies
$\operatorname{fl}(x_1 \cdots x_n) = (x_1 \cdots x_n)(1 + \theta_{n-1}), \text{ for } |\theta_{n-1}| \le \gamma_{n-1}$.
Thus $\gamma_n$ gives a compact way to express worst-case relative error growth in sequences of operations, especially products. It is often useful to allow in such bounds a small positive integer $c$ constant independent of the problem size as a multiple of $nu$: $\tilde{\gamma}_n \defeq \frac{cnu}{1 - cnu}$. This notation $\tilde{\gamma}_n$ enables concise statements of error bounds when several rounding terms combine. Both $\gamma_n$ and $\tilde{\gamma}_n$ are treated as quantities of order $nu$ whenever $nu \ll 1$ and $cnu \ll 1$~\cite{Higham2002}.

A numerical algorithm is \emph{backward stable} if the computed solution can be interpreted as the exact solution to a slightly perturbed input. 
\begin{definition}[Backward Stability~\cite{Higham2002}]
    Let $C$ be a numerical computation of a function $f$: if $y = f(x)$ for some $x$ then $\hat{y}=C(x)$. We say that $C$ is \emph{backward stable} if and only if there exists $\Delta x$ such that $\hat{y} = f(x+\Delta x)$ and $\Delta x$ is small relative to $x$. The \emph{backward error} is defined as $\min_{\Delta x: f(x+\Delta x)=\hat{y}}|\Delta x|$. The computation $C$ is \emph{$\epsilon$-backward stable} if for any input $x$, its backward error satisfies $|\Delta x|\leq\epsilon\cdot |x|$.
\end{definition}
We use the convention that computed quantities wear a hat: $\hat{y}$ denotes the computed approximation to $y$. A numerical computation is \emph{forward stable} if the computed solution $\hat{y}$ is close to the exact output: $f(x) = \hat{y} + \Delta y$, where $\Delta y$  is small relative to $y$. A computation of $f$ is \emph{mixed forward-backward stable} if its output is close to the exact output corresponding to a nearby admissible input: $f(x+\Delta x) = \hat{y} + \Delta y$.

\noindent{\bf Perturbation Bounds.} For a matrix $\mathbf A$ and a perturbation $\Delta\mathbf A$, we distinguish three types of relative perturbation bounds: (1) Entry-wise: $|\Delta\mathbf A[i,j]|\leq\epsilon\cdot|\mathbf A[i,j]|$; (2) Column-wise: $\|\Delta\mathbf A[*,j]\|_2\leq\epsilon\cdot\|\mathbf A[*,j]\|_2$; and (3) Frobenius-norm-wise:  $\|\Delta\mathbf A\|_F\leq\epsilon\cdot\|\mathbf A\|_F$. These bounds are progressively weaker: an entry-wise $\epsilon$-perturbation is also a column-wise $\epsilon$-perturbation, which is also a Frobenius-norm-wise $\epsilon$-perturbation. Entry-wise backward stability implies column-wise backward stability, which implies Frobenius-norm-wise backward stability.

\noindent{\bf Condition Number of a Matrix~\cite{Higham2002}.} Consider a rank-$k$ matrix with its $k$ non-zero singular values $\sigma_1\geq\cdots\geq\sigma_k>0$. Its spectral-norm condition number, denoted hereafter as the condition number, is $\kappa(\mathbf A)\defeq\sigma_1/\sigma_k$.
\section{Backward Stability with Respect to the Input Database}
\label{sec:bs}

We initiate the investigation of backward stability of simple numerical operations over join matrices {\em with respect to the input database}. Whereas the stability of many numerical operations with respect to their input matrices is well understood~\cite{Higham2002}, our setting, which investigates the stability of the composition of the numerical operation with the join, has not been addressed so far. In this section, we exemplify our setting for two simple numerical operations: the Frobenius norm and matrix-vector multiplication.

Computing the Frobenius norm over an arbitrary matrix is backward stable~\cite{Higham2002} and remains backward stable also with respect to the input database.

\begin{restatable}{proposition}{frobeniusnorminputrelationsnaive}
\label{proposition:frobenius_norm_input_relations_naive}
Consider the setting in Def.~\ref{def:eval-setting}. The computation $C$ of an operation $f$, which first materializes the join matrix and then computes its Frobenius norm, is backward stable with respect to the input database. That is, there exists a perturbation $\Delta \mathbf{D}$ such that, if $\hat z=C(\mathbf{D})$, then $\hat z = f(\mathbf{D}+\Delta \mathbf{D})$. Moreover, $|\Delta \vect D|\le \tilde\gamma_{mn}\cdot|\vect D|$.
\end{restatable}
\begin{proof}
Let $z=f(\mathbf{D})=\|\mathbf A\|_F$, where the join matrix $\vect A$ is fully materialized. 
If $z=0$, then every entry of $\mathbf A$ is zero, and we choose $\Delta\mathbf D=\mathbf 0$. Suppose now that $z>0$. The computation $C$ of $f$ needs one pass over $\mathbf A$, so it does $mn$ square operations followed by the square root of their sum. Thus the computation requires $O(mn)$ arithmetic operations. Since all terms in the sum are nonnegative and the computation of the square root is forward stable, the forward error satisfies $|\hat z-z|\leq \tilde\gamma_{mn}z$. We write $\hat z-z=\theta z$ for some $\theta$ such that $|\theta|\leq \tilde\gamma_{mn}$.

We construct $\Delta \mathbf{D}$ by scaling every numerical attribute included in the matrix $\mathbf{A}$ by the factor $1+\theta$, denoted compactly by $\Delta \mathbf{D}=\theta\mathbf{D}$. Then, $\mathbf{D} + \Delta \mathbf{D}=(1+\theta)\mathbf{D}$. Now, consider the result of $f(\mathbf{D}+\Delta\mathbf{D})$. The key attributes are unchanged and every entry of the materialized join matrix is scaled by the same factor $1+\theta$. Then, ${\mathbf A} + \Delta\vect A =(1+\theta)\mathbf A$ and $\|{\mathbf A} + \Delta\vect A\|_F=\|(1+\theta)\mathbf A\|_F=|1+\theta|\cdot\|\mathbf A\|_F=|1+\theta|z=(1+\theta)z=\hat z$. In the previous chain of equalities, $|1+\theta| = 1+\theta$ since $|\theta|\leq \tilde\gamma_{mn}$. Finally, $\Delta\vect D=\theta \vect D \implies |\Delta\vect D|\leq \tilde\gamma_{mn}\cdot|\vect D|$.
\end{proof}
Prop.~\ref{proposition:frobenius_norm_input_relations_naive} establishes backward stability by analyzing the standard algorithm that first materializes the join matrix and then computes its Frobenius norm. However, materializing the join is unnecessary, as observed for aggregates over joins~\cite{Schleich:SIGMOD:2016}: we can push the computation of the norm past the join. Besides avoiding the materialization, this yields a smaller error bound. Whereas the bound in Prop.~\ref{proposition:frobenius_norm_input_relations_naive} depends on the number of values in the join matrix ($mn$), the improved bound only depends on the number of values in the input database ($\ell$). Since the materialized join can be much larger than the input database ($\ell \ll mn$), this means a substantially smaller backward error.

\begin{restatable}{proposition}{frobeniusnorminputrelationsoverthejoin}
\label{proposition:frobenius_norm_input_relations_over_the_join}
Consider the setting in Def.~\ref{def:eval-setting}. The computation $C$ of an operation $f$, which pushes the computation of the Frobenius norm of the join matrix past the join, is backward stable with respect to the input database. That is, there exists a perturbation $\Delta \mathbf{D}$ such that, if $\hat z=C(\mathbf{D})$, then $\hat z = f(\mathbf{D}+\Delta \mathbf{D})$. Moreover, $|\Delta \vect D|\le \tilde\gamma_{\ell}\cdot|\vect D|$.
\end{restatable}

Even though standard matrix-vector multiplication is backward stable with respect to the join matrix~\cite{Higham2002}, it is not with respect to the input database: there exist problem instances for which no perturbation of the input database can explain the computed result.

\begin{restatable}{proposition}{matrixvectormultiplicationinputrelations}
\label{proposition:matvec_input_relations_counterexample}
Consider the setting in Def.~\ref{def:eval-setting} and the operation $f$ that defines the matrix-vector multiplication $f_{\vect x}(\vect D) = \vect A\vect x$ for a vector $\vect x$. Then, there exist a join matrix $\vect A$ and vector $\vect x$ for which the standard computation of matrix-vector multiplication is not backward stable with respect to the input database. 
\end{restatable}

\begin{proof}
Consider the join matrix defined by the natural join of the relations $S_1(X, A) = \{(k, t), (k, 0)\}$ and $S_2(X, B) = \{(k, +1), (k, -1)\}$, where $X$ is the join attribute and $A$ and $B$ are data attributes and $t=2^{p+2}$ for $p$ defined below. For any floating point number system~\cite{Higham2002} with base 2, each nonzero normalized floating point number can be represented by $\mu\times2^{e-p}$, where $e$ and $\mu$ are both integers, $p$ is the positive integer precision, and $2^{p-1}\leq|\mu|\leq 2^p-1$.\footnote{For IEEE double: $p=53$ and $-1021 \leq e \leq 1024$. Then, $t = 2^{55}$ and $fl(2^{55}+1) = fl(2^{55}-1)=2^{55}$.} Let $t=2^{p+2}$; $t$ can be represented using $\mu=2^{p-1}$ and $e=p+3$. The next floating point number has $\mu=2^{p-1}+1$ and $e=p+3$ and its value is $(2^{p-1}+1)\times2^{p+3-p}=t+8$. The previous number has $\mu=2^p-1$ and $e=p+2$ and its value is $(2^p-1)\times2^{p+2-p}=t-4$. No numbers between $t$ and $t+8$, as well as between $t-4$ and $t$, exist in this floating point number system. Therefore, the integers $t+1$ and $t-1$ will both be rounded to the closest floating point number, which is $t$. We therefore have: $\operatorname{fl}(t+1)=\operatorname{fl}(t-1)=t$ for $t=2^{p+2}$.

Let $\mathbf{x}=(1,1)^\top$. With data columns ordered as $(A,B)$, the join matrix $\mathbf{A}$, exact multiplication $\mathbf{A}\mathbf{x}$, and the row-wise floating-point computation $\hat{\mathbf{y}} = \operatorname{fl}(\mathbf{A}\mathbf{x})$ are:
\[
\mathbf{A}=\begin{pmatrix} t&+1\\ t&-1\\ 0&+1\\ 0&-1 \end{pmatrix}, \quad
\mathbf{A}\mathbf{x}=\begin{pmatrix} t+1\\ t-1\\ +1\\ -1 \end{pmatrix}, \quad
\hat{\mathbf{y}}=\begin{pmatrix} t\\ t\\ +1\\ -1 \end{pmatrix}, \quad 
\mathbf{A}'=\begin{pmatrix} \! a_1 & b_1 \!\\ \! a_1& b_2 \!\\ \! a_2& b_1 \!\\ \! a_2& b_2 \!\end{pmatrix}\!, \quad 
\mathbf{A}'\mathbf{x}=\begin{pmatrix} \! a_1+ b_1 \!\\ \! a_1+ b_2 \!\\ \! a_2+ b_1 \!\\ \! a_2+ b_2 \!\end{pmatrix}\!.
\]

We now show by contradiction that the computed vector $\hat{\mathbf{y}}$ cannot be obtained as the exact matrix-vector multiplication over any perturbation of the numerical attributes of the input database. Suppose there exist relations $S'_1$ and $S'_2$ that (i) are perturbations of $S_1$ and $S_2$, respectively, and (ii) define the input matrices $\vect D'=(\vect S'_1, \vect S'_2)$ and the join matrix $\vect A'$, such that $f_{\vect x}(\vect D') = \hat{\vect y}$. This means that $\vect A'\vect x = \hat{\vect y}$. Since $S'_1$ and $S'_2$ are  $S_1$ and $S_2$, where the numerical values may be perturbed, they have the general form $S'_1=\{(k,a_1),(k,a_2)\}$ and $S'_2=\{(k,b_1),(k,b_2)\}$ for the same join key $k$ and numerical values $a_1,a_2, b_1,b_2\in\mathbb{R}$. Then, the new join matrix $\vect A'$ and the multiplication $\vect A'\vect x$ are as shown above.

Since ${\mathbf{A}}'\mathbf{x}=\hat{\mathbf{y}}$, we obtain the equalities $a_1 + b_1 =t$, $a_1+ b_2=t$, $a_2+ b_1=+1$, and $a_2+b_2=-1$ by inspecting the first four rows of $\mathbf{A}'\mathbf{x}$ and $\hat{\mathbf{y}}$. Subtracting the first two equations gives $b_1-b_2=0$, while subtracting the fourth equation from the third gives $b_1-b_2=2$. We reached a contradiction. Therefore, the numerical computation of $\mathbf{A}\mathbf{x}$ is not backward stable with respect to the input database.
\end{proof}

Although exemplified for two simple operations, the above observations apply to all linear algebra operations over the join matrix. The takeaway is that even though the join is trivially backward stable, since it does not touch the data values in the database, numerical operations that are backward stable with respect to the matrix may trivially lose this desirable property when analyzed with respect to the input database! So the composition of two backward stable operations may not be backward stable in our setting. The main reason is that the join forces copies of the same input value in the join matrix, yet the numerical operations over the join matrix may introduce different numerical errors for these copies. In particular, if these operations produce an output whose size is more than the size of the input database, then backward stability is often lost since we may not be able to propagate the many independent errors to a smaller number of input data values. This is the case for workhorse linear algebra operations, such as QR. This calls for a new notion of stability for the composition of a query with a numerical operation and that allows us to quantify the numerical error introduced by the computation. We discuss this new notion in the next section.

\section{Projected Backward Stability}
\label{sec:projected-bs}

Prop.~\ref{proposition:matvec_input_relations_counterexample} shows that the standard multiplication of the join matrix with a vector is not backward stable with respect to the input database. This failure, however, should not be interpreted as numerical instability, but it is due to the rigid structure of join matrices: their numerical perturbations may not necessarily be join matrices. Yet although perturbed join matrices may not be defined by the query over a perturbed input database, they remain close to valid join matrices. This observation motivates a generalized notion of backward stability that takes the query into account. We first define this notion for arbitrary functions $f$ and $g$ and then specialize it to our setting, where $f$ is the join and $g$ the numerical operation.

\begin{definition}[Projected backward stability]
\label{def:proj_b_s}
Let $\mathcal X$ be a normed space, $\mathcal Y$ be a finite-dimensional Hilbert space (so that projections are well-defined), and let $\mathcal Y_0\subseteq\mathcal Y$ be a subspace. Let $f:\mathcal X\to\mathcal Y_0$, $g:\mathcal Y\to\mathcal Z$, and $h=g\circ f$. Let $\mathcal{Y}_0^\perp$ denote the orthogonal complement of the subspace $\mathcal{Y}_0$ in $\mathcal{Y}$. Let $C$ be a numerical computation of $g\circ f$: if $z=g(f(x))$ for some $x$, then $\hat{z}=C(x)$. We say that $C$ is {\em projected backward stable} if there exists a perturbation $\Delta x\in\mathcal X$ and a normal perturbation $\Delta y_0^\perp\in\mathcal Y_0^\perp$ such that $\hat z=g\bigl(f(x+\Delta x)+\Delta y_0^\perp\bigr)$, $\Delta x$ is small relative to $x$, and $\Delta y_0^{\perp}$ is small relative to $f(x+\Delta x)$. 

The computation $C$ is $(\alpha, \beta)$-projected backward stable if for any input $x$, there exists $\Delta x$ and $\Delta y_0^\perp$ described above satisfying $|\Delta x|\leq \alpha\cdot|x|$ and $|\Delta y_0^\perp|\leq\beta\cdot|f(x+\Delta x)|$.
\end{definition}

The standard notions of normed spaces and Hilbert spaces are recalled in Appendix~\ref{app:preliminaries}.
Projected backward stability is depicted in Fig.~\ref{fig:proj-backward}. In our database setting (Def.~\ref{def:eval-setting}), $f$ is the join query, $\mathcal{X}$ is the set of matrices $\mathbf{D}=(\mathbf{S}_1,\ldots,\mathbf{S}_r)$ defined by relational databases, $\mathcal{Y}_0$ is the set of join matrices $\mathbf{A}$ such that $f(\mathbf{D})=\mathbf{A}$, $\mathcal{Y} \supseteq \mathcal{Y}_0$ is the unconstrained set of all $m \times n$ matrices, and $g$ is the numerical operation applied to the join matrix $\mathbf{A}$. The orthogonal complement $\mathcal{Y}_0^\perp$ captures the perturbations that break the join dependencies from the query. The normal perturbation $\Delta y_0^\perp$ records the component that cannot be represented by perturbing the database, and its magnitude is controlled separately. Notice that if the intermediate space lacks constraints (i.e., $\mathcal{Y}_0 = \mathcal{Y}$), then $\Delta y_0^\perp = 0$, recovering the standard notion of backward stability. A computation $C$ of $g\circ f$ is projected backward stable with respect to $\mathbf{D}$ if $C(\mathbf{D})=\hat{z}=g\bigl(\tilde{\mathbf{A}}+\Delta\mathbf{A}^{\perp}\bigr)$ for some join matrix $\tilde{\mathbf{A}}\in\mathcal{Y}_0$ and perturbation $\Delta\mathbf{A}^{\perp}\in\mathcal{Y}_0^\perp$, where $\Delta\mathbf{A}^{\perp}$ is relatively small compared to $\tilde{\mathbf{A}}$. Furthermore, $\tilde{\mathbf{A}}$ must correspond to a small relative perturbation of the input database, i.e., $\tilde{\mathbf{A}}=f(\mathbf{D}+\Delta\mathbf{D})$, where $\Delta\mathbf{D}$ is relatively small compared to $\mathbf{D}$.

We are now ready to apply this new stability lens to the matrix-vector multiplication $\vect A\vect x$. The power of projected backward stability lies in its ability to reconcile floating-point errors with the structure of the join matrix. By projecting the inevitable rounding errors of the matrix-vector multiplication back onto the subspace of valid join matrices, the computed output maps to a small, structurally valid perturbation of the input database, as stated next.

\begin{restatable}{proposition}{Axisprojectedbackwardstable}
\label{proposition:matvec_proj_b_s}
Consider the setting in Def.~\ref{def:eval-setting}, the join query $f$ such that $f(\mathbf{D})=\vect A$,  $g_{\vect x}(\vect A) = \vect A\vect x$ for a vector $\vect x$, and the join space $\mathcal{Y}_0$ of $\mathbf{A}$. Let $C$ be the computation of $g_{\vect x}\circ f$ that forms $\mathbf A$ and computes $\mathbf A\mathbf x$ using dot product. Then, $C$ is $(\tilde\gamma_n, \frac{2\tilde\gamma_n}{1-\tilde\gamma_n})$-projected backward stable if $\tilde\gamma_n<1$. That is, there exist $\tilde{\mathbf{A}}\in\mathcal{Y}_0$, $\Delta\mathbf A^\perp\in\mathcal{Y}_0^{\perp}$ and a perturbation $\Delta\mathbf{D}$ of $\mathbf{D}$ such that $f(\mathbf{D}+\Delta\mathbf{D})=\tilde{\mathbf{A}}$, $g_{\vect x}(\tilde{\vect A}+\Delta\mathbf{A}^\perp)=C(\mathbf{D})$, $|\Delta\mathbf A^\perp|\leq\frac{2\tilde\gamma_n}{1-\tilde\gamma_n}\cdot|\tilde{\mathbf A}|$, and $|\Delta \mathbf{D}|\leq\tilde\gamma_n\cdot|\mathbf D|$.
\end{restatable}


\nop{
\begin{restatable}{proposition}{Axisprojectedbackwardstable}
\label{proposition:matvec_proj_b_s}
Consider the setting in Def.~\ref{def:eval-setting} and the join query $f$ such that $f(\mathbf{D})=\vect A$, and $g_{\vect x}(\vect A) = \vect A\vect x$ for a given vector $\vect x$. Suppose that $\hat{\mathbf y}$ is obtained by materializing $\mathbf A$ and then evaluating $\mathbf A\mathbf x$ using the standard dot-product algorithm. Then this procedure is $(\tilde\gamma_n, \frac{2\tilde\gamma_n}{1-\tilde\gamma_n})$-projected backward stable with respect to $\mathbf{D}$. More precisely, there exist $\tilde{\mathbf{A}}\in\mathbb R^{m\times n}$, $\Delta\mathbf A^\perp\in\mathcal{Y}_0^{\perp}$ and a perturbation $\Delta\mathbf{D}=(\Delta\mathbf S_1,\ldots,\Delta\mathbf S_r)$ such that $f(\mathbf{D}+\Delta\mathbf{D})=\tilde{\mathbf{A}}$ and $\hat{\mathbf y}=(\tilde{\mathbf {A}}+\Delta\mathbf A^\perp)\mathbf x$. Furthermore, $|\Delta\mathbf A^\perp|\leq\frac{2\tilde\gamma_n}{1-\tilde\gamma_n}|\tilde{\mathbf A}|$ and $|\Delta \mathbf{D}|\leq\tilde\gamma_n|\mathbf D|$.
\end{restatable}
}

\begin{example}
    We use the counterexample in the proof of Prop.~\ref{proposition:matvec_input_relations_counterexample} to show how a small normal perturbation $\Delta \mathbf{A}^\perp$ maps the computed output $\hat{\mathbf{y}}$ to a valid input database. Consider the query and input database from the proof of  Prop.~\ref{proposition:matvec_input_relations_counterexample}: the natural join of the relations $S_1(X, A) = \{(k, t), (k, 0)\}$ and $S_2(X, B) = \{(k, +1), (k, -1)\}$, for $t=2^{p+2}$ in a floating point number system with base 2 and precision $p$. $\vect A$, $\vect x=(1,1)^\top$, $\mathbf{A}\mathbf{x}$, and $\hat{\mathbf{y}}=\operatorname{fl}(\mathbf{A}\mathbf{x})$ are as before, whereas the new quantities $\tilde{\mathbf{A}}$ and its normal perturbation $\Delta \mathbf{A}^\perp$ are as shown below: 
\[
\mathbf{A}=\begin{pmatrix} t&+1\\ t&-1\\ 0&+1\\ 0&-1 \end{pmatrix}, \quad
\mathbf{A}\mathbf{x}=\begin{pmatrix} t+1\\ t-1\\ +1\\ -1 \end{pmatrix}, \quad
\hat{\mathbf{y}}=\begin{pmatrix} t\\ t\\ +1\\ -1 \end{pmatrix}, \quad 
\tilde{\mathbf{A}}=\begin{pmatrix} t&+1\\ t&-1\\ 0&+1\\ 0&-1 \end{pmatrix},
\quad 
\Delta \mathbf{A}^\perp=\begin{pmatrix} -1&0\\ +1&0\\ 0&0\\ 0&0 \end{pmatrix}\!.
\]
In this example, $\mathbf{A}=\tilde{\mathbf{A}}$: The entries $t+1$ and $t-1$ in $\tilde{\mathbf{A}}+\Delta\mathbf{A}^\perp$ originate from the same input data value $t$ and are replaced by their average $t$ in $\tilde{\mathbf{A}}$. This defines the perturbation $\Delta\mathbf{A}^\perp$ as shown above. This replacement gives the desired projection onto the space of valid join matrices (Lemma~\ref{lemma:projection_is_taking_mean} in Appendix~\ref{app:projected-bs}).
It can be checked that $(\tilde{\mathbf{A}}+\Delta\mathbf{A}^\perp)\mathbf{x}=\hat{\mathbf{y}}$.
Moreover, the entry-wise perturbations of $\Delta\mathbf{A}^\perp$ with respect to $\tilde{\mathbf{A}}$ at indices $(1,1)$ and $(2,1)$ are $1/t$=$2^{-p-2}$, well below the unit round off $2^{-p}$.
Since $\mathbf{A}=\tilde{\mathbf{A}}$, $\Delta\mathbf{D}=0$ in this example. 
\end{example}

\section{The Database Condition Number}
\label{sec:join-condition-number}

The perturbation of the database that explains the error in the output of a numerical operation over the join matrix is a function of two fundamental phenomena: (1) the inherent numerical error due to the numerical operation; and (2) the amplification of this error due to the join that may produce many copies of the same input data value. In this section, we capture the second by the new notion of \emph{database condition number} and then use it to bound the projected backward error.
The amplification of an error is controlled by how many times an entry from an input matrix is copied in the join matrix. The database condition number is a property of the database {\em and} of the join query. We show it to be a natural analogue to the classical condition number of matrices and exhibit the connection between the two.

\nop{In this section, we only consider fully reduced input databases for clarity and simplicity: why?? However, the perturbation bound translation from the join matrix to the input database works even if the database is not calibrated, as further illustrated in Lemma~\ref{lemma:condition_number_translation_with_dangling_tuples}.}

We first introduce a simple notation. For a matrix $\vect S_i$, we fix an order on its rows so we can refer uniquely to the $j$-th row in $\vect S_i$. We denote by $\phi_i(j)$ the number of occurrences of the $j$-th row of $\vect S_i$ in the join matrix. This count holds for all entries in that row.

\begin{definition}[Relation/Database Condition Number]
Consider the setting in Def.~\ref{def:eval-setting}. For any matrix $\vect S_i$ with $d_i$ rows in the database $\vect D = (\vect S_1,\ldots,\vect S_r)$, the \emph{relation condition number} $ \kappa_{\mathrm{join}}(\vect S_i)$ and the  \emph{database condition number} $ \kappa_{\mathrm{join}}(\vect D)$ {\em with respect to the join query} are:
\begin{align*}
 \kappa_{\mathrm{join}}(\vect S_i) = 
 \sqrt{\frac{\max_{j\in[d_i]}\phi_i(j)} {\min_{j\in[d_i]:\,\phi_i(j)>0}\phi_i(j)}} 
\hspace*{3em} 
\kappa_{\mathrm{join}}(\vect D) =
\sqrt{
\frac{
\max_{\vect S_i\in\vect D}\max_{j\in[d_i]}\phi_i(j)}{\min_{\vect S_i\in\vect D}\min_{j\in[d_i]:\,\phi_i(j)>0}\phi_i(j)}}.
\end{align*}
The above holds in case the join matrix is not empty. If the join matrix is empty, then the relation/database condition number is 1.
\end{definition}
The relation (database) condition number is thus the square root of the ratio of the maximum and the minimum numbers of occurrences of a row from a relation (the database) in the query result. It ranges from 1 to the square root of the size of the query result. 

The relation and database condition numbers are equivalent to the classical condition number of so-called expansion matrices that transform the input matrices represented by the input relations and the entire database into the join matrix. To show this connection, we first introduce the expansion matrix for a relation and the database.
\begin{definition}[Relation/Database Expansion Matrix]
\label{def:expansion-matrices}
Consider the setting in Def.~\ref{def:eval-setting}. Let $\mathbf D=(\mathbf S_1,\ldots,\mathbf S_r)$, where $\mathbf S_i\in\mathbb R^{d_i\times n_i}$, and write the join matrix as $\mathbf A=[\mathbf A_1\ \cdots\ \mathbf A_r]$, where $\mathbf A_i\in\mathbb R^{m\times n_i}$ consists of the columns contributed by $\mathbf S_i$, with $d=\sum_{i\in[r]} d_i$ and $n=\sum_{i\in[r]} n_i$. For each $i\in[r]$, the \emph{relation expansion matrix} $\mathbf E_i\in\{0,1\}^{m\times d_i}$ is (for $k\in[m]$ and $j\in[d_i]$):
\[ \mathbf E_i[k,j]=\begin{cases} 1, & \text{if the row $k$ in $\mathbf{A}_i$ is a copy of the row $j$ in $\mathbf{S}_i$},\\ 0, & \text{otherwise}. \end{cases} \]
The \emph{database expansion matrix} $\vect E_{\vect D}$ is the block-diagonal matrix $(\mathbf E_1,\ldots,\mathbf E_r)$.
\end{definition}

By construction, $\mathbf A_i=\mathbf E_i\mathbf S_i$ for every $i\in[r]$ and therefore $\mathbf E_{\mathbf D}\mathbf D$ is a block-diagonal representation $(\vect A_1,\ldots,\vect A_r)$ of $\mathbf A$, which we call $\vect A_{\mathrm{block}}$. 
We show that the classical condition number of the expansion matrix of each relation (database) is equal to our newly introduced relation (database) condition number.

\begin{restatable}{proposition}{cnequalsjcn}
\label{proposition:join_to_matrix_condition_number}
Consider the setting in Def.~\ref{def:eval-setting}. If the join matrix is non-empty, the following hold: 
\begin{itemize}
    \item $\kappa(\mathbf E_i)=\kappa_{\mathrm{join}}(\mathbf S_i)$, where $\vect E_i$ is the expansion matrix for input matrix $\vect S_i$, and
    \item $\kappa(\mathbf E_{\mathbf D})=\kappa_{\mathrm{join}}(\mathbf D)$, where $\vect E_{\vect D}$ is the expansion matrix of the database matrix $\vect D$.
\end{itemize} 
\end{restatable}

We next illustrate the connection between the database(relation) condition number and the condition number of the expansion matrix.
\begin{example}
\label{ex:condition-numbers}
Consider the join $Q=S_1(X,Y_1)\bowtie S_2(X,Y_2)$ and the join matrix $\vect A$:
\[
\begin{array}{c}
S_1 \\
\begin{array}{c|c}
X & Y_1\\ \hline
x_1 & a_1\\
x_2 & a_2
\end{array}
\end{array},
\begin{array}{c}
S_2 \\
\begin{array}{c|c}
X & Y_2\\ \hline
x_1 & b_1\\
x_2 & b_2\\
x_2 & b_3\\
x_2 & b_4
\end{array}
\end{array},
\vect E_1=
\begin{bmatrix}
1&0\\
0&1\\
0&1\\
0&1
\end{bmatrix},
\vect S_1 = 
\begin{bmatrix}
a_1\\
a_2
\end{bmatrix},
\vect E_2=\vect I_4,
\vect S_2=
\begin{bmatrix}
b_1\\
b_2\\
b_3\\
b_4
\end{bmatrix},
\vect A=
\begin{bmatrix}
a_1 & b_1\\
a_2 & b_2\\
a_2 & b_3\\
a_2 & b_4
\end{bmatrix}.
\]
It can be checked that: $\mathbf A_1=\mathbf A[*,Y_1]=\mathbf E_1\mathbf S_1$ and $\mathbf A_2=\mathbf A[*,Y_2]=\mathbf E_2\mathbf S_2$. Also, $\vect A_{\text{block}} = \mathbf{E}_{\vect D} \vect D$, where the block-diagonal matrices are: $\vect A_{\text{block}}=(\vect A_1, \vect A_2)$, $\mathbf{E}_{\mathbf{D}}=(\mathbf{E}_1, \mathbf{E}_2)$, and $\vect D = (\vect S_1, \vect S_2)$.  
The two rows of $\vect S_1$ occur once and thrice, respectively, in $\vect A$, whereas every row of $\vect S_2$ occurs once in $\vect A$. Hence $\bigl(\phi_1(1),\phi_1(2)\bigr)=(1,3),
\bigl(\phi_2(1),\ldots,\phi_2(4)\bigr)=(1,1,1,1)$. It follows that $\kappa_{\mathrm{join}}(\mathbf S_1)=\sqrt{3}$, $\kappa_{\mathrm{join}}(\mathbf S_2)=1$, and $\kappa_{\mathrm{join}}(\mathbf D)=\sqrt{3}$.
We next show that $\kappa(\mathbf E_1)=\sqrt{3}=\kappa_{\mathrm{join}}(\vect S_1)$ and $\kappa(\mathbf E_2)=1=\kappa_{\mathrm{join}}(\vect S_2)$, so the relation condition numbers are the condition numbers of the expansion matrices $\vect E_1$ and $\vect E_2$. We first show how to compute the condition number of $\vect E_1$. Using the SVD decomposition $\vect E_1 = \vect U_1\vect \Sigma_1\vect V_1^{\top}$, we have $\vect E_1^\top\vect E_1 = \vect V_1\vect \Sigma_1^2\vect V_1^{\top}$. Since $\mathbf E_1^\top\mathbf E_1$ is the diagonal matrix $\operatorname{diag}(1,3)$ in our example, we have $\vect V_1 = \vect I_2$ and $\vect \Sigma_1^2=\operatorname{diag}(1,3)$. So $3$ ($1$) is the max (min) singular value of $\vect E_1^\top\vect E_1$ and its square root is the max (min) singular value of $\vect E_1$. So $\kappa(\vect E_1) = \sqrt{3}$.
Since $\vect E_2 = \vect I_4$, the singular values of $\vect E_2$ are all $1$, so $\kappa(\vect E_2) = 1$.
Moreover, $\mathbf{E}_{\mathbf{D}}=(\mathbf{E}_1, \mathbf{E}_2)$ and $\mathbf{E}_{\mathbf{D}}^\top \mathbf{E}_{\mathbf{D}}=\operatorname{diag}(1,3,1,1,1,1)$. Therefore, $\kappa(\mathbf{E}_{\mathbf{D}})=\sqrt{3}=\kappa_{\mathrm{join}}(\mathbf{D})$.
\end{example}

The relation (database) condition number translates the column-wise (Frobenius-norm-wise) perturbation bound from the join space to the input database, as shown next.

\begin{restatable}{proposition}{condboundtranslationdangling}
\label{proposition:condition_number_translation_with_dangling_tuples}
Consider our setting in Def.~\ref{def:eval-setting} and let $\Delta\vect A$ be a perturbation of the join matrix $\vect A$ on the join space. Then, there exists a perturbation $\Delta\mathbf{D}=(\Delta\vect S_1,\dots,\Delta\vect S_r)$ of the input $\mathbf{D}$ such that $\vect A+\Delta\vect A$ is the join matrix of the perturbed input  $\mathbf{D}+\Delta\mathbf{D}$.
Moreover, for every $\gamma\geq 0$, the following statements hold:
\begin{itemize}
    \item For every $j\in[r]$ and $Y\in\mathbf{Y}_j$, if $\norm{\Delta\vect A[*,Y]}_2\leq\gamma\cdot\norm{\vect A[*,Y]}_2$  then \\$\norm{\Delta\vect S_j[*,Y]}_2\leq\kappa_{\mathrm{join}}(\mathbf{S}_j)\cdot\gamma\cdot\norm{\vect S_j[*,Y]}_2$.

    \item If $\norm{\Delta\vect A}_F\leq\gamma\cdot\norm{\vect A}_F$, then $\norm{\Delta\mathbf{D}}_F\leq\kappa_{\mathrm{join}}(\mathbf{D})\cdot\gamma\cdot\norm{\mathbf{D}}_F$.
\end{itemize}
\end{restatable}

The database condition number can be seen as a measure of the data skew in the query result (and the join matrix). A value of one means  a well-conditioned database, where the input entries that contribute to the output are represented uniformly, and no amplification of relative perturbations.

\section{Relationship of Projected Backward Stability and Backward Stability}
\label{sec:pbs-bs-connection}

There is a strong relationship between the projected backward stability of the computation $C$ of $g\circ f$ with respect to the input database, where $f$ is the join query and $g$ is the numerical computation, and the backward stability of the computation $C_g$ of $g$ with respect to the join matrix. In particular, the projected backward stability of $C$ implies the backward stability of $C_g$ (Theorem~\ref{theorem:pbs_implies_bs}), while the backward stability of $C_g$ implies the projected backward stability of $C$ (Theorem~\ref{theorem:bs_implies_pbs}).

\nop{
In this section, we study computations that first materialize the join from the input database and then perform matrix operations on the resulting join matrix. Let $f$ denote the join query, $g$ the matrix operation, and $C=(C_f,C_g)$ the corresponding computation. We establish a connection between the backward stability of $C_g$ and the projected backward stability of the composed computation $C$. The key quantities governing this connection are the relation and database condition numbers, which measure how perturbations of the join matrix can be attributed to perturbations of the underlying input relations. We establish both directions of this relationship: under suitable conditions, backward stability of $C_g$ yields projected backward stability of $C$ (Theorem~\ref{theorem:bs_implies_pbs}), while projected backward stability of $C$ in turn yields backward stability of $C_g$ (Theorem~\ref{theorem:pbs_implies_bs}}

\begin{restatable}{theorem}
{pbsimpliesbs}
\label{theorem:pbs_implies_bs}
Consider the setting in Def.~\ref{def:eval-setting} and let $C$ be the computation of the join query followed by the computation $C_g$ of a numerical operation over the join matrix $\vect A$. For $\alpha,\beta\geq0$, the following hold:
\begin{itemize}
    \item If $C$ is entry-wise $(\alpha,\beta)$-projected backward stable, then $C_g$ is entry-wise $\bigl(\alpha+\beta\cdot(1+\alpha)\bigr)$-backward stable on the subspace $\mathcal{Y}_0$;
    \item If $C$ is column-wise $(\alpha,\beta)$-projected backward stable, then $C_g$ is column-wise\\ $\sqrt{\bigl(\kappa_{\mathrm{join}}(\mathbf{S}_i)\cdot\alpha\bigr)^2+\beta^2\cdot\bigl(1+\kappa_{\mathrm{join}}(\mathbf S_i)\cdot\alpha\bigr)^2}$-backward stable on the subspace $\mathcal{Y}_0$  for each input data column  in $\vect S_i$ and $i\in[r]$;
    \item If $C$ is Frobenius-norm-wise $(\alpha,\beta)$-projected backward stable, then $C_g$ is Frobenius-norm-wise $\sqrt{\bigl(\kappa_{\mathrm{join}}(\mathbf{D})\cdot\alpha\bigr)^2+\beta^2\cdot\bigl(1+\kappa_{\mathrm{join}}(\mathbf D)\cdot\alpha\bigr)^2}$-backward stable on the subspace $\mathcal{Y}_0$.
\end{itemize}
\end{restatable}

\begin{figure}[tb]
  \centering
  \begin{tikzpicture}[font=\small, line cap=round, line join=round]
    \coordinate (A) at (0,2);
    \coordinate (T) at (0,0);
    \coordinate (B) at (1.6,0);
    \coordinate (R) at (7,1.5);

    \draw[-,semithick, dotted](A) -- node[midway,left=3pt] {$\|\mathbf{A}-\mathbf{\tilde{\mathbf{A}}}\|_F\leq\rho\cdot\|\mathbf A\|_F$} (T);
    \draw[-,semithick, dotted] (T) -- node[midway,below=7pt] {} (B);
    \draw[-,semithick, dotted] (A) -- node[midway,right=3pt] {$\|\tilde{\mathbf A}+\Delta\mathbf A^\perp-\mathbf A\|_F$} (B);
    \draw[-{Stealth}, semithick, dashed](A) -- node[midway, above=1pt] {$C_g$} (R);
    \draw[-{Stealth}, semithick] (B) -- node[midway, below=1pt] {$g$} (R);

    \fill (A) circle (1.2pt) node[left=1pt] {$\mathbf A$};
    \fill (T) circle (1.2pt) node[left=6pt] {$\tilde{\mathbf A}$};
    \fill (B) circle (1.2pt) node[below=3pt, right=5pt] {$\tilde{\mathbf A}+\Delta\mathbf A^\perp\quad\|\Delta\mathbf{A}^\perp\|_F\leq\beta\cdot\|\tilde{\mathbf{A}}\|_F\leq\beta\cdot(1+\rho)\cdot\|\mathbf A\|_F$ };
    \fill (R) circle (1.2pt)
    node[below, xshift=40pt, yshift=-3pt]{$\hat{z}=C_g(\mathbf{A})=g(\tilde{\mathbf{A}}+\Delta\mathbf{A}^\perp)$};

  \end{tikzpicture}
\caption{Geometric interpretation of the Frobenius-norm-wise backward-error bound of $C_g$ in Theorem~\ref{theorem:pbs_implies_bs}. We let $\rho=\kappa_\mathrm{join}(\mathbf D)\cdot\alpha$. The quantities $\|\mathbf{A}-\tilde{\mathbf{A}}\|_F$ and $\|\Delta\mathbf{A}^\perp\|_F$ are the lengths of the perpendicular legs of a right triangle, whose hypotenuse length is the Frobenius norm of the total perturbation $\mathbf E=\tilde{\mathbf A}+\Delta\mathbf A^\perp-\mathbf A$.}
  \label{fig:pbs-to-bs-geometry}
\end{figure}
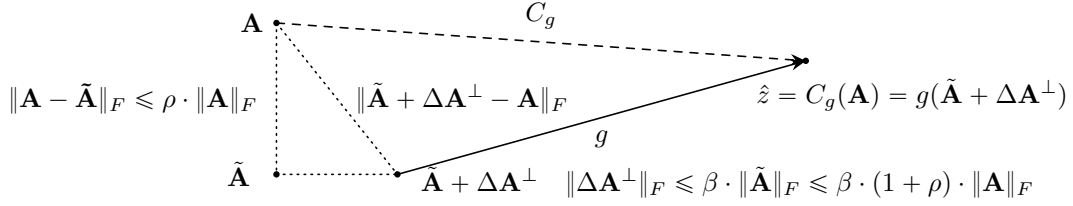

The Frobenius-norm-wise bound in Theorem~\ref{theorem:pbs_implies_bs} admits a geometric explanation (Fig.~\ref{fig:pbs-to-bs-geometry}). Let $\rho=\kappa_\mathrm{join}(\mathbf D)\cdot\alpha$, where $\alpha$ satisfies $\|\Delta\mathbf{D}\|_F\leq\alpha\cdot\|\mathbf{D}\|_F$. Lemma~\ref{lemma:join_condition_number_fnorm} and the normal-error bound $\beta$ give\footnote{We assume here that the input database is fully reduced (Lemma~\ref{lemma:join_condition_number_fnorm} applies to fully reduced databases). Theorem~\ref{theorem:pbs_implies_bs} holds for uncalibrated databases, too (cf.\@ its proof in Appendix~\ref{sec:pbs-bs-connection-proof}).}
$\|\tilde{\mathbf A}-\mathbf A\|_F \leq\rho\cdot\|\mathbf A\|_F$ and $\|\Delta\mathbf A^\perp\|_F \leq\beta\cdot\|\tilde{\mathbf A}\|_F \leq\beta\cdot(1+\rho)\cdot\|\mathbf A\|_F$.
These two perturbations are orthogonal to each other because $\tilde{\mathbf A}-\mathbf A\in\mathcal Y_0$ and $\Delta\mathbf A^\perp\in\mathcal Y_0^\perp$. Their sum $\mathbf E=\tilde{\mathbf A}+\Delta\mathbf A^\perp-\mathbf A$ therefore represents the hypotenuse of the right triangle in Fig.~\ref{fig:pbs-to-bs-geometry}. By the Pythagorean theorem,
$\|\mathbf E\|_F^2=\|\mathbf A-\tilde{\mathbf A}\|_F^2+\|\Delta\mathbf A^\perp\|_F^2\leq\bigl(\rho^2+\beta^2\cdot(1+\rho)^2\bigr)\cdot\|\mathbf A\|_F^2$. Since $C_g(\mathbf A)=g(\mathbf A+\mathbf E)$ and $\|\mathbf E\|_F\leq\sqrt{\rho^2+\beta^2\cdot(1+\rho)^2}\cdot\|\mathbf A\|_F$, the relative backward-error bound of $C_g$ is $\sqrt{\rho^2+\beta^2\cdot(1+\rho)^2}$.
The column-wise error bound follows from the same geometric argument.

\begin{restatable}{theorem}
{bsimpliespbs}
\label{theorem:bs_implies_pbs}
Consider the setting in Def.~\ref{def:eval-setting} and let $C$ be the computation of the join query followed by the computation $C_g$ of a numerical operation over the join matrix $\vect A$.  For $0\leq\epsilon<1$, the following hold:
\begin{itemize}
    \item If $C_g$ is entry-wise $\epsilon$-backward stable, then $C$  is entry-wise $(\epsilon, \frac{2\epsilon}{1-\epsilon})$-projected backward stable.
    \item If $C_g$ is column-wise $\epsilon$-backward stable, then  $C$ is column-wise $(\kappa_{\mathrm{join}}(\vect S_i)\cdot\epsilon, \frac{\epsilon}{\sqrt{1-\epsilon^2}})$-projected backward stable for each input data column in $\vect S_i$ and $i\in[r]$;
    \item If $C_g$ is Frobenius-norm-wise $\epsilon$-backward stable, then $C$ is Frobenius-norm-wise\\ $(\kappa_{\mathrm{join}}(\vect D)\cdot\epsilon, \frac{\epsilon}{\sqrt{1-\epsilon^2}})$-projected backward stable.
\end{itemize}


\end{restatable}

The condition $\epsilon<1$ in Theorem~\ref{theorem:bs_implies_pbs} is necessary. We explain this next for the Frobenius-norm-wise bound by showing that when $\epsilon\geq1$, there can be no finite bound $\beta$ for the normal perturbation.
Assume $\mathbf A\in\mathcal Y_0\setminus\{\mathbf 0\}$ and $\mathcal Y_0^\perp\neq\{\mathbf 0\}$. Further, assume $C_g$ is $\epsilon$-backward stable for $\epsilon>1$. Consider the scenario where $\tilde{\mathbf A}=\mathbf 0$ and $\Delta\mathbf A^\perp\in\mathcal Y_0^\perp$ with $\|\Delta\mathbf A^\perp\|_F=\sqrt{\epsilon^2-1}\cdot\|\mathbf A\|_F$. By the Pythagorean Theorem, $\|\Delta\mathbf A^\perp-\mathbf A\|_F^2=\|\mathbf A\|_F^2+\|\Delta\mathbf A^\perp\|_F^2=\epsilon^2\cdot\|\mathbf A\|_F^2$ (the geometry is the same as shown in Figure~\ref{fig:pbs-to-bs-geometry}, with the extra condition that $\tilde{\mathbf{A}}=\mathbf{0}$).
Denote $\mathbf E=\tilde{\mathbf A}+\Delta\mathbf A^\perp-\mathbf A$ and we have $C_g(\mathbf{A})=g(\mathbf{A}+\mathbf{E})$, with $\|\mathbf{E}\|_F=\epsilon\cdot\|\mathbf{A}\|_F$, satisfying the backward error bound of $\epsilon$ for $C_g$. But the projected backward stability of $C$ requires a finite $\beta$ satisfying $\|\Delta\mathbf A^\perp\|_F\leq\beta\cdot\|\tilde{\mathbf A}\|_F$, which cannot be fulfilled in our scenario, since the left-hand side is positive and the right-hand side is zero.
If $\epsilon=1$, consider the scenario where $\tilde{\mathbf A}=t\mathbf A$ and $\Delta\mathbf A^\perp\in\mathcal Y_0^\perp$ with $\|\Delta\mathbf A^\perp\|_F=\sqrt{2t-t^2}\,\|\mathbf A\|_F$, where $0<t<1$. By the Pythagorean Theorem, $\|\tilde{\mathbf A}+\Delta\mathbf A^\perp-\mathbf A\|_F^2=\|\Delta\mathbf{A}^\perp\|_F^2+\|\mathbf{A}-\tilde{\mathbf{A}}\|_F^2=\bigl((1-t)^2+2t-t^2\bigr)\cdot\|\mathbf A\|_F^2=\|\mathbf A\|_F^2$. Denote $\mathbf E=\tilde{\mathbf A}+\Delta\mathbf A^\perp-\mathbf A$ and we have $C_g(\mathbf{A})=g(\mathbf{A}+\mathbf{E})$, with $\|\mathbf{E}\|_F=\|\mathbf{A}\|_F$, satisfying the backward error bound of 1 for $C_g$. But $\frac{\|\Delta\mathbf A^\perp\|_F}{\|\tilde{\mathbf A}\|_F}=\sqrt{\frac{2}{t}-1}\to\infty$ as $t\to0$. Thus, there is no finite bound $\beta$ for the normal perturbation for $\epsilon=1$.
Therefore, the condition $\epsilon<1$ is necessary in Theorem~\ref{theorem:bs_implies_pbs}.

Prop.~\ref{proposition:matvec_proj_b_s} is a direct corollary of Theorem~\ref{theorem:bs_implies_pbs}. As the matrix vector product $\mathbf{A}\vect x$ is entry-wise $\tilde{\gamma}_n$-backward stable with respect to $\mathbf{A}$~\cite[Sec.~3.5, Eq.~(3.11)]{Higham2002}, Theorem~\ref{theorem:bs_implies_pbs} states that the entire computation of the join matrix $\mathbf{A}$ followed by $\mathbf{A}\mathbf{x}$ is entry-wise $(\tilde{\gamma}_n, \frac{2\tilde{\gamma}_n}{1-\tilde{\gamma}_n})$-projected backward stable with respect to the input database.


\section{Numerical Stability of QR and SVD}
\label{sec:qr-svd}

We now investigate the stability of the QR and SVD decompositions. In the standard setting, so for  arbitrary input matrices, there are backward stable algorithms for these decompositions. In our database setting with structurally restricted join matrices, backward stability fails while projected backward stability holds. We first discuss the QR and then the SVD.

Assume we want to compute the upper triangular matrix $\hat{\vect{R}}$ from the QR decomposition of a matrix $\vect{A} \in \mathbb{R}^{m \times n}$ using Givens Rotations~\cite{GivensRotations:1958}. Then, there exists an orthogonal matrix $\vect{Q}'$ and a perturbation matrix $\Delta \vect{A}$, such that $\vect{A} + \Delta \vect{A} = \vect{Q}' \hat{\vect{R}}$, where $\Delta \vect{A}$ satisfies the column-wise bounds $\|\Delta \vect{A}[*,j]\|_2 \le \tilde{\gamma}_{m+n-2} \, \|\vect{A}[*,j]\|_2$, for $j \in [n]$~\cite{Higham2002}.

Backward stability with respect to the join matrix does not generally imply backward stability with respect to the input database. Prop.~\ref{proposition:qr_input_relations_counterexample} in Appendix~\ref{sec:qr-svd-proofs} establishes this failure for a common algorithm for QR decomposition. Yet Theorem~\ref{theorem:bs_implies_pbs} together with the backward stability result for QR decomposition~\cite{Higham2002} implies projected backward stability.

\begin{restatable}{corollary}{qrnaiveprojected}[Theorem~\ref{theorem:bs_implies_pbs} and \cite{Higham2002}]
\label{proposition:qr_naive_projected}
Consider the setting in Def.~\ref{def:eval-setting} and let $C$ be the computation that materializes the join matrix $\mathbf A$ and then applies Givens Rotations to compute the upper-triangular matrix $\mathbf{R}$ from the QR decomposition of $\vect A$. If $\tilde{\gamma}_{m+n-2}<1$, then $C$ is column-wise $(\kappa_\mathrm{join}(\mathbf{S}_j)\cdot\tilde{\gamma}_{m+n-2}, \frac{\tilde{\gamma}_{m+n-2}}{\sqrt{1-\tilde{\gamma}_{m+n-2}^2}})$-projected backward stable  for each input data column in $\vect S_j$ and $j\in[r]$. 
\end{restatable}

A similar guarantee holds for computing the singular values.
LAPACK's \texttt{xGESVD} returns the exact singular values of
$\mathbf A+\Delta \mathbf A$, where
$\|\Delta \mathbf A\|_2\le p(m,n)\cdot u\cdot\|\mathbf A\|_2$, $p(m,n)$ is a
modestly growing dimension-dependent
factor, and $u$ is the unit round-off~\cite[Sec.~4.9]{LAPACK1999}. 
Since $\|\Delta \mathbf A\|_F\le\sqrt{\min(m,n)}\cdot\|\Delta \mathbf A\|_2$
and $\|\mathbf A\|_2\le\|\mathbf A\|_F$, \texttt{xGESVD} is
Frobenius-norm-wise $\epsilon_{\mathrm{svd}}(m,n)$-backward stable,
where $\epsilon_{\mathrm{svd}}(m,n)
  \defeq \sqrt{\min(m,n)}\cdot p(m,n)\cdot u$. Combining this classical result with 
Theorem~\ref{theorem:bs_implies_pbs}, we obtain the following
corollary.

\begin{restatable}{corollary}{svdnaiveprojected}[Theorem~\ref{theorem:bs_implies_pbs} and \cite[Sec.~4.9]{LAPACK1999}]
\label{proposition:svd_naive_projected}
Consider the setting in Def.~\ref{def:eval-setting} and let $C$ be the computation that materializes the join matrix $\mathbf A$ and then applies \texttt{xGESVD} to compute $\mathbf{\Sigma}$, the diagonal matrix containing singular values of $\vect A$. If $\epsilon_{\mathrm{svd}}<1$, then $C$
is Frobenius-norm-wise $\left(
    \kappa_{\mathrm{join}}(\mathbf D)\cdot\epsilon_{\mathrm{svd}},\,
    \frac{\epsilon_{\mathrm{svd}}}{\sqrt{1-\epsilon_{\mathrm{svd}}^2}}
  \right)$-projected backward stable.
\end{restatable}

Recall that pushing the computation of the Frobenius norm past the join down to the input relations improves the numerical stability by avoiding repeated arithmetic on copied values (Prop.~\ref{proposition:frobenius_norm_input_relations_over_the_join}).
We establish a similar benefit for QR using the \textsc{FiGaRo} algorithm~\cite{FIGARO:VLDBJ:2023}, which pushes the computation of the upper-triangular matrix $\vect R$ past the $\alpha$-acyclic joins and thereby avoids the materialization of the join matrix $\vect A$. In particular, \textsc{FiGaRo} computes an almost upper-triangular matrix $\vect R_0$ such that $\vect A = \vect Q_0 \vect R_0$, where $\vect Q_0$ is some orthonormal matrix defined by a sequence of Givens rotations. The matrix $\vect R_0$ can then be transformed into the upper-triangular matrix $\vect R$ using, e.g., the Givens Rotations algorithm for QR decomposition. Note also that the singular values of $\vect R_0$ are precisely those of $\vect A$.

By pushing the linear operation past the join we obtain a much stronger numerical stability, where the error depends on the number $d$ of rows in the input database instead of the number $m$ of rows in the query result (and the join matrix). Note that $m=\mathcal{O}(d^{\rho^*(Q)})$ for any input database and there are databases for which this is asymptotically tight. Here, $\rho^*(Q)$ is the fractional edge cover number of $Q$~\cite{AGM:2008}. It can be as large as the number $r$ of relations in $Q$. Since $d$ can easily be in the order of hundreds of millions and $r$ in the order of 5-10, the error reduction from $m$ to $d$ brought by \textsc{FiGaRo} can be significant.

\begin{definition} [\textsc{FiGaRo} setting]
\label{def:figaro-setting}
In addition to the database setting from Def.~\ref{def:eval-setting}, in the \textsc{FiGaRo} setting we require calibrated databases, $\alpha$-acyclic queries, and given join trees as input to the \textsc{FiGaRo} algorithm.
\end{definition}

\begin{restatable}{theorem}{QRprojectedstable}
\label{theorem:QR_projected_stable}
Consider the \textsc{FiGaRo} setting in Def.~\ref{def:figaro-setting} and let $C$ be the computation that uses  $\textsc{FiGaRo}$ on the input database followed by Givens Rotations to compute the upper-triangular matrix $\mathbf{R}$ from the QR decomposition of $\vect A$. 
Let $\eta=\tilde{\gamma}_{d}+\tilde{\gamma}_{d+n-2}\cdot\bigl(1+\tilde{\gamma}_{d}\bigr)$. If $\eta<1$, then $C$ is column-wise $\Big(\kappa_\mathrm{join}(\mathbf{S}_j)\cdot\eta, \frac{\eta}{\sqrt{1-\eta^2}}\Big)$-projected backward stable for each input data column in $\vect S_j$ and $j\in[r]$. 
\end{restatable}

\textsc{FiGaRo} can improve the numerical stability of computing the singular values, too.

\begin{restatable}{theorem}{SVDprojectedstable}
\label{theorem:SVD_projected_stable}
Consider the \textsc{FiGaRo} setting in Def.~\ref{def:figaro-setting} and  let $C$ be the computation that uses  $\textsc{FiGaRo}$ on the input database followed by \texttt{xGESVD} to compute the singular values of $\mathbf{A}$.
Let $\eta=\tilde{\gamma}_{d}+\epsilon_\mathrm{svd}(d,n)\cdot\bigl(1+\tilde{\gamma}_{d}\bigr)$. If $\eta<1$, then $C$ is Frobenius-norm-wise $\Big(\kappa_\mathrm{join}(\mathbf{D})\cdot\eta, \frac{\eta}{\sqrt{1-\eta^2}}\Big)$-projected backward stable. 
\end{restatable}

\section{Conclusion}
\label{sec:conclusion}

\nop{
shows that backward stability of the linear algebra computation with respect to the materialized join matrix may not translate to backward stability of the combined relational-linear algebra computations with respect to the input database, because floating-point perturbations can violate the join dependencies in the join matrix. To address this structural mismatch, it introduces the new notion of projected backward stability and establishes its connection to the classical notion of backward stability. The paper also introduces the new concept of database condition number that captures the amplification factor due to the join query of the perturbation produced by the linear algebra computation over the join matrix. The database condition number is specific to the problem (database and join query) and can be related to the classical condition number of the matrices that transform the input database into the join matrix. Finally, by pushing the linear algebra computation past the join, such as the QR decomposition and Singular Value Decomposition, it gives bounds on the numerical error that depend on the database condition number and the input database size as opposed to the size of the join matrix.
}

This paper initiates the study of numerical stability for computations that combine relational algebra and linear algebra. It introduces the fundamental concepts of projected backward stability and database condition number to allow making formal statements about the stability of numerical operations over matrices defined by database join queries.

The framework introduced in this paper is a fertile ground for future work. We mention below a few directions:
\begin{itemize}
    \item Analyze the numerical stability in the context of more complex queries, e.g., different join types and even self-joins. 
    Every column-wise error bound in the formal statements of this paper fails in the presence of self-joins. 
    The presence of self-joins does not make a difference, however, for our entry-wise and Frobenius-norm-wise error bounds.   
    \item Determine tight or optimal bounds for projected backward stability. Our current analysis establishes upper bounds on the database and orthogonal perturbations, but does not show whether these bounds are tight or characterize the best achievable pair of bounds.
    \item Extend the framework to analytical pipelines, e.g., series-parallel graphs whose nodes are relational and linear algebra operators. How do the structural and numerical perturbations compose across such pipelines? Develop a syntax-directed, compositional method that derives global error bounds from the operators in the pipeline, while avoiding a numerical analysis for every combination of relational and linear-algebraic operations.
    \item Our work examines numerical stability with respect to perturbations of data values. An alternative view is that it investigates the sensitivity of (relational and linear) algebraic data transformations: How much the input data must change to account for a change in the output of a given scale. This perspective can inform what-if analysis of computational pipelines: understanding whether a desired change in the output can be realized through an input perturbation, and how large this perturbation must be. This view relates to prior work~\cite{WalenzY16} that studies how varying query parameters on a fixed database affects query results, with applications to computational fact-checking.     
\end{itemize}

\bibliography{bibtex}

\newpage

\appendix

This appendix complements the main paper with further preliminaries and proofs of the formal statements.

\section{Further Preliminaries}
\label{app:preliminaries}

We review the notions of normed and Hilbert spaces used in our
definition of projected backward stability (Definition~\ref{def:proj_b_s}). A \emph{normed space} is a
vector space $\mathcal{X}$ equipped with a norm
$\|\cdot\|:\mathcal{X}\to\mathbb{R}_{\geq 0}$ satisfying, for all
$x,y\in\mathcal{X}$ and $\alpha\in\mathbb{R}$,
\[
\|x\|=0 \iff x=0,
\qquad
\|\alpha x\|=|\alpha|\,\|x\|,
\qquad
\|x+y\|\leq \|x\|+\|y\|.
\]
The norm provides a notion of magnitude, and hence allows us to quantify
the size of perturbations such as $\Delta x$ through quantities such as
$\|\Delta x\|/\|x\|$.

A \emph{finite-dimensional Hilbert space} is a finite-dimensional
vector space $\mathcal{Y}$ equipped with an inner product
$\langle\cdot,\cdot\rangle$. The inner product induces a norm
\[
\|y\| = \sqrt{\langle y,y\rangle}.
\]
The inner product also defines orthogonality: two vectors
$y_1,y_2\in\mathcal{Y}$ are orthogonal if
$\langle y_1,y_2\rangle=0$. For a subspace
$\mathcal{Y}_0\subseteq\mathcal{Y}$, its \emph{orthogonal complement} is
\[
\mathcal{Y}_0^\perp
=
\{y\in\mathcal{Y}:
\langle y,y_0\rangle=0
\text{ for all }y_0\in\mathcal{Y}_0\}.
\]
Because $\mathcal{Y}$ is finite-dimensional, every subspace
$\mathcal{Y}_0$ is closed and every $y\in\mathcal{Y}$ has a unique
orthogonal decomposition
\[
y = y_0 + y_0^\perp,
\qquad
y_0\in\mathcal{Y}_0,\quad
y_0^\perp\in\mathcal{Y}_0^\perp.
\]
Equivalently, there is a unique orthogonal projection
$P_{\mathcal{Y}_0}:\mathcal{Y}\to\mathcal{Y}_0$ such that
\[
P_{\mathcal{Y}_0}y=y_0,
\qquad
(I-P_{\mathcal{Y}_0})y=y_0^\perp.
\]
Thus, in our setting, the finite-dimensional Hilbert space structure
provides both a well-defined notion of magnitude through the norm and a
well-defined notion of orthogonal projection, which we use to separate
the component in $\mathcal{Y}_0$ from the off-range component in
$\mathcal{Y}_0^\perp$.
\section{Proofs of Statements from the Main Body of the Paper}
\label{sec:app-theorem_proofs}

\subsection{Proofs of Statements from Section~\ref{sec:bs}}
Here we provide the proofs of the formal statements in Section~\ref{sec:bs}.
\frobeniusnorminputrelationsoverthejoin*

\begin{proof}

Let $z=f(\mathbf D)=\|\mathbf A\|_F$ and $\hat z=C(\mathbf D)$. Compute the multiplicity $\phi_k(i)$ of each row $i$ of $S_k$
in the join result. Then
\[
z=\sqrt{\sum_{k=1}^r\sum_i\sum_j
\phi_k(i)\mathbf S_k[i,j]^2}.
\]

Once the multiplicities have been computed exactly, evaluating this
expression requires $O(\ell)$ arithmetic operations. Since all
terms are nonnegative and the square root is forward stable,
the standard rounding model gives $|\hat z-z| \leq \tilde\gamma_{\ell} z$.

If $z=0$, then $\hat z = 0$, and we set $\Delta \vect D = \vect 0$. Otherwise, we write $\hat z-z = \theta z$, where $|\theta|\leq \tilde\gamma_{\ell}$. Since $\hat z\geq 0$, we have $1+\theta = \hat z/z\geq 0$. Define the database perturbation $\Delta \vect D = \theta\vect D$. 

This scales every numerical data entry by $\theta$ to form its
perturbation, while preserving the join keys, row counts, and row
order. Thus $\mathbf D+\Delta\mathbf D=(1+\theta)\mathbf D$
on the numerical data entries. The join structure and all
multiplicities $\phi_k(i)$ remain unchanged. Consequently,
\[
f(\mathbf D+\Delta\mathbf D)
=\sqrt{\sum_{k=1}^r\sum_i\sum_j
\phi_k(i)(1+\theta)^2\mathbf S_k[i,j]^2}=|1+\theta|z
 =(1+\theta)z
 =\hat z.
\]

Finally, entry-wise, $|\Delta\mathbf D|
=|\theta|\,|\mathbf D|
\leq\tilde\gamma_\ell\cdot|\mathbf D|$. This proves backward stability with respect to the input database.
\end{proof}

\subsection{Proofs of Statements from Section~\ref{sec:projected-bs}}
\label{app:projected-bs}

Here we provide the proofs of the formal statements in Section~\ref{sec:projected-bs}.

\begin{lemma}
\label{lemma:join_consistent_subspace}
Consider our setting in Definition~\ref{def:eval-setting} and let $\mathbf{A}$ be the query matrix, and $\mathcal Y_0\subseteq\mathbb R^{m\times n}$ be the query space. Then $\mathcal Y_0$ is a linear subspace of $\mathbb R^{m\times n}$.
\end{lemma}

\begin{proof}
With the row labels, key columns, and join structure fixed, every entry of a join matrix is a copy of a fixed numerical entry of one of the input relations in the database. Hence the map from the numerical entries of $S_1,\ldots,S_r$ to the materialized join matrix is linear. Therefore, if $\mathbf A,\mathbf B\in\mathcal Y_0$ and $\alpha,\beta\in\mathbb R$, then $\alpha\mathbf A+\beta\mathbf B$ is the join matrix obtained by applying the same linear combination to the corresponding numerical entries of the input database. Thus $\alpha\mathbf A+\beta\mathbf B\in\mathcal Y_0$, and $\mathcal Y_0$ is a linear subspace.
\end{proof}

\begin{lemma}
\label{lemma:projection_is_taking_mean}
Consider our setting in Definition~\ref{def:eval-setting} and let $\mathbf{D}$ be an input database and $\mathbf{A}$ be the query matrix,  $\mathcal{Y}_0$ be the join space, and $\mathcal{Y}=\mathbb{R}^{m\times n}$. For any $\tilde{\mathbf{A}}+\Delta\mathbf{A}^\perp\in\mathcal{Y}$ with $\tilde{\mathbf{A}}\in\mathcal{Y}_0$ and $\Delta \mathbf{A}^\perp\in\mathcal{Y}_0^\perp$, its projection on $\mathcal{Y}_0$, $\tilde{\mathbf{A}}$, satisfies the following property: For any fixed entry $s$ in a relation in the input database $\mathbf{D}$ that appears in the join, denote $\mathcal{P}_s=\{(i,j)|\mathbf{A}[i,j] \text{ is a copy of } s\}$ as the set of indices of entries in $\mathbf{A}$ parameterized by the same input entry $s$, then, we have: for any $(i',j')\in \mathcal{P}_s$, $\tilde{\mathbf{A}}[i',j'
]=\frac{\sum_{(i,j)\in \mathcal{P}_s}(\tilde{\mathbf{A}}+\Delta\mathbf{A}^\perp)[i,j]}{|\mathcal{P}_s|}$
\end{lemma}

\begin{proof}
Fix an entry $s$ in a relation in the input database $\mathbf{D}$, and define the matrix $\mathbf{M}_s\in\mathbb{R}^{m\times n}$ by
\[
\mathbf{M}_s[i,j]=\begin{cases}1,&(i,j)\in \mathcal{P}_s,\\0,&(i,j)\notin \mathcal{P}_s.\end{cases}
\]
Because the join keys are fixed and the data entry $s$ is a free parameter, perturbing $s$ by any scalar changes precisely the entries of the query matrix indexed by $\mathcal{P}_s$ by that scalar. Consequently, $\mathbf{M}_s\in\mathcal{Y}_0$.

Moreover, since all entries indexed by $\mathcal{P}_s$ are parameterized by the same input entry $s$, every matrix in $\mathcal{Y}_0$ has the same value at all indices in $\mathcal{P}_s$. In particular, there exists a scalar $c_s$ such that $\tilde{\mathbf{A}}[i,j]=c_s$ for every $(i,j)\in \mathcal{P}_s$.

Since $\Delta\mathbf{A}^{\perp}\in\mathcal{Y}_0^\perp$ and $\mathbf{M}_s\in\mathcal{Y}_0$, $\Delta\mathbf{A}^\perp$ is orthogonal to $\mathbf{M}_s$ and their inner product is 0, so we have
\[
0=\left\langle\Delta\mathbf{A}^{\perp},\mathbf{M}_s\right\rangle_F=\sum_{(i,j)\in \mathcal{P}_s}\Delta\mathbf{A}^{\perp}[i,j].
\]
It follows that
\[
0=\sum_{(i,j)\in \mathcal{P}_s}\left((\Delta\mathbf{A}^\perp+\tilde{\mathbf{A}})[i,j]-\tilde{\mathbf{A}}[i,j]\right)=\sum_{(i,j)\in \mathcal{P}_s}(\Delta\mathbf{A}^\perp+\tilde{\mathbf{A}})[i,j]-|\mathcal{P}_s|c_s.
\]
Therefore,
\[
c_s=\frac{\sum_{(i,j)\in \mathcal{P}_s}(\tilde{\mathbf{A}}+\Delta\mathbf{A}^{\perp})[i,j]}{|\mathcal{P}_s|}.
\]
Thus, for every $(i',j')\in \mathcal{P}_s$,
\[
\tilde{\mathbf{A}}[i',j']=\frac{\sum_{(i,j)\in \mathcal{P}_s}(\tilde{\mathbf{A}}+\Delta\mathbf{A}^{\perp})[i,j]}{|\mathcal{P}_s|}.
\]
Since $s$ was arbitrary, the claimed property holds for every input entry $s$.
\end{proof}

\Axisprojectedbackwardstable*

\begin{proof}
For each $i\in[m]$, the computed
entry $\hat{\mathbf y}_i$ is obtained by evaluating the length-$n$ dot product $\mathbf A_{i,:}\mathbf x=\sum_{j=1}^n\mathbf A_{ij}\mathbf x_j$
in floating-point arithmetic. Under the standard model
\[
\operatorname{fl}(a\mathbin{\mathrm{op}}b)
=(a\mathbin{\mathrm{op}}b)(1+\delta),
\qquad
|\delta|\leq u,
\]
and assuming that no overflow or exceptional underflow occurs, the standard dot-product error analysis shows that there exist numbers $\theta_{ij}$ satisfying $|\theta_{ij}|\leq\tilde\gamma_n$
such that
\[
\hat{\mathbf y}_i
=
\sum_{j=1}^n
\mathbf A_{ij}(1+\theta_{ij})\mathbf x_j.
\]
Define the perturbation matrix $\mathbf E$ by $\mathbf E_{ij}:=\theta_{ij}\mathbf A_{ij}.$
Then $\hat{\mathbf y} =
(\mathbf A+\mathbf E)\mathbf x,$
and $|\mathbf E|
\leq
\tilde\gamma_n|\mathbf A|.$
This means that the computation of the matrix vector product $\mathbf{A}\mathbf{x}$ is $\tilde{\gamma}_{n}$-backward stable with respect to $\mathbf{A}$. Using Theorem~\ref{theorem:bs_implies_pbs}, we immediately obtain: the computation for $g_{\mathbf{x}}(f(\mathbf{D}))$ is $(\tilde{\gamma}_n, \frac{2\tilde{\gamma}_n}{1-\tilde{\gamma}_n})$-projected backward stable. This completes the proof.
\end{proof}

\subsection{Proofs of Statements from Section~\ref{sec:join-condition-number}}
\label{sec:join-condition-number-proofs}
Here we provide the proofs of the formal statements in Section~\ref{sec:join-condition-number}.

\cnequalsjcn*
\begin{proof}
Consider a relation expansion matrix $\mathbf E_i$. Its $j$-th column contains exactly $\phi_i(j)$ ones, so the inner product of this column with itself is $\phi_i(j)$. Two distinct columns have inner product zero: each join row uses exactly one tuple from $S_i$, so each row of $\mathbf E_i$ contains exactly one entry equal to $1$. Consequently, $\mathbf E_i^\top\mathbf E_i=\operatorname{diag}\bigl(\phi_i(1),\ldots,\phi_i(d_i)\bigr)$.
Let $\mathbf E_i=\mathbf U_i\mathbf \Sigma_i\mathbf V_i^\top$ be a full SVD. Since $\mathbf U_i^\top\mathbf U_i=\mathbf I_m$, we obtain
\[ \mathbf E_i^\top\mathbf E_i=\mathbf V_i(\mathbf \Sigma_i^\top\mathbf \Sigma_i)\mathbf V_i^\top. \]
The matrix $\mathbf\Sigma_i^\top\mathbf\Sigma_i$ is diagonal and nonnegative, with positive diagonal entries equal to the squares of the positive singular values of $\mathbf E_i$. Since $\mathbf V_i$ is orthogonal, this is an SVD of $\mathbf E_i^\top\mathbf E_i$. On the other hand, $\mathbf E_i^\top\mathbf E_i$ is itself diagonal and nonnegative, so its singular values are its diagonal entries, up to ordering. Comparing these two descriptions, the positive singular values of $\mathbf E_i$ are precisely $\sqrt{\phi_i(j)}$ for $j\in[d_i]$ with $\phi_i(j)>0$. Dangling tuples contribute zero diagonal entries, which do not enter the positive-singular-value ratio.

Since the join is nonempty, at least one tuple of $S_i$ has positive multiplicity. Thus,
\[ \sigma_{\max}(\mathbf E_i)=\sqrt{\max_{j\in[d_i]}\phi_i(j)},\qquad \sigma_{\min}^{+}(\mathbf E_i)=\sqrt{\min_{j\in[d_i]:\,\phi_i(j)>0}\phi_i(j)}. \]
Here $\sigma_{\min}^{+}(\mathbf E_i)$ denotes the smallest positive singular value of $\mathbf{E}_i$. Taking their ratio gives
\[ \kappa(\mathbf E_i)=\frac{\sigma_{\max}(\mathbf E_i)}{\sigma_{\min}^{+}(\mathbf E_i)}=\sqrt{\frac{\max_{j\in[d_i]}\phi_i(j)}{\min_{j\in[d_i]:\,\phi_i(j)>0}\phi_i(j)}}=\kappa_{\mathrm{join}}(\mathbf S_i), \]
where the last equality is the definition of the relation condition number.

For the database expansion matrix, block-diagonal multiplication gives
\[ \mathbf E_{\mathbf D}^\top\mathbf E_{\mathbf D}=\bigl(\mathbf E_1^\top\mathbf E_1,\ldots,\mathbf E_r^\top\mathbf E_r\bigr). \]
Each block is diagonal, so the entire matrix is diagonal, with entries $\phi_i(j)$ over all $i\in[r]$ and $j\in[d_i]$. Writing a full SVD as $\mathbf E_{\mathbf D}=\mathbf U_{\mathbf D}\mathbf \Sigma_{\mathbf D}\mathbf V_{\mathbf D}^\top$, we also have
\[ \mathbf E_{\mathbf D}^\top\mathbf E_{\mathbf D}=\mathbf V_{\mathbf D}(\mathbf\Sigma_{\mathbf D}^\top\mathbf \Sigma_{\mathbf D})\mathbf V_{\mathbf D}^\top. \]
As above, this is an SVD of $\mathbf E_{\mathbf D}^\top\mathbf E_{\mathbf D}$, whose positive singular values are therefore the squares of the positive singular values of $\mathbf E_{\mathbf D}$. Since the same matrix is diagonal and nonnegative, these squared singular values are precisely the positive tuple multiplicities across the database. Consequently,
\[ \kappa(\mathbf E_{\mathbf D})=\sqrt{\frac{\max_{i\in[r]}\max_{j\in[d_i]}\phi_i(j)}{\min_{i\in[r]}\min_{j\in[d_i]:\,\phi_i(j)>0}\phi_i(j)}}=\kappa_{\mathrm{join}}(\mathbf D), \]
as required.
\end{proof}

\begin{lemma}
\label{lemma:jcn_base}
Let $\mathbf{s}=(s_1,\ldots,s_n)\neq\mathbf{0}$, and form the vector
\[ \mathbf{a}=(\underbrace{s_1,\ldots,s_1}_{\lambda_1\text{ times}},\underbrace{s_2,\ldots,s_2}_{\lambda_2\text{ times}},\ldots,\underbrace{s_n,\ldots,s_n}_{\lambda_n\text{ times}})^\top, \]
where the positive integers $\lambda_i$ satisfy
\[ 1\leq\lambda_1\leq\lambda_2\leq\cdots\leq\lambda_n. \]
For a perturbation $\Delta\mathbf{s}=(\Delta s_1,\ldots,\Delta s_n)^\top$, define the corresponding repeated-entry perturbation
\[ \Delta\mathbf{a}=(\underbrace{\Delta s_1,\ldots,\Delta s_1}_{\lambda_1\text{ times}},\underbrace{\Delta s_2,\ldots,\Delta s_2}_{\lambda_2\text{ times}},\ldots,\underbrace{\Delta s_n,\ldots,\Delta s_n}_{\lambda_n\text{ times}})^\top. \]
Then
\[ \sqrt{\frac{\lambda_1}{\lambda_n}}\cdot\frac{\|\Delta\mathbf{a}\|_2}{\|\mathbf{a}\|_2}\leq\frac{\|\Delta\mathbf{s}\|_2}{\|\mathbf{s}\|_2}\leq\sqrt{\frac{\lambda_n}{\lambda_1}}\cdot\frac{\|\Delta\mathbf{a}\|_2}{\|\mathbf{a}\|_2}. \]
\end{lemma}

\begin{proof}
By the definitions of $\mathbf{a}$ and $\Delta\mathbf{a}$,
\[ \|\mathbf{a}\|_2^2=\sum_{i=1}^n\lambda_i|s_i|^2,\qquad \|\Delta\mathbf{a}\|_2^2=\sum_{i=1}^n\lambda_i|\Delta s_i|^2. \]
Since $\lambda_1\leq\lambda_i\leq\lambda_n$ for every $i$,
\[ \lambda_1\|\Delta\mathbf{s}\|_2^2\leq\|\Delta\mathbf{a}\|_2^2\leq\lambda_n\|\Delta\mathbf{s}\|_2^2 \]
and
\[ \lambda_1\|\mathbf{s}\|_2^2\leq\|\mathbf{a}\|_2^2\leq\lambda_n\|\mathbf{s}\|_2^2. \]
Therefore,
\[ \frac{\lambda_1}{\lambda_n}\cdot\frac{\|\Delta\mathbf{a}\|_2^2}{\|\mathbf{a}\|_2^2}\leq\frac{\|\Delta\mathbf{s}\|_2^2}{\|\mathbf{s}\|_2^2}\leq\frac{\lambda_n}{\lambda_1}\cdot\frac{\|\Delta\mathbf{a}\|_2^2}{\|\mathbf{a}\|_2^2}. \]
Taking square roots gives the result.
\end{proof}

Now we provide two lemmas that translate the perturbation bounds between the input database and the join matrix in calibrated databases. We then generalize the results to any arbitrary database.

\begin{lemma}
\label{lemma:join_condition_number_columnwise}
Consider the setting in Def.~\ref{def:eval-setting} with the condition that the input database is fully reduced and let $\vect A+\Delta \vect A$ be the join matrix obtained from the perturbed input database $
\mathbf{D}+\Delta\mathbf{D}=(\vect S_1+\Delta\vect S_1,\dots,\vect S_r+\Delta\vect S_r)$. For every $j\in[r]$, $Y\in\mathbf{Y}_j$, and $\gamma,\gamma'\geq 0$, the following statements hold:
\begin{itemize}
    \item If $\norm{\Delta\vect A[*,Y]}_2\leq\gamma\cdot\norm{\vect A[*,Y]}_2$, then $\norm{\Delta\vect S_j[*,Y]}_2\leq\kappa_{\mathrm{join}}(\mathbf S_j)\cdot\gamma\cdot\norm{\vect S_j[*,Y]}_2$.
    \item If $\norm{\Delta\vect S_j[*,Y]}_2\leq \gamma'\cdot\norm{\vect S_j[*,Y]}_2$, then $\norm{\Delta\vect A[*,Y]}_2\leq\kappa_{\mathrm{join}}(\mathbf S_j)\cdot\gamma'\cdot\norm{\vect A[*,Y]}_2$.
\end{itemize}
\end{lemma}

\begin{proof}

Let $s_p=\mathbf{S}_j[p,Y]$. In the column $\mathbf{A}[*,Y]$, the entry $s_p$ occurs exactly $\lambda_p=\phi_j(p)$ times. Since the relations are fully reduced, $\lambda_p>0$ for every $p\in[d_j]$ where $d_j$ denotes the number of rows in $\mathbf{S}_j$. Moreover, because the perturbations preserve the key columns and the join structure, $\Delta\mathbf{S}_j[p,Y]$ is repeated exactly $\lambda_p$ times in $\Delta\mathbf{A}[*,Y]$.

If $\mathbf{S}_j[*,Y]=\mathbf{0}$, then $\mathbf{A}[*,Y]=\mathbf{0}$. In the first case, the assumption implies $\Delta\mathbf{A}[*,Y]=\mathbf{0}$ and hence $\Delta\mathbf{S}_j[*,Y]=\mathbf{0}$. In the second case, the assumption implies $\Delta\mathbf{S}_j[*,Y]=\mathbf{0}$ and hence $\Delta\mathbf{A}[*,Y]=\mathbf{0}$. Thus, both claims hold. We may therefore assume that $\mathbf{S}_j[*,Y]\neq\mathbf{0}$. 

Applying Lemma~\ref{lemma:jcn_base} gives
\[ \sqrt{\frac{\min_{p\in [d_j]}\lambda_p}{\max_{p\in [d_j]}\lambda_p}}\cdot\frac{\|\Delta\mathbf{A}[*,Y]\|_2}{\|\mathbf{A}[*,Y]\|_2}\leq\frac{\|\Delta\mathbf{S}_j[*,Y]\|_2}{\|\mathbf{S}_j[*,Y]\|_2}\leq\sqrt{\frac{\max_{p\in [d_j]}\lambda_p}{\min_{p\in [d_j]}\lambda_p}}\cdot\frac{\|\Delta\mathbf{A}[*,Y]\|_2}{\|\mathbf{A}[*,Y]\|_2}. \] By the definition of the relation condition number and using the fact that the database is fully reduced and hence the multiplicity $\phi_j(p)>0$ for every row $p$ in relation $\mathbf{S}_j$,
\[ \sqrt{\frac{\max_{p\in [d_j]}\lambda_p}{\min_{p\in [d_j]}\lambda_p}}=\sqrt{\frac{\max_{p\in [d_j]}\phi_j(p)}{\min_{p\in [d_j]}\phi_j(p)}}=\kappa_\mathrm{join}(\mathbf{S}_j). \]
Consequently,
\[ \frac{1}{\kappa_\mathrm{join}(\mathbf{S}_j)}\cdot\frac{\|\Delta\mathbf{A}[*,Y]\|_2}{\|\mathbf{A}[*,Y]\|_2}\leq\frac{\|\Delta\mathbf{S}_j[*,Y]\|_2}{\|\mathbf{S}_j[*,Y]\|_2}\leq\kappa_\mathrm{join}(\mathbf{S}_j)\cdot\frac{\|\Delta\mathbf{A}[*,Y]\|_2}{\|\mathbf{A}[*,Y]\|_2}. \]
If $\|\Delta\mathbf{A}[*,Y]\|_2\leq\gamma\|\mathbf{A}[*,Y]\|_2$, then the right-hand inequality gives
\[ \|\Delta\mathbf{S}_j[*,Y]\|_2\leq\kappa_\mathrm{join}(\mathbf{S}_j)\cdot\gamma\cdot\|\mathbf{S}_j[*,Y]\|_2. \]
Similarly, if $\|\Delta\mathbf{S}_j[*,Y]\|_2\leq\gamma'\cdot\|\mathbf{S}_j[*,Y]\|_2$, then the left-hand inequality gives
\[ \|\Delta\mathbf{A}[*,Y]\|_2\leq\kappa_\mathrm{join}(\mathbf{S}_j)\cdot\gamma'\cdot\|\mathbf{A}[*,Y]\|_2. \]
This proves both claims.
\end{proof}

\begin{lemma}
\label{lemma:join_condition_number_fnorm}
Consider the setting in Def.~\ref{def:eval-setting} with the condition that the input database is fully reduced and let $\vect A+\Delta \vect A$ be the join matrix obtained from the perturbed input database $
\mathbf{D}+\Delta\mathbf{D}=(\vect S_1+\Delta\vect S_1,\dots,\vect S_r+\Delta\vect S_r)$. For $\gamma,\gamma'\geq 0$, 
\begin{itemize}
\item If $\|\Delta \mathbf{A}\|_F \leq \gamma \cdot \|\mathbf{A}\|_F$, then $\|\Delta \mathbf{D}\|_F \leq \kappa_{\mathrm{join}}(\mathbf{D})\cdot \gamma \cdot \|\mathbf{D}\|_F$.
\item If $\|\Delta \mathbf{D}\|_F \leq \gamma' \cdot \|\mathbf{D}\|_F$, then $\|\Delta \mathbf{A}\|_F \leq \kappa_{\mathrm{join}}(\mathbf{D})\cdot \gamma' \cdot\|\mathbf{A}\|_F$.
\end{itemize}

\end{lemma}

\begin{proof}
If $\mathbf{D}=\mathbf 0$, the join matrix $\mathbf{A}=\mathbf 0$. Conversely, if the join matrix $\mathbf{A}=\mathbf 0$, the input database $\mathbf{D}=\mathbf{0}$ since it is fully reduced. In both cases, the statements trivially hold. We consider cases where $\mathbf{A}\neq \mathbf 0$ and $\mathbf{D}\neq \mathbf 0$. For the $p$-th row in $\mathbf S_j$, each of its data entries occurs in $\mathbf{A}$ exactly $\lambda_{j,p}=\phi_j(p)$ times. Since the relations are fully reduced, all such multiplicities are positive. Moreover, because the perturbations preserve the key values and the join structure, the corresponding perturbed entries have the same multiplicities in $\Delta\mathbf{A}$. Therefore, the minimum and maximum parameter multiplicities are
\[ \lambda_{\min}=\min_{j\in[r]}\min_{p\in[d_j]}\phi_j(p),\qquad \lambda_{\max}=\max_{j\in[r]}\max_{p\in[d_j]}\phi_j(p). \]

Applying Lemma~\ref{lemma:jcn_base} gives
\[ \sqrt{\frac{\lambda_{\min}}{\lambda_{\max}}}\cdot\frac{\|\Delta\mathbf{A}\|_F}{\|\mathbf{A}\|_F}\leq\frac{\sqrt{\sum_{j=1}^r\|\Delta\mathbf{S}_j\|_F^2}}{\sqrt{\sum_{j=1}^r\|\mathbf{S}_j\|_F^2}}\leq\sqrt{\frac{\lambda_{\max}}{\lambda_{\min}}}\cdot\frac{\|\Delta\mathbf{A}\|_F}{\|\mathbf{A}\|_F}. \]
Since $\sqrt{\frac{\lambda_{\max}}{\lambda_{\min}}}=\kappa_\mathrm{join}(\mathbf{D})$,
we obtain
\[ \frac{1}{\kappa_\mathrm{join}(\mathbf{D})}\cdot\frac{\|\Delta\mathbf{A}\|_F}{\|\mathbf{A}\|_F}\leq\frac{\|\Delta\mathbf{D}\|_F}{\|\mathbf{D}\|_F}\leq\kappa_\mathrm{join}(\mathbf{D})\cdot\frac{\|\Delta\mathbf{A}\|_F}{\|\mathbf{A}\|_F}. \]
If $\|\Delta\mathbf{A}\|_F\leq\gamma\cdot\|\mathbf{A}\|_F$, then the right-hand inequality gives
\[ \|\Delta\mathbf{D}\|_F\leq\kappa_\mathrm{join}(\mathbf{D})\cdot\gamma\cdot\|\mathbf{D}\|_F. \]
Similarly, if $\|\Delta\mathbf{D}\|_F\leq\gamma'\cdot\|\mathbf{D}\|_F$, then the left-hand inequality gives
\[ \|\Delta\mathbf{A}\|_F\leq\kappa_\mathrm{join}(\mathbf{D})\cdot\gamma'\cdot\|\mathbf{A}\|_F. \]
This proves both claims.
\end{proof}

We now show that Lemmas~\ref{lemma:join_condition_number_columnwise} and~\ref{lemma:join_condition_number_fnorm} imply that if we can bound the perturbation on the join matrix, we can prove the existence of a bounded perturbation in the input database corresponding to the perturbation in the join matrix, even if the database is not fully reduced.

\condboundtranslationdangling*

\begin{proof}
For every $j\in[r]$, let $I_j^+=\{p:\phi_j(p)>0\}$,
and let $\vect S_j^+$ denote the restriction of $\vect S_j$ to the tuples indexed by $I_j^+$. Thus, $\vect S_j^+$ consists exactly of the non-dangling tuples of $\vect S_j$. Let
\[ \mathbf{D}^+=(\vect S_1^+,\dots,\vect S_r^+). \]
The database $\mathbf{D}^+$ is fully reduced and has the same join matrix $\vect A$ as $\mathbf{D}$.

For each $p\in I_j^+$ and $Y\in\mathbf{Y}_j$, define $\Delta\vect S_j[p,Y]$ to be the common value of the entries of $\Delta\vect A[*,Y]$ corresponding to occurrences of $\vect S_j[p,Y]$. This is well defined because $\Delta\vect A$ is on the join space. For each dangling tuple, set
\[ \Delta\vect S_j[p,Y]=0\qquad\text{whenever}\qquad\phi_j(p)=0. \]
Because the perturbations preserve the join-key columns and the join structure, the join matrix obtained from $\mathbf{D}+\Delta\mathbf{D}$ is exactly $\vect A+\Delta\vect A$.

Let $\Delta\vect S_j^+$ be the restriction of $\Delta\vect S_j$ to $I_j^+$, and let $\Delta\mathbf{D}^+=(\Delta\vect S_1^+,\dots,\Delta\vect S_r^+)$.
Since removing dangling tuples does not change any positive tuple multiplicity, $\kappa_\mathrm{join}(\mathbf{S}_j^+)=\kappa_\mathrm{join}(\mathbf{S}_j)$ for every $j\in[r]$, and $\kappa_\mathrm{join}(\mathbf{D}^+)=\kappa_\mathrm{join}(\mathbf{D})$.

Suppose first that $\norm{\Delta\vect A[*,Y]}_2\leq\gamma\cdot\norm{\vect A[*,Y]}_2$
for every $j\in[r]$ and $Y\in\mathbf{Y}_j$. Since $\mathbf{D}^+$ is fully reduced, Lemma~\ref{lemma:join_condition_number_columnwise} gives $\norm{\Delta\vect S_j^+[*,Y]}_2\leq\kappa_\mathrm{join}(\mathbf{S}_j^+)\cdot\gamma\cdot\norm{\vect S_j^+[*,Y]}_2$.
The perturbation vanishes on the dangling tuples, so $\norm{\Delta\vect S_j[*,Y]}_2=\norm{\Delta\vect S_j^+[*,Y]}_2$.
Moreover, $\norm{\vect S_j^+[*,Y]}_2\leq\norm{\vect S_j[*,Y]}_2$.
Consequently, $\norm{\Delta\vect S_j[*,Y]}_2\leq\kappa_\mathrm{join}(\mathbf{S}_j)\cdot\gamma\cdot\norm{\vect S_j[*,Y]}_2$.

Now suppose that $\norm{\Delta\vect A}_F\leq\gamma\norm{\vect A}_F$. Since $\mathbf{D}^+$ is fully reduced, Lemma~\ref{lemma:join_condition_number_fnorm} gives $\norm{\Delta\mathbf{D}^+}_F\leq\kappa_\mathrm{join}(\mathbf{D}^+)\cdot\gamma\cdot\norm{\mathbf{D}^+}_F$. Because the perturbation vanishes on all dangling tuples, \\$\norm{\Delta\mathbf{D}}_F=\norm{\Delta\mathbf{D}^+}_F$, while $\norm{\mathbf{D}^+}_F\leq\norm{\mathbf{D}}_F$.
It follows that $\norm{\Delta\mathbf{D}}_F\leq\kappa_\mathrm{join}(\mathbf{D})\cdot\gamma\cdot\norm{\mathbf{D}}_F$. 
This proves both claims.
\end{proof}

\subsection{Proofs of Statements from Section~\ref{sec:pbs-bs-connection}}
\label{sec:pbs-bs-connection-proof}
Here we provide the proofs of the formal statements in Section~\ref{sec:pbs-bs-connection}.

\pbsimpliesbs*
\begin{proof}
We first show that it suffices to consider calibrated input databases. Fix an arbitrary $\mathbf A\in\mathcal Y_0$, and let $\mathbf D$ be any database with the fixed join keys such that $f(\mathbf D)=\mathbf A$. Let $\bar{\mathbf D}$ be obtained from $\mathbf D$ by setting all data entries of every dangling tuple to zero, without changing any tuple identity or join key. Since dangling tuples do not contribute to the join, $f(\bar{\mathbf D})=f(\mathbf D)=\mathbf A$, and hence $C(\bar{\mathbf D})=C_g(\mathbf A)$.

Apply the assumed $(\alpha,\beta)$-projected backward stability of $C$ to $\bar{\mathbf D}$, and let $\Delta\bar{\mathbf D}$ and $\Delta\mathbf A^\perp$ be the corresponding perturbations, with $\tilde{\mathbf A}=f(\bar{\mathbf D}+\Delta\bar{\mathbf D})$. Because perturbations do not change the join keys, every tuple that is dangling in $\bar{\mathbf D}$ remains dangling in $\bar{\mathbf D}+\Delta\bar{\mathbf D}$. Therefore, setting the entries of $\Delta\bar{\mathbf D}$ corresponding to dangling tuples to zero does not change $\tilde{\mathbf A}$ and cannot increase any of the entry-wise, column-wise, or Frobenius-norm-wise perturbation bounds.

Now let $\mathbf D^+$ be obtained from $\bar{\mathbf D}$ by removing all dangling tuples, and let $\Delta\mathbf D^+$ be the corresponding restriction of $\Delta\bar{\mathbf D}$. Then $\mathbf D^+$ is calibrated and $f(\mathbf D^+)=\mathbf A$, $f(\mathbf D^++\Delta\mathbf D^+)=\tilde{\mathbf A}$. Moreover, since all removed data entries are zero, the relevant input norms are unchanged. For every relation $\mathbf S_j$ and column $k$,
\[ \|\mathbf S_j^+[*,k]\|_2=\|\bar{\mathbf S}_j[*,k]\|_2,\qquad \|\mathbf D^+\|_F=\|\bar{\mathbf D}\|_F, \]
while restricting the perturbation can only decrease these norms:
\[ \|\Delta\mathbf S_j^+[*,k]\|_2\leq\|\Delta\bar{\mathbf S}_j[*,k]\|_2,\qquad \|\Delta\mathbf D^+\|_F\leq\|\Delta\bar{\mathbf D}\|_F. \]
The entry-wise perturbation bound is preserved directly by restriction. Finally, removing dangling tuples does not change any positive tuple-occurrence number in the join, and therefore
\[ \kappa_{\mathrm{join}}(\mathbf S_j^+)=\kappa_{\mathrm{join}}(\bar{\mathbf S}_j),\qquad \kappa_{\mathrm{join}}(\mathbf D^+)=\kappa_{\mathrm{join}}(\bar{\mathbf D}). \]
Thus the projected-backward-stability result on an uncalibrated database can always be represented using a calibrated database, given by the subset of tuples that contribute to the join matrix, without changing $\mathbf A$, $\tilde{\mathbf A}$, $\Delta\mathbf A^\perp$, the parameters $\alpha,\beta$, or the relevant join condition numbers. Hence, without loss of generality, we assume throughout the remainder of the proof that the input database $\mathbf D$ is calibrated.

We write $C=(C_f, C_g)$, where $C_f$ denotes forming the join, and $C_g$ denotes the computation of $g$ on the join matrix.

\textbf{Case 1: $C=(C_f, C_g)$ is entry-wise $(\alpha,\beta)$-projected backward stable.}
Let $\hat{z}=C_g(\mathbf{A})=C(\mathbf{D})$ be the computed result of $g\circ f$. Since $C=(C_f, C_g)$ is entry-wise $(\alpha,\beta)$-projected backward stable, there exists a perturbation on the input database $\Delta\mathbf{D}$ satisfying $|\Delta\mathbf{D}| \leq \alpha|\mathbf{D}|$, and the join matrix obtained from $\mathbf{D}+\Delta\mathbf{D}$ can be denoted as $\tilde{\mathbf{A}}$. Further, there exists $\Delta\mathbf{A}^\perp\in\mathcal{Y}_0^\perp$ satisfying $|\Delta\mathbf{A}^\perp|\leq\beta|\tilde{\mathbf{A}}|$, such that $g(\tilde{\mathbf{A}}+\Delta\mathbf{A}^\perp)=\hat{z}$.

Denote $\mathbf{A}$ to be the join matrix of the input database $\mathbf{D}$. Since $|\Delta\mathbf{D}| \leq \alpha|\mathbf{D}|$, we have $|\mathbf{A}-\tilde{\mathbf{A}}| \leq \alpha|\mathbf{A}|$ since the entries in $\mathbf{A}$ and $\tilde{\mathbf{A}}$ are copies of the entries in $\mathbf{D}$ and $\mathbf{D}+\Delta\mathbf{D}$, respectively. We then have:
\[|\mathbf{A}-(\tilde{\mathbf{A}}+\Delta\mathbf{A}^\perp)|\leq|\mathbf{A}-\tilde{\mathbf{A}}|+|\Delta\mathbf{A}^\perp|\leq\alpha|\mathbf{A}|+\beta|\tilde{\mathbf{A}}|\leq\alpha|\mathbf{A}|+\beta(|\tilde{\mathbf{A}}-\mathbf{A}|+|\mathbf{A}|)\]
\[\leq\alpha|\mathbf{A}|+\beta(\alpha|\mathbf{A}|+|\mathbf{A}|)=\Bigl(\alpha+\beta(\alpha+1)\Bigr)|\mathbf{A}|.\]
In other words, for any input $\mathbf{A}\in\mathcal{Y}_0$, write $\mathbf{E}=\tilde{\mathbf{A}}+\Delta\mathbf{A}^\perp-\mathbf{A}$, if $C_g(\mathbf{A})=\hat{z}$ then $g(\mathbf{A}+\mathbf{E})=\hat{z}$ with $|\mathbf{E}|\leq \Big(\alpha+\beta(\alpha+1)\Bigr)|\mathbf{A}|$. Hence $C_g$ is entry-wise $\alpha+\beta(\alpha+1)$-backward stable when the input lies in $\mathcal{Y}_0$.

\textbf{Case 2: $C=(C_f, C_g)$ is column-wise $(\alpha,\beta)$-projected backward stable.}

Let $\hat{z}=C_g(\mathbf{A})=C(\mathbf{D})$ be the computed result of $g\circ f$. Since $C=(C_f, C_g)$ is column-wise $(\alpha,\beta)$-projected backward stable, there exists a perturbation on the input database $\Delta\mathbf{D}=\set{\Delta\mathbf{S}_1,...,\Delta\mathbf{S}_r}$ satisfying $\|\Delta\mathbf{S}_j[*,k]\|_2 \leq \alpha\|\mathbf{S}_j[*,k]\|_2$ for every relation $S_j$ and column $k$, and the join matrix obtained from $\mathbf{D}+\Delta\mathbf{D}$ can be denoted as $\tilde{\mathbf{A}}$. Further, there exists $\Delta\mathbf{A}^\perp\in\mathcal{Y}_0^\perp$ satisfying $\|\Delta\mathbf{A}^\perp[*,i]\|_2\leq\beta\|\tilde{\mathbf{A}}[*,i]\|_2$ for every column $i$, such that $g(\tilde{\mathbf{A}}+\Delta\mathbf{A}^\perp)=\hat{z}$.

Denote by $\mathbf{A}$ the join matrix for the input database $\mathbf{D}$. $\|\Delta\mathbf{S}_j[*,k]\|_2 \leq \alpha\|\mathbf{S}_j[*,k]\|_2$ for every relation $S_j$ and column $k$. Using Lemma~\ref{lemma:join_condition_number_columnwise}, we have $\|(\mathbf{A}-\tilde{\mathbf{A}})[*,i]\|_2 \leq \kappa_\mathrm{join}(\mathbf{S}_j)\alpha\|\mathbf{A}[*,i]\|_2$ for column $i$ in $\mathbf{A}$ that comes from the relation $S_j$, as the entries in $\mathbf{A}$ and $\tilde{\mathbf{A}}$ are copies of the entries in $\mathbf{D}$ and $\mathbf{D}+\Delta\mathbf{D}$, respectively. Notice that the query $Q$ on $\mathbf{D}$ does not contain self joins, projecting $\tilde{\mathbf{A}}+\Delta\mathbf{A}^\perp$ on $\mathcal{Y}_0$ is equivalent to projecting on $\mathcal{Y}_0$ with each column projected independently. The Pythagorean theorem therefore tells us that $\|(\mathbf{A}-\tilde{\mathbf{A}})[*,i]\|_2^2+\|\Delta\mathbf{A}^\perp[*,i]\|_2^2=\big\|\big(\mathbf{A}-(\tilde{\mathbf{A}}+\Delta\mathbf{A}^\perp)\big)[*,i]\big\|_2^2$ for every column $i$.

We then have:
\[\big\|\big(\mathbf{A}-(\tilde{\mathbf{A}}+\Delta\mathbf{A}^\perp)\big)[*,i]\big\|_2^2\;=\;\|(\mathbf{A}-\tilde{\mathbf{A}})[*,i]\|_2^2+\|\Delta\mathbf{A}^\perp[*,i]\|_2^2\]
\[\leq\;\alpha^2\kappa_\mathrm{join}(\mathbf{S}_j)^2\|\mathbf{A}[*,i]\|_2^2+\|\Delta\mathbf{A}^\perp[*,i]\|_2^2\;\leq\;\alpha^2\kappa_\mathrm{join}(\mathbf{S}_j)^2\|\mathbf{A}[*,i]\|_2^2+\beta^2\|\tilde{\mathbf{A}}[*,i]\|_2^2.\]
Further, using the triangle inequality, \[\|\tilde{\mathbf{A}}[*,i]\|_2\leq\|\mathbf{A}[*,i]\|_2+\|(\mathbf{A}-\tilde{\mathbf{A}})[*,i]\|_2\leq\big(1+\alpha\kappa_\mathrm{join}(\mathbf{S}_j)\big)\|\mathbf{A}[*,i]\|_2,\]
we obtain 
\[\big\|\big(\mathbf{A}-(\tilde{\mathbf{A}}+\Delta\mathbf{A}^\perp)\big)[*,i]\big\|_2^2\;\leq\;\alpha^2\kappa_\mathrm{join}(\mathbf{S}_j)^2\|\mathbf{A}[*,i]\|_2^2+\beta^2\|\tilde{\mathbf{A}}[*,i]\|_2^2\]
\[\leq\alpha^2\kappa_\mathrm{join}(\mathbf{S}_j)^2\|\mathbf{A}[*,i]\|_2^2+\beta^2\big(1+\alpha\kappa_\mathrm{join}(\mathbf{S}_j)\big)^2\|\mathbf{A}[*,i]\|_2^2,\]
which gives us the bound for the backward error of $g$:
\[\big\|\big(\mathbf{A}-(\tilde{\mathbf{A}}+\Delta\mathbf{A}^\perp)\big)[*,i]\big\|_2\leq\sqrt{\kappa_\mathrm{join}(\mathbf{S}_j)^2\alpha^2+\beta^2\big(1+\kappa_\mathrm{join}(\mathbf{S}_j)\alpha\big)^2}\;\Big\|\mathbf{A}[*,i]\Big\|_2\]

In other words, for any input $\mathbf{A}\in\mathcal{Y}_0$, write $\mathbf{E}=\tilde{\mathbf{A}}+\Delta\mathbf{A}^\perp-\mathbf{A}$, if $C_g(\mathbf{A})=\hat{z}$ then $g(\mathbf{A}+\mathbf{E})=\hat{z}$ with $\|\mathbf{E}[*,i]\|_2\leq \sqrt{\kappa_\mathrm{join}(\mathbf{S}_j)^2\alpha^2+\beta^2\big(1+\kappa_\mathrm{join}(\mathbf{S}_j)\alpha\big)^2}\;\Big\|\mathbf{A}[*,i]\Big\|_2$. Hence $C_g$ is column-wise $\sqrt{\kappa_\mathrm{join}(\mathbf{S}_j)^2\alpha^2+\beta^2\big(1+\kappa_\mathrm{join}(\mathbf{S}_j)\alpha\big)^2}$-backward stable when the input lies in $\mathcal{Y}_0$.

\textbf{Case 3: $C=(C_f, C_g)$ is Frobenius-norm-wise $(\alpha,\beta)$-projected backward stable.}

Let $\hat{z}=C_g(\mathbf{A})=C(\mathbf{D})$ be the computed result of $g\circ f$. Since $C=(C_f, C_g)$ is Frobenius-norm-wise $(\alpha,\beta)$-projected backward stable, there exists a perturbation on the input database $\Delta\mathbf{D}=\set{\Delta\mathbf{S}_1,...,\Delta\mathbf{S}_r}$ satisfying $\|\Delta\mathbf{D}\|_F \leq \alpha\|\mathbf{D}\|_F$, and the join matrix obtained from $\mathbf{D}+\Delta\mathbf{D}$ can be denoted as $\tilde{\mathbf{A}}$. Further, there exists $\Delta\mathbf{A}^\perp\in\mathcal{Y}_0^\perp$ satisfying $\|\Delta\mathbf{A}^\perp\|_F\leq\beta\|\tilde{\mathbf{A}}\|_F$ such that $g(\tilde{\mathbf{A}}+\Delta\mathbf{A}^\perp)=\hat{z}$.

Denote $\mathbf{A}$ to be the join matrix of the input database $\mathbf{D}$. Since $\|\Delta\mathbf{D}\|_F \leq \alpha\|\mathbf{D}\|_F$, using Lemma~\ref{lemma:join_condition_number_fnorm}, we have $\|(\mathbf{A}-\tilde{\mathbf{A}})\|_F \leq \kappa_\mathrm{join}(\mathbf{D})\alpha\|\mathbf{A}\|_F$. Moreover, the Pythagorean theorem tells us that $\|(\mathbf{A}-\tilde{\mathbf{A}})\|_F^2+\|\Delta\mathbf{A}^\perp\|_F^2=\big\|\big(\mathbf{A}-(\tilde{\mathbf{A}}+\Delta\mathbf{A}^\perp)\big)\big\|_F^2$, we then have:
\[\big\|\big(\mathbf{A}-(\tilde{\mathbf{A}}+\Delta\mathbf{A}^\perp)\big)\big\|_F^2=\|(\mathbf{A}-\tilde{\mathbf{A}})\|_F^2+\|\Delta\mathbf{A}^\perp\|_F^2\leq\alpha^2\kappa_\mathrm{join}(\mathbf{D})^2\|\mathbf{A}\|_F^2+\|\Delta\mathbf{A}^\perp\|_F^2\]
\[\leq\alpha^2\kappa_\mathrm{join}(\mathbf{D})^2\|\mathbf{A}\|_F^2+\beta^2\|\tilde{\mathbf{A}}\|_F^2.\]
Further, using the triangle inequality, \[\|\tilde{\mathbf{A}}\|_F\leq\|\mathbf{A}\|_F+\|(\mathbf{A}-\tilde{\mathbf{A}})\|_F\leq\big(1+\alpha\kappa_\mathrm{join}(\mathbf{D})\big)\|\mathbf{A}\|_F,\]
we obtain 
\[
\big\|\big(\mathbf{A}-(\tilde{\mathbf{A}}+\Delta\mathbf{A}^\perp)\big)\big\|_F^2\leq\alpha^2\kappa_\mathrm{join}(\mathbf{D})^2\|\mathbf{A}\|_F^2+\beta^2\|\tilde{\mathbf{A}}\|_F^2
\]
\[\leq \alpha^2\kappa_\mathrm{join}(\mathbf{D})^2\|\mathbf{A}\|_F^2+\beta^2\big(1+\alpha\kappa_\mathrm{join}(\mathbf{D})\big)^2\|\mathbf{A}\|_F^2\]
which gives us the bound for the backward error of $g$:
\[\big\|\big(\mathbf{A}-(\tilde{\mathbf{A}}+\Delta\mathbf{A}^\perp)\big)\big\|_F\leq\sqrt{\kappa_\mathrm{join}(\mathbf{D})^2\alpha^2+\beta^2\big(1+\kappa_\mathrm{join}(\mathbf{D})\alpha\big)^2}\Big\|\mathbf{A}\Big\|_F\]
In other words, for any input $\mathbf{A}\in\mathcal{Y}_0$, write $\mathbf{E}=\tilde{\mathbf{A}}+\Delta\mathbf{A}^\perp-\mathbf{A}$, if $C_g(\mathbf{A})=\hat{z}$ then $g(\mathbf{A}+\mathbf{E})=\hat{z}$ with $\|\mathbf{E}\|_F\leq \sqrt{\kappa_\mathrm{join}(\mathbf{D})^2\alpha^2+\beta^2\big(1+\kappa_\mathrm{join}(\mathbf{D})\alpha\big)^2}\Big\|\mathbf{A}\Big\|_F$. Hence $C_g$ is Frobenius-norm-wise \\$\sqrt{\kappa_\mathrm{join}(\mathbf{D})^2\alpha^2+\beta^2\big(1+\kappa_\mathrm{join}(\mathbf{D})\alpha\big)^2}$-backward stable when the input lies in $\mathcal{Y}_0$.
\end{proof}

\bsimpliespbs*
\begin{proof}
We write $C=(C_f, C_g)$, where $C_f$ denotes forming the join, and $C_g$ denotes the computation of $g$ on the join matrix.

\textbf{Case 1: $C_g$ is entry-wise $\epsilon$-backward stable.}

Let $\hat{z}=C_g(\mathbf{A})=C(\mathbf{D})$ be the computed result of $g\circ f$. Since $C_g$ is entry-wise $\epsilon$-backward stable, there exists a perturbation matrix $\mathbf{E}$ satisfying $\hat{z}=g(\mathbf{A}+\mathbf{E})$, such that $|\mathbf{E}|\leq\epsilon|\mathbf{A}|$. We can write $\mathbf{A}+\mathbf{E}$ as $\tilde{\mathbf{A}}+\Delta\mathbf{A}^\perp$, where $\tilde{\mathbf{A}}$ is the projection of $\mathbf{A}+\mathbf{E}$ on the join space $\mathcal{Y}_0$, and $\Delta\mathbf{A}^\perp=\mathbf{A}+\mathbf{E}-\tilde{\mathbf{A}}$ denotes the normal perturbation.

Now, consider any entry in the input database $s=\mathbf{S}_k[p,q]$ that occurs in the join matrix $\mathbf{A}$. In this case, the set $\mathcal{P}_s=\set{(i,j)\mid\mathbf{A}[i,j]\text{ is parameterized by }s}$ is nonempty. Lemma~\ref{lemma:projection_is_taking_mean} tells us that, in the projection of $\mathbf{A}+\mathbf{E}$ on $\mathcal{Y}_0$, namely $\tilde{\mathbf{A}}$, for any $(i',j')\in\mathcal{P}_s$,
\[ \tilde{\mathbf{A}}[i',j']=\frac{\sum_{(i,j)\in\mathcal{P}_s}(\tilde{\mathbf{A}}+\Delta\mathbf{A}^\perp)[i,j]}{|\mathcal{P}_s|}=\frac{\sum_{(i,j)\in\mathcal{P}_s}(\mathbf{A}+\mathbf{E})[i,j]}{|\mathcal{P}_s|}=s+\frac{\sum_{(i,j)\in\mathcal{P}_s}\mathbf{E}[i,j]}{|\mathcal{P}_s|}, \]
where
\[ \Big|\frac{\sum_{(i,j)\in\mathcal{P}_s}\mathbf{E}[i,j]}{|\mathcal{P}_s|}\Big|\leq\frac{\sum_{(i,j)\in\mathcal{P}_s}|\mathbf{E}[i,j]|}{|\mathcal{P}_s|}\leq\frac{\sum_{(i,j)\in\mathcal{P}_s}\epsilon|\mathbf{A}[i,j]|}{|\mathcal{P}_s|}=\epsilon|s|. \]

This means $|\tilde{\mathbf{A}}[i,j]-\mathbf{A}[i,j]|\leq\epsilon|s|$ for every $(i,j)\in\mathcal{P}_s$. We construct the perturbed database $\mathbf{D}+\Delta\mathbf{D}=(\mathbf{S}_1+\Delta\mathbf{S}_1,\dots,\mathbf{S}_r+\Delta\mathbf{S}_r)$ as follows. For every database entry $s=\mathbf{S}_k[p,q]$ that occurs in the join matrix, set its perturbation so that its perturbed value equals the corresponding value in $\tilde{\mathbf{A}}$. By the argument above, this perturbation satisfies
\[ |\Delta\mathbf{S}_k[p,q]|\leq\epsilon|\mathbf{S}_k[p,q]|. \]
If $s=\mathbf{S}_k[p,q]$ belongs to a dangling tuple and therefore does not occur in the join matrix, set $\Delta\mathbf{S}_k[p,q]=0$.
Such an entry does not affect the join result, and its perturbation trivially satisfies $|\Delta\mathbf{S}_k[p,q]|=0\leq\epsilon|\mathbf{S}_k[p,q]|$.

Therefore, $f(\mathbf{D}+\Delta\mathbf{D})=\tilde{\mathbf{A}}$ and $|\Delta\mathbf{S}_k[p,q]|\leq\epsilon|\mathbf{S}_k[p,q]|$ for every relation $k$ and every entry $(p,q)$. Hence, $\epsilon$ is a valid entry-wise error bound on the input database.

Next, we show that $\frac{2\epsilon}{1-\epsilon}$ is a valid bound for the normal perturbation $\Delta\mathbf{A}^\perp$. Notice that 
$\tilde{\mathbf{A}}[i,j
]=s+\frac{\sum_{(i,j)\in \mathcal{P}_s}\mathbf{E}[i,j]}{|\mathcal{P}_s|}$ and therefore $s=\tilde{\mathbf{A}}[i,j
]-\frac{\sum_{(i,j)\in \mathcal{P}_s}\mathbf{E}[i,j]}{|\mathcal{P}_s|}$. We further have:
\[|s|\leq|\tilde{\mathbf{A}}[i,j
]|+\Big|\frac{\sum_{(i,j)\in \mathcal{P}_s}\mathbf{E}[i,j]}{|\mathcal{P}_s|}\Big|\leq|\tilde{\mathbf{A}}[i,j
]|+\epsilon|s|,\]
which means $|\tilde{\mathbf{A}}[i,j
]|\geq (1-\epsilon)|s|$. On the other hand, $|\Delta\mathbf{A}^\perp[i,j]|=|(\mathbf{A}+\mathbf{E}-\mathbf{\tilde{A}})[i,j]|$ where
\[|(\mathbf{A}+\mathbf{E}-\mathbf{\tilde{A}})[i,j]|\leq|\mathbf{E}[i,j]|+|\mathbf{A}[i,j]-\mathbf{\tilde{A}}[i,j]|\leq \epsilon|s|+\epsilon|s|=2\epsilon|s|.\]
If $|s|=0$, we immediately have $|\Delta\mathbf{A}^\perp[i,j]|=0$ and $|\Delta\mathbf{A}^\perp[i,j]|\leq\frac{2\epsilon}{1-\epsilon}|\tilde{\mathbf{A}}[i,j]|$ holds true. If $|s|>0$, we also immediately have $|\Delta\mathbf{A}^\perp[i,j]|\leq\frac{2\epsilon|s|}{(1-\epsilon)|s|}|\tilde{\mathbf{A}}[i,j]|=\frac{2\epsilon}{1-\epsilon}|\tilde{\mathbf{A}}[i,j]|$. Applying this argument entry-wise gives $|\Delta\mathbf{A}^\perp|\leq\frac{2\epsilon}{1-\epsilon}|\tilde{\mathbf{A}}|$.

Therefore, the computation $C=(C_f, C_g)$ is $(\epsilon, \frac{2\epsilon}{1-\epsilon})$-projected backward stable on $g\circ f$.

\textbf{Case 2: $C_g$ is column-wise $\epsilon$-backward stable.}
Let $\hat{z}=C_g(\mathbf{A})=C(\mathbf{D})$ be the computed result of $g\circ f$. Since $C_g$ is column-wise $\epsilon$-backward stable, and the query has no self-joins, there exists a perturbation matrix $\mathbf{E}$ satisfying $\hat{z}=g(\mathbf{A}+\mathbf{E})$, such that $\|\mathbf{E}[*,i]\|_2\leq\epsilon\|\mathbf{A}[*,i]\|_2$ for every column $i$. 

Since the query has no self-joins, any data value in the input database will not parameterize entries from different columns in $\mathbf{A}$. In other words, projecting $\mathbf{A}+\mathbf{E}$ on $\mathcal{Y}_0$ is equivalent to projecting each column in $\mathbf{A}+\mathbf{E}$ onto the vector space defined by each column in $\mathcal{Y}_0$ individually. We can write $\mathbf{A}+\mathbf{E}$ as $\tilde{\mathbf{A}}+\Delta\mathbf{A}^\perp$, where $\tilde{\mathbf{A}}$ is the projection of $\mathbf{A}+\mathbf{E}$ on the join space $\mathcal{Y}_0$, and $\Delta\mathbf{A}^\perp=\mathbf{A}+\mathbf{E}-\tilde{\mathbf{A}}$ denotes the normal perturbation. For each column $i$, we apply the Pythagorean theorem, which states: $\|(\mathbf{A}-\tilde{\mathbf{A}})[*,i]\|_2^2+\|\Delta\mathbf{A}^\perp[*,i]\|_2^2=\|\mathbf{E}[*,i]\|_2^2$. Therefore  $\|(\mathbf{A}-\tilde{\mathbf{A}})[*,i]\|_2\leq \|\mathbf{E}[*,i]\|_2\leq\epsilon\|\mathbf{A}[*,i]\|_2$ for every column $i$. Applying Proposition~\ref{proposition:condition_number_translation_with_dangling_tuples}, we obtain that there exists a perturbed database $\mathbf{D}+\Delta\mathbf{D}=\set{\mathbf{S}_1+\Delta\mathbf{S}_1,...,\mathbf{S}_r+\Delta\mathbf{S}_r}$ satisfying $\tilde{\mathbf{A}}$ is the join matrix of $\mathbf{D}+\Delta\mathbf{D}$, and for every relation $\mathbf{S}_j$ and column $k$, $\|\Delta\mathbf{S}_j[*,k]\|_2\leq\kappa_\mathrm{join}(\mathbf{S}_j)\cdot\epsilon\cdot\|\mathbf{S}_j[*,k]\|_2$.

We now analyze the normal perturbation bound. Fix a column $i$. We first consider the case where $\|\mathbf{A}[*,i]\|_2=0$. Since $C_g$ is column-wise $\epsilon$-backward stable, we have
\[ \|\mathbf{E}[*,i]\|_2\leq\epsilon\|\mathbf{A}[*,i]\|_2=0. \]
Thus, $\mathbf{E}[*,i]=0$. By the Pythagorean theorem,
\[ \|(\mathbf{A}-\tilde{\mathbf{A}})[*,i]\|_2^2+\|\Delta\mathbf{A}^{\perp}[*,i]\|_2^2=\|\mathbf{E}[*,i]\|_2^2, \]
we obtain $(\mathbf{A}-\tilde{\mathbf{A}})[*,i]=0$ and $\Delta\mathbf{A}^{\perp}[*,i]=0$.
Consequently, $\tilde{\mathbf{A}}[*,i]=\mathbf{A}[*,i]=0$, and hence $\|\Delta\mathbf{A}^{\perp}[*,i]\|_2\leq\frac{\epsilon}{\sqrt{1-\epsilon^2}}\|\tilde{\mathbf{A}}[*,i]\|_2$ holds trivially.

Now suppose that $\|\mathbf{A}[*,i]\|_2>0$, and define
\[ \lambda=\frac{\|(\mathbf{A}-\tilde{\mathbf{A}})[*,i]\|_2}{\|\mathbf{A}[*,i]\|_2}. \]
Since
\[ \|(\mathbf{A}-\tilde{\mathbf{A}})[*,i]\|_2\leq\|\mathbf{E}[*,i]\|_2\leq\epsilon\|\mathbf{A}[*,i]\|_2, \]
we have $0\leq\lambda\leq\epsilon<1$.
By the triangle inequality,
\[ \|\tilde{\mathbf{A}}[*,i]\|_2=\|\mathbf{A}[*,i]-(\mathbf{A}-\tilde{\mathbf{A}})[*,i]\|_2\geq\|\mathbf{A}[*,i]\|_2-\|(\mathbf{A}-\tilde{\mathbf{A}})[*,i]\|_2=(1-\lambda)\|\mathbf{A}[*,i]\|_2>0. \]
Therefore, division by $\|\tilde{\mathbf{A}}[*,i]\|_2$ and by $(1-\lambda)\|\mathbf{A}[*,i]\|_2$ is well defined. 

Using the Pythagorean theorem again, we obtain
\[ \|\Delta\mathbf{A}^{\perp}[*,i]\|_2^2=\|\mathbf{E}[*,i]\|_2^2-\|(\mathbf{A}-\tilde{\mathbf{A}})[*,i]\|_2^2. \]
Since $\|\mathbf{E}[*,i]\|_2\leq\epsilon\|\mathbf{A}[*,i]\|_2$, it follows that
\[ \|\Delta\mathbf{A}^{\perp}[*,i]\|_2^2\leq\epsilon^2\|\mathbf{A}[*,i]\|_2^2-\lambda^2\|\mathbf{A}[*,i]\|_2^2=(\epsilon^2-\lambda^2)\|\mathbf{A}[*,i]\|_2^2. \]
Combining this bound with the lower bound on $\|\tilde{\mathbf{A}}[*,i]\|_2$, we obtain
\[ \frac{\|\Delta\mathbf{A}^{\perp}[*,i]\|_2}{\|\tilde{\mathbf{A}}[*,i]\|_2}\leq\frac{\sqrt{\epsilon^2-\lambda^2}\|\mathbf{A}[*,i]\|_2}{(1-\lambda)\|\mathbf{A}[*,i]\|_2}=\frac{\sqrt{\epsilon^2-\lambda^2}}{1-\lambda}. \]
Moreover,
\[ \frac{\epsilon^2}{1-\epsilon^2}-\frac{\epsilon^2-\lambda^2}{(1-\lambda)^2}=\frac{(\lambda-\epsilon^2)^2}{(1-\epsilon^2)(1-\lambda)^2}\geq0. \]
Since both sides are nonnegative, this implies
\[ \frac{\sqrt{\epsilon^2-\lambda^2}}{1-\lambda}\leq\frac{\epsilon}{\sqrt{1-\epsilon^2}}. \]
Therefore,
\[ \|\Delta\mathbf{A}^{\perp}[*,i]\|_2\leq\frac{\epsilon}{\sqrt{1-\epsilon^2}}\|\tilde{\mathbf{A}}[*,i]\|_2. \]
The inequality holds for every column $i$, including the case $\|\mathbf{A}[*,i]\|_2=0$. Thus, $C=(C_f,C_g)$ is column-wise $\left(\kappa_\mathrm{join}(\mathbf{S}_j)\cdot\epsilon,\frac{\epsilon}{\sqrt{1-\epsilon^2}}\right)$-projected backward stable on $g\circ f$ for columns originating from relation $\mathbf{S}_j$.

\textbf{Case 3: $C_g$ is Frobenius-norm-wise $\epsilon$-backward stable.}

Let $\hat{z}=C_g(\mathbf{A})=C(\mathbf{D})$ be the computed result of $g\circ f$. Since $C_g$ is Frobenius-norm-wise $\epsilon$-backward stable, there exists a perturbation matrix $\mathbf{E}$ satisfying $\hat{z}=g(\mathbf{A}+\mathbf{E})$, such that $\|\mathbf{E}\|_F\leq\epsilon\|\mathbf{A}\|_F$. 

We can write $\mathbf{A}+\mathbf{E}$ as $\tilde{\mathbf{A}}+\Delta\mathbf{A}^\perp$, where $\tilde{\mathbf{A}}$ is the projection of $\mathbf{A}+\mathbf{E}$ on the join space $\mathcal{Y}_0$, and $\Delta\mathbf{A}^\perp=\mathbf{A}+\mathbf{E}-\tilde{\mathbf{A}}$ denotes the normal perturbation. We apply the Pythagorean theorem, which states: $\|(\mathbf{A}-\tilde{\mathbf{A}})\|_F^2+\|\Delta\mathbf{A}^\perp\|_F^2=\|\mathbf{E}\|_F^2$. Therefore, $\|(\mathbf{A}-\tilde{\mathbf{A}})\|_F\leq \|\mathbf{E}\|_F\leq\epsilon\|\mathbf{A}\|_F$. Applying Proposition~\ref{proposition:condition_number_translation_with_dangling_tuples}, we obtain that there exists a perturbed database $\mathbf{D}+\Delta\mathbf{D}=\set{\mathbf{S}_1+\Delta\mathbf{S}_1,...,\mathbf{S}_r+\Delta\mathbf{S}_r}$ satisfying $\tilde{\mathbf{A}}$ is the join matrix of $\mathbf{D}+\Delta\mathbf{D}$, and $\|\Delta\mathbf{D}\|_F\leq\kappa_\mathrm{join}(\mathbf{D})\epsilon\|\mathbf{D}\|_F$.

We now analyze the normal perturbation bound. We first consider the case where $\|\mathbf{A}\|_F=0$. Since $C_g$ is Frobenius-norm-wise $\epsilon$-backward stable, we have $\|\mathbf{E}\|_F\leq\epsilon\|\mathbf{A}\|_F=0$. Thus, $\mathbf{E}=0$. By the Pythagorean theorem, \[ \|\mathbf{A}-\tilde{\mathbf{A}}\|_F^2+\|\Delta\mathbf{A}^{\perp}\|_F^2=\|\mathbf{E}\|_F^2, \] we therefore obtain $ \mathbf{A}-\tilde{\mathbf{A}}=0$ and $\Delta\mathbf{A}^{\perp}=0$. Consequently, $\tilde{\mathbf{A}}=\mathbf{A}=0$, and hence \[ \|\Delta\mathbf{A}^{\perp}\|_F\leq\frac{\epsilon}{\sqrt{1-\epsilon^2}}\|\tilde{\mathbf{A}}\|_F \] holds trivially. Now suppose that $\|\mathbf{A}\|_F>0$, and define \[ \lambda=\frac{\|\mathbf{A}-\tilde{\mathbf{A}}\|_F}{\|\mathbf{A}\|_F}. \] Since $ \|\mathbf{A}-\tilde{\mathbf{A}}\|_F\leq\|\mathbf{E}\|_F\leq\epsilon\|\mathbf{A}\|_F$, we have $0\leq\lambda\leq\epsilon<1$. By the triangle inequality, \[ \|\tilde{\mathbf{A}}\|_F=\|\mathbf{A}-(\mathbf{A}-\tilde{\mathbf{A}})\|_F\geq\|\mathbf{A}\|_F-\|\mathbf{A}-\tilde{\mathbf{A}}\|_F=(1-\lambda)\|\mathbf{A}\|_F>0. \] Therefore, division by $\|\tilde{\mathbf{A}}\|_F$ and by $(1-\lambda)\|\mathbf{A}\|_F$ is well defined.

Using the Pythagorean theorem again, we obtain \[ \|\Delta\mathbf{A}^{\perp}\|_F^2=\|\mathbf{E}\|_F^2-\|\mathbf{A}-\tilde{\mathbf{A}}\|_F^2. \] Since $\|\mathbf{E}\|_F\leq\epsilon\|\mathbf{A}\|_F$, it follows that \[ \|\Delta\mathbf{A}^{\perp}\|_F^2\leq\epsilon^2\|\mathbf{A}\|_F^2-\lambda^2\|\mathbf{A}\|_F^2=(\epsilon^2-\lambda^2)\|\mathbf{A}\|_F^2. \] Combining this bound with the lower bound on $\|\tilde{\mathbf{A}}\|_F$, we obtain \[ \frac{\|\Delta\mathbf{A}^{\perp}\|_F}{\|\tilde{\mathbf{A}}\|_F}\leq\frac{\sqrt{\epsilon^2-\lambda^2}\|\mathbf{A}\|_F}{(1-\lambda)\|\mathbf{A}\|_F}=\frac{\sqrt{\epsilon^2-\lambda^2}}{1-\lambda}. \] Moreover, \[ \frac{\epsilon^2}{1-\epsilon^2}-\frac{\epsilon^2-\lambda^2}{(1-\lambda)^2}=\frac{(\lambda-\epsilon^2)^2}{(1-\epsilon^2)(1-\lambda)^2}\geq0. \] Since both sides are nonnegative, this implies \[ \frac{\sqrt{\epsilon^2-\lambda^2}}{1-\lambda}\leq\frac{\epsilon}{\sqrt{1-\epsilon^2}}. \] Therefore, \[ \|\Delta\mathbf{A}^{\perp}\|_F\leq\frac{\epsilon}{\sqrt{1-\epsilon^2}}\|\tilde{\mathbf{A}}\|_F. \] This inequality also holds in the case $\|\mathbf{A}\|_F=0$. Thus, we have proved that $C=(C_f,C_g)$ is Frobenius-norm-wise $\left(\kappa_\mathrm{join}(\mathbf{D})\epsilon,\frac{\epsilon}{\sqrt{1-\epsilon^2}}\right)$-projected backward stable on $g\circ f$.
\end{proof}

\subsection{Proofs of Statements from Section~\ref{sec:qr-svd}}
Here we provide the proofs of the formal statements in Section~\ref{sec:qr-svd}.
\label{sec:qr-svd-proofs}

\begin{restatable}{proposition}{qrinputrelations}
\label{proposition:qr_input_relations_counterexample}

There exists an $\alpha$-acyclic join query $Q$ with no self-joins over fully reduced relations for which the standard floating-point computation of the upper triangular matrix $\mathbf{R}$, obtained from QR decomposition via Givens Rotations, is not backward stable with respect to the input database.
\end{restatable}

\begin{proof}
Consider the four input relations
\[ S_j=\begin{array}{c|c}
\text{key} & Y_j\\ \hline
k & 1\\
k & -1
\end{array},\qquad j=1,2,3,4. \]
Their join is the full Cartesian product over the common key. We order the resulting tuples lexicographically, with $1$ preceding $-1$. Thus the materialized join matrix $\mathbf A\in\mathbb R^{16\times 4}$ contains all $16$ rows
\[ (y_1,y_2,y_3,y_4),\qquad y_j\in\{1,-1\}. \]

We compute the QR factorization of $\mathbf A$ using Givens rotations in IEEE double, whose unit roundoff is $u=2^{-53}$. The algorithm processes columns from left to right, so $j=1,2,3,4$. Denote the intermediate matrix in Givens Rotations as $\mathbf{B}$. In column $j$, it eliminates the entries below the diagonal from bottom to top. Thus, for $i=16,15,\ldots,j+1$, the algorithm eliminates the current entry $\mathbf B[i,j]$ using the entry immediately above it, $\mathbf B[i-1,j]$.

For one such elimination, set $x=\mathbf{B}[i-1,j]$, $y=\mathbf{B}[i,j]$.
If $y=0$, no rotation is needed and the algorithm proceeds to the next entry. Otherwise, it computes
\[ r=fl\!\left(\sqrt{fl(x^2+y^2)}\right),\qquad c=fl(x/r),\qquad s=fl(y/r). \]
The corresponding Givens rotation is then applied to the two rows $i-1$ and $i$. Only columns $k=j,j+1,\ldots,4$ need to be updated, since columns $1,\ldots,j-1$ have already been reduced and have zero on those rows and the values will stay zero after any rotation. For each such $k$, the update is
\[ \begin{pmatrix}\mathbf{B}[i-1,k]\\\mathbf{B}[i,k]\end{pmatrix}\leftarrow fl\!\left(\begin{pmatrix}c&s\\-s&c\end{pmatrix}\begin{pmatrix}\mathbf{B}[i-1,k]\\\mathbf{B}[i,k]\end{pmatrix}\right). \]
After the update, the entry $\mathbf{B}[i,j]$, which we target to annihilate, is explicitly set to zero. 

In exact arithmetic, the columns of $\mathbf A$ are mutually orthogonal and each column has squared norm $16$. Hence $\mathbf A^\top\mathbf A=16\mathbf I_4,$
and the exact QR factor with positive diagonal is $\mathbf{R}=4\mathbf{I}_4.$

Let $\hat{\mathbf R}$ be the upper-triangular factor computed by the Givens algorithm described above. A direct IEEE double evaluation of this specified algorithm gives the following leading $4\times4$ triangular factor:
\[\hat{\mathbf R}
=
4\mathbf I_4
+
u
\begin{pmatrix}
0 & -\frac52 & \frac54 & \frac32\\
0 & 16 & -1 & -1\\
0 & 0 & 8 & -1\\
0 & 0 & 0 & 0
\end{pmatrix}
=
\begin{pmatrix}
4 & -\frac52u & \frac54u & \frac32u\\
0 & 4+16u & -u & -u\\
0 & 0 & 4+8u & -u\\
0 & 0 & 0 & 4
\end{pmatrix}.
\]

We now show that this $\hat{\mathbf R}$ cannot be the exact $R$ factor of the join of any perturbed input relations, no matter how large the perturbation is. A perturbation of the numerical attributes of the input relations replaces the two numerical values in $S_1,S_2,S_3,S_4$ by
$(a'_1,a'_2), (b'_1,b'_2),(c'_1,c'_2),(d'_1,d'_2)$,
respectively. Thus the perturbed relations are
\[S'_1(K,Y_1)=\{(k,a'_1),(k,a'_2)\},\qquad
S'_2(K,Y_2)=\{(k,b'_1),(k,b'_2)\},\]
\[S'_3(K,Y_3)=\{(k,c'_1),(k,c'_2)\},\qquad
S'_4(K,Y_4)=\{(k,d'_1),(k,d'_2)\}.\]
Let $\mathbf A'$ be the join matrix obtained from the natural join of these perturbed relations.

Suppose, for contradiction, that $\hat{\mathbf R}$ is the exact $R$ factor of $\mathbf A'$. Then there exists a matrix $\mathbf Q'$ with orthonormal columns such that
$\mathbf A'=\mathbf Q'\hat{\mathbf R}.$
Therefore
$\mathbf A'^{\top}\mathbf A'=\hat{\mathbf R}^{\top}\hat{\mathbf R}$. Let
$\mathbf G=\mathbf A'^{\top}\mathbf A'.$
Also define
\[\sigma_a=a'_1+a'_2,\qquad
\sigma_b=b'_1+b'_2,\qquad
\sigma_c=c'_1+c'_2,\qquad
\sigma_d=d'_1+d'_2.\]
Since the join is the full Cartesian product, the off-diagonal entries of $\mathbf G$ satisfy
\[\mathbf G[1,2]=4\sigma_a\sigma_b,\quad
\mathbf G[1,3]=4\sigma_a\sigma_c,\quad
\mathbf G[1,4]=4\sigma_a\sigma_d,\]
\[\mathbf G[2,3]=4\sigma_b\sigma_c,\quad
\mathbf G[2,4]=4\sigma_b\sigma_d,\quad
\mathbf G[3,4]=4\sigma_c\sigma_d.\]
Hence every Gram matrix obtained from any such perturbed input relations must satisfy the polynomial invariant
\[\mathbf G[1,2]\mathbf G[3,4]=\mathbf G[1,3]\mathbf G[2,4]=\mathbf G[1,4]\mathbf G[2,3].\]

Now define the Gram matrix implied by the computed triangular factor as
$\hat{\mathbf G}=\hat{\mathbf R}^{\top}\hat{\mathbf R}.$
Using the exact real values of the displayed entries of $\hat{\mathbf R}$, we obtain
\[\hat{\mathbf G}[1,2]\hat{\mathbf G}[3,4]
-
\hat{\mathbf G}[1,3]\hat{\mathbf G}[2,4]
=
30u^2(2+5u)\neq 0.\]
Thus $\hat{\mathbf G}$ violates the invariant that every Gram matrix arising from perturbed input relations must satisfy. This contradicts
$\mathbf A'^{\top}\mathbf A'=\hat{\mathbf R}^{\top}\hat{\mathbf R}.$
Therefore there are no real perturbed input values
\[a'_1,a'_2,b'_1,b'_2,c'_1,c'_2,d'_1,d'_2\]
whose join matrix has $\hat{\mathbf R}$ as an exact QR factor.

Thus, in this example, the computed Givens QR factor is not even backward-engineerable from any perturbation of the input relations, let alone from a small backward perturbation. This proves that Givens QR, although backward stable with respect to an arbitrary materialized join matrix, need not be backward stable with respect to the input relations that generate the join matrix.
\end{proof}

The following lemma recalls the backward stability result of the QR decomposition using Givens rotations by Higham~\cite[Theorem~19.10]{Higham2002}.
\begin{lemma}
\label{lemma:qr_backward_original}
Let $\hat{\vect{R}}$ be the computed upper triangular matrix obtained by performing the QR decomposition of $\vect{A} \in \mathbb{R}^{m \times n}$ using a sequence of Givens rotations. Then there exists an orthogonal matrix $\vect{Q}'$ and a perturbation matrix $\Delta \vect{A}$, such that $\vect{A} + \Delta \vect{A} = \vect{Q}' \hat{\vect{R}}$, where $\Delta \vect{A}$ satisfies the column-wise bounds $\|\Delta \vect{A}[*,j]\|_2 \le \tilde{\gamma}_{m+n-2} \, \|\vect{A}[*,j]\|_2,  j \in [n]$, where the constant inside $\tilde{\gamma}_{m+n-2}$ is a modest integer constant independent of $m$ and $n$.
\end{lemma}

Therefore, the computation of the upper triangular factor in the QR decomposition is backward stable. It produces the exact upper triangular matrix of a matrix that differs from $\mathbf{A}$ by a small relative perturbation.

\qrnaiveprojected*

\begin{proof}
By Lemma~\ref{lemma:qr_backward_original}, there exist a perturbation $\mathbf E\in\mathbb R^{m\times n}$ and an orthogonal matrix $\mathbf Q'$ such that $\mathbf A+\mathbf E=\mathbf Q'\hat{\mathbf R}$,
with $\|\mathbf E[*,i]\|\leq\tilde\gamma_{m+n-2}\|\mathbf A[*,i]\|$ for every column $i$. Theorem~\ref{theorem:bs_implies_pbs} immediately tells us that the process of computing $g(f)$ is column-wise $(\kappa_\mathrm{join}(\mathbf{S}_j)\cdot\tilde{\gamma}_{m+n-2}, \frac{\tilde{\gamma}_{m+n-2}}{\sqrt{1-\tilde{\gamma}_{m+n-2}^2}})$-projected backward stable. This completes our proof.
\end{proof}

\begin{lemma}[Data independence of the \textsc{FiGaRo} rotations]
\label{lemma:figaro_rotations_data_independent}
Consider the \textsc{FiGaRo} setting in Def.~\ref{def:figaro-setting}. There exists an orthogonal matrix $\mathbf{Q}_0$, depending only on the fixed join keys and the join tree, such that, for every choice of the data columns,
\[\mathbf{A}=\mathbf{Q}_0\begin{bmatrix}\mathbf{R}_0\\\mathbf{0}\end{bmatrix},\]
where $\mathbf{A}$ is the corresponding materialised join matrix and $\mathbf{R}_0$ is the exact output of \textsc{FiGaRo}.
\end{lemma}

\begin{proof}
\textsc{FiGaRo} is equivalent, in exact arithmetic, to applying a collection of ordinary and generalized Givens transformations to blocks of the materialised join matrix~\cite{FIGARO:VLDBJ:2023}.

For an ordinary head-and-tail transformation applied to a block with $m$ rows, the corresponding Givens rotations have coefficients
\[\sin\theta_i=-\frac{1}{\sqrt{i}},\qquad \cos\theta_i=\sqrt{\frac{i-1}{i}},\qquad i=2,\ldots,m.\]
These coefficients depend only on the number $m$ of rows in the block and not on any data value.

For a generalized head-and-tail transformation with weight vector $\mathbf{v}=(v_1,\ldots,v_m)^{\mathsf T}$, the corresponding rotations have coefficients
\[\sin\theta_i=-\frac{v_i}{\|\mathbf{v}_{1:i}\|_2},\qquad \cos\theta_i=\frac{\|\mathbf{v}_{1:i-1}\|_2}{\|\mathbf{v}_{1:i}\|_2},\qquad i=2,\ldots,m.\]
These coefficients depend only on the weight vector $\mathbf{v}$.

Once the join keys and the join tree are fixed, every key group processed by \textsc{FiGaRo} is fixed, as is the number of rows in each group. Moreover, all entries of the vectors $\mathit{scales}$, and hence all weight vectors used by the generalized head-and-tail transformations, are determined by group sizes and by the count aggregates $\Phi_i^{\circ}$, $\Phi_i^{\uparrow}$, and $\Phi_i^{\downarrow}$. These quantities depend only on the fixed join keys and not on the data columns.

It follows that every local Givens rotation represented by \textsc{FiGaRo}, as well as the order in which these rotations act on the rows of the materialised join, is fixed once the join keys, join tree, and row orderings are fixed. Let $\mathbf{G}_0$ be the product of all these rotations, together with any fixed row permutations required by the chosen ordering. Then $\mathbf{G}_0$ is orthogonal, depends only on the join keys and the join tree, and satisfies
\[\mathbf{G}_0\mathbf{A}=\begin{bmatrix}\mathbf{R}_0\\\mathbf{0}\end{bmatrix}\]
for every choice of the data columns. Setting $\mathbf{Q}_0:=\mathbf{G}_0^{\mathsf T}$ gives
\[\mathbf{A}=\mathbf{Q}_0\begin{bmatrix}\mathbf{R}_0\\\mathbf{0}\end{bmatrix}.\]
Therefore, $\mathbf{Q}_0$ depends only on the fixed join keys and the join tree, and not on the data values.
\end{proof}

We recall the following result, stated as Lemma 3.6 by Higham~\cite{Higham2002}, adapted here to our notation:

\begin{lemma}
\label{lemma:higham_3_6}

If $\mathbf{G}_j+\Delta\mathbf{G}_j\in\mathbb{R}^{n\times n}$ satisfies $\|\Delta\mathbf{G}_j\|\leq \delta_j\|\mathbf{G}_j\|$ for all $j$, where $\|\cdot\|$ is a consistent norm, then

\[ \left\|\prod_{j=0}^{m}(\mathbf{G}_j+\Delta\mathbf{G}_j)-\prod_{j=0}^{m}\mathbf{G}_j\right\|\leq\left(\prod_{j=0}^{m}(1+\delta_j)-1\right)\prod_{j=0}^{m}\|\mathbf{G}_j\|. \]

\end{lemma}

Before presenting the proof of Theorem~\ref{theorem:QR_projected_stable}, first of all, we show that the \textsc{FiGaRo} algorithm guarantees column-wise forward stability.
\begin{lemma}
\label{lemma:figaro_one_step_lemma}
Consider the \textsc{FiGaRo} setting (Def.~\ref{def:figaro-setting}). Let $\mathbf{R}_0$ and $\hat{\mathbf{R}}_0$ be the exact and computed outputs of \textsc{FiGaRo} applied to 
the input database $\mathbf{D}$, respectively. Then for every column $i$, we have 
\[\|(\hat{\mathbf{R}}_0 - \mathbf{R}_0)[*,i]\|_2\;\le\;\tilde{\gamma}_{d}\|(\mathbf{R}_0)[*,i]\|_2\]
\end{lemma}

\begin{proof}

If $\mathbf A[*,i]=\mathbf 0$, then all values computed for this column are zero and the result follows immediately. We therefore assume that $\mathbf A[*,i]\neq\mathbf 0$.

Lemma~\ref{lemma:figaro_rotations_data_independent} states that there is an orthogonal matrix $\mathbf Q_0$ such that
\[
\mathbf A=\mathbf Q_0
\begin{bmatrix}
\mathbf R_0\\
\mathbf 0
\end{bmatrix}.
\]
Consequently,
\[
\|\mathbf A[*,i]\|_2
=
\left\|
\begin{bmatrix}
\mathbf R_0[*,i]\\
\mathbf 0
\end{bmatrix}
\right\|_2
=
\|\mathbf R_0[*,i]\|_2.
\]

We adopt the following implementation choice of \textsc{FiGaRo}: we assume that all tuple-count aggregates and scale radicands are computed exactly, such that every scale is stored through its exact integer radicand. Every scale appearing in \textsc{FiGaRo} is the square root of a positive integer. In particular, $\sqrt{\lambda_1}\cdots\sqrt{\lambda_k}=\sqrt{\lambda_1\cdots\lambda_k}$. Moreover, if the weights of a generalized head-tail computation are $v_j=\sqrt{\lambda_j}$, then each squared prefix norm is the integer
\[
v_1^2+\cdots+v_j^2=\lambda_1+\cdots+\lambda_j.
\]
Hence products of count-based scales and squared prefix norms can be maintained exactly as integer radicands. A square root is rounded only when the corresponding coefficient is formed and used. This consolidates each product of scales that occurs at one point in the algorithm; scale applications separated by other transformations are not moved or combined, and are counted separately below.

We now use the stagewise interpretation of the exact \textsc{FiGaRo} computation. At any point, conceptually restore the repeated rows represented by the current scale factors, append the rows already written to \texttt{Out}, and pad the rows not stored by \textsc{FiGaRo} with zeros. This gives a vector with the same number of entries as $\mathbf A[*,i]$. For example, multiplying a stored value $x$ by $\sqrt{\lambda}$ represents the ordinary head-tail transformation
\[
\begin{bmatrix}x\\ \vdots\\ x\end{bmatrix}\longmapsto\begin{bmatrix}\sqrt{\lambda}\,x\\0\\\vdots\\0\end{bmatrix}
\]
on $\lambda$ equal virtual copies of $x$. Generalized head-tail operations similarly combine blocks whose sizes are represented by their scale factors. Thus each exact local operation is an orthogonal transformation of the padded virtual vector.

Order these transformations as $\mathbf G_1,\ldots,\mathbf G_s$, absorbing exact permutations between stages. Each $\mathbf G_t$ is orthogonal and, after a suitable exact permutation of coordinates, is block diagonal with ordinary or generalized head-tail blocks and identity blocks. Define
\[
\mathbf z_0=\mathbf A[*,i],
\qquad
\mathbf z_t=\mathbf G_t\mathbf z_{t-1},
\qquad
t=1,\ldots,s.
\]
Every $\mathbf G_t$ is orthogonal, and therefore $\|\mathbf z_t\|_2=\|\mathbf z_{t-1}\|_2=\|\mathbf A[*,i]\|_2$.
After the final stage,
\[
\mathbf z_s=
\begin{bmatrix}
\mathbf R_0[*,i]\\
\mathbf 0
\end{bmatrix}.
\]

Apply the same virtual expansion to the floating-point values stored in \texttt{Data} and \texttt{Out}. Denote the resulting state after stage $t$ by $\hat{\mathbf z}_t$. Then
\[
\hat{\mathbf z}_0=\mathbf z_0,\qquad \hat{\mathbf z}_s=\begin{bmatrix}\hat{\mathbf R}_0[*,i]\\\mathbf 0\end{bmatrix}.
\]

We next analyze one stage. Consider one active generalized head-tail block with input vector $\mathbf{x}$ and positive weight vector $\mathbf{v}$. Let $\rho_j=\sqrt{\sum_{k=1}^j \mathbf{v}[k]^2}$. The exact head and tails computed by this block are
\[
\mathcal H(\mathbf{x},\mathbf{v})=\frac{1}{\rho_p}\sum_{k=1}^p \mathbf{v}[k]\mathbf{x}[k],
\]
and, for $j=1,\ldots,p-1$,
\[
\mathcal T_j(\mathbf{x},\mathbf{v})=\frac{\rho_j\mathbf{x}[j+1]-\dfrac{\mathbf{v}[j+1]}{\rho_j}\sum_{k=1}^j \mathbf{v}[k]\mathbf{x}[k]}{\rho_{j+1}}.
\]
An ordinary head-tail block is the special case $\mathbf{v}[1]=\cdots=\mathbf{v}[p]=1$.

$\mathcal{T}_j(\mathbf x,\mathbf v)$ is also the formula evaluated by the sequential running-sum implementation. Importantly, a fixed input $x_k$ contributes through only one algebraic term to each nonzero output coefficient. Thus, after expanding the rounded running sums, multiplications, divisions, square roots, and subtractions, each computed coefficient equals the corresponding exact coefficient multiplied by a product of elementary rounding factors $(1+\delta)^{\pm1}$. 

Let $q_t$ be an upper bound on the number of such elementary rounding factors occurring along any coefficient path in stage $t$. The standard $\theta$-notation therefore gives
\[
\hat{\mathbf z}_t[j]=\sum_k \mathbf G_t[j,k]\bigl(1+\theta_{jk}^{(t)}\bigr)\hat{\mathbf z}_{t-1}[k],\qquad |\theta_{jk}^{(t)}|\leq\gamma_{q_t}.
\]
Equivalently, there exists a perturbation matrix $\Delta\mathbf G_t$ such that
\[
\hat{\mathbf z}_t
=
(\mathbf G_t+\Delta\mathbf G_t)\hat{\mathbf z}_{t-1},
\qquad
|\Delta\mathbf G_t|
\leq
\gamma_{q_t}|\mathbf G_t|,
\]
where the absolute values and inequality are entry-wise. 

We now bound the norm of the entrywise absolute value of an exact head-tail block. Let
\[
\mathbf G=\begin{bmatrix}\mathbf h^{\mathsf T}\\\mathbf T\end{bmatrix}
\]
be one such block. If the size of the head-tail transformation is 1, then $\mathbf G=[1]$ and $\||\mathbf G|\|_2=1$. Suppose therefore that the size is at least $2$. Since $\mathbf G$ is orthogonal, its head has unit norm and its tail rows are orthonormal:
\[
\|\mathbf h\|_2=1,\qquad \|\mathbf T\|_2=1.
\]
Every generalized tail row contains one nonnegative coefficient corresponding to the newly introduced input and nonpositive coefficients corresponding to preceding inputs. Hence the tail matrix can be written as $\mathbf T=\mathbf P-\mathbf L$, where $\mathbf P$ and $\mathbf L$ are entrywise nonnegative. The matrix $\mathbf P$ contains the single positive coefficient in each tail row. It has at most one nonzero entry in each row and column, and every such entry has magnitude at most one. Therefore, $\|\mathbf P\|_2\leq1$. Since $\mathbf L=\mathbf P-\mathbf T$,
\[
\|\mathbf L\|_2
\leq
\|\mathbf P\|_2+\|\mathbf T\|_2
\leq2.
\]
The nonzero entry positions of $\mathbf P$ and $\mathbf L$ are disjoint, so $|\mathbf T|=\mathbf P+\mathbf L$.
It follows that
\[
\bigl\||\mathbf T|\bigr\|_2
\leq
\|\mathbf P\|_2+\|\mathbf L\|_2
\leq3.
\]
Combining the head and tail bounds gives
\[
\bigl\||\mathbf G|\bigr\|_2
\leq
\sqrt{\|\mathbf h\|_2^2+\bigl\||\mathbf T|\bigr\|_2^2}
\leq
\sqrt{1+3^2}
=
\sqrt{10}.
\]
The same estimate holds for $\mathbf G_t$, because $\mathbf G_t$ is block diagonal and the spectral norm of a block-diagonal matrix is the maximum spectral norm of its blocks. Identity and permutation blocks have norm one. Thus, $\bigl\||\mathbf G_t|\bigr\|_2\leq\sqrt{10}$.
We now have the following stagewise bound: 
\[
\begin{aligned}
\|\Delta\mathbf G_t\|_2\leq \bigl\||\Delta\mathbf G_t|\bigr\|_2\leq \gamma_{q_t}\bigl\||\mathbf G_t|\bigr\|_2\leq \sqrt{10}\,\gamma_{q_t}\leq 4\gamma_{q_t}\leq \gamma_{4q_t},
\end{aligned}
\]
provided $4q_t u<1$. Since $\mathbf G_t$ is orthogonal, $\|\mathbf G_t\|_2=1$. Hence $\|\Delta\mathbf G_t\|_2\leq\gamma_{4q_t}\|\mathbf G_t\|_2$.

Since $\hat{\mathbf z}_0=\mathbf z_0=\mathbf A[*,i]$, repeated application of $\hat{\mathbf z}_t=(\mathbf G_t+\Delta\mathbf G_t)\hat{\mathbf z}_{t-1}$ gives
\[
\hat{\mathbf z}_s=(\mathbf G_s+\Delta\mathbf G_s)\cdots(\mathbf G_1+\Delta\mathbf G_1)\mathbf z_0,
\]
whereas the exact computation satisfies
\[
\mathbf z_s=\mathbf G_s\cdots\mathbf G_1\mathbf z_0.
\]
Applying Lemma~\ref{lemma:higham_3_6} in the spectral norm therefore gives
\[
\begin{aligned}
\|\hat{\mathbf z}_s-\mathbf z_s\|_2
&\leq
\left\|
(\mathbf G_s+\Delta\mathbf G_s)\cdots(\mathbf G_1+\Delta\mathbf G_1)
-\mathbf G_s\cdots\mathbf G_1
\right\|_2
\|\mathbf z_0\|_2\\
&\leq
\left(\prod_{t=1}^{s}(1+\gamma_{4q_t})-1\right)
\left(\prod_{t=1}^{s}\|\mathbf G_t\|_2\right)
\|\mathbf z_0\|_2\\
&=
\left(\prod_{t=1}^{s}(1+\gamma_{4q_t})-1\right)
\|\mathbf z_0\|_2.
\end{aligned}
\]
Let $\nu_i:=\sum_{t=1}^{s}q_t$. Using $1+\gamma_x=\frac{1}{1-xu}$ and $\prod_{t=1}^{s}(1-4q_tu)\geq 1-4\sum_{t=1}^{s}q_tu$,
we obtain
\[
\prod_{t=1}^{s}(1+\gamma_{4q_t})\leq 1+\gamma_{4\nu_i}.
\]
Consequently, $\|\hat{\mathbf z}_s-\mathbf z_s\|_2\leq\gamma_{4\nu_i}\|\mathbf z_0\|_2$.

A head-tail computation on a group with $p$ stored rows has arithmetic depth at most a fixed constant times $p$: the running sums are updated once per row, and every head or tail uses only a fixed number of additional floating-point operations.

For each relation $S_j$, the ordinary head-tail groups partition the rows of $S_j$, so the sum of their sizes is $d_j$, where $d_j$ denotes the number of rows in $S_j$. Immediately before projecting $S_j$ toward its parent, the current \texttt{Data} matrix has at most one row for each distinct join-key value occurring in $S_j$, and therefore at most $d_j$ rows. The generalized head-tail groups partition these rows, so the sum of their sizes is also at most $d_j$.

Finally, for a fixed column, every scale product required for one stored entry is first consolidated into one exact integer radicand. Its floating-point application therefore has constant arithmetic depth, independent of the number of factors in that product. Under the no-self-join assumption, a fixed data column is carried through a unique path toward the root, and the number of stored entries of that column handled at relation $S_j$ is at most a fixed constant times $d_j$. The tail-scaling operations are bounded in the same way.

Consequently, there is an absolute implementation constant $c$ such that $\nu_i\leq c\sum_{j=1}^r d_j=cd$.
In particular, this bound has no additional dependence on the number of relations or on the size of the materialized join. Therefore, the total number of rounded operations along all stages affecting column $i$ satisfies $\nu_i\leq cd$,
where $d$ is the total number of rows in the input database and $c$ is independent of the materialized join size. Therefore,
\[
\|\hat{\mathbf z}_s-\mathbf z_s\|_2\leq\gamma_{4cd}\|\mathbf A[*,i]\|_2=\tilde{\gamma}_d\|\mathbf A[*,i]\|_2,
\]
provided $4cd u<1$.

Finally,
\[
\mathbf z_s=
\begin{bmatrix}
\mathbf R_0[*,i]\\
\mathbf 0
\end{bmatrix},
\qquad
\hat{\mathbf z}_s=
\begin{bmatrix}
\hat{\mathbf R}_0[*,i]\\
\mathbf 0
\end{bmatrix}.
\]
Hence
\[
\|\hat{\mathbf R}_0[*,i]-\mathbf R_0[*,i]\|_2=\|\hat{\mathbf z}_s-\mathbf z_s\|_2\leq\tilde{\gamma}_d\|\mathbf A[*,i]\|_2.
\]
Since $\mathbf A=\mathbf Q_0[\mathbf R_0;\mathbf 0]$ with $\mathbf Q_0$ orthogonal,
\[
\|\mathbf A[*,i]\|_2=\|\mathbf R_0[*,i]\|_2.
\]
Therefore,
\[
\|\hat{\mathbf R}_0[*,i]-\mathbf R_0[*,i]\|_2\leq\tilde{\gamma}_d\|\mathbf R_0[*,i]\|_2,
\]
which proves the result.
\end{proof}

\QRprojectedstable*

\begin{proof}
Let $f(\mathbf{D})=\mathbf{A}$ be the join, and $g(\mathbf{A})=\mathbf{R}$ be the function mapping the join matrix to the upper-triangular $\mathbf{R}$, and we use $\hat{\mathbf{R}}$ to denote the computed result from \textsc{FiGaRo} followed by Givens Rotations. Notice that this computation computes $g\circ f$ without ever materializing the join matrix $\mathbf{A}$. Let $\mathbf R_0$ and $\hat{\mathbf R}_0$ denote the exact and computed outputs of $\textsc{FiGaRo}$ on $\mathbf D$, respectively. Whenever a FiGaRo output is multiplied by an $m\times m$ matrix, we implicitly pad it with zero rows to obtain an $m\times n$ matrix. By Lemma~\ref{lemma:figaro_rotations_data_independent}, there exists an orthogonal matrix $\mathbf Q_0$, determined only by the fixed key columns and the join tree $\tau$, such that
\[\mathbf A=\mathbf Q_0\begin{bmatrix}\mathbf R_0\\\mathbf 0\end{bmatrix}.\]

Define the \textsc{FiGaRo} forward error by $\mathbf E_F:=\hat{\mathbf R}_0-\mathbf R_0$. By Lemma~\ref{lemma:figaro_one_step_lemma}, for every column $i$, $\|\mathbf E_F[*,i]\|_2\leq\tilde{\gamma}_{d}\|\mathbf R_0[*,i]\|_2.$ The matrix $\hat{\mathbf R}_0$ has at most $d$ nonzero rows. Therefore, by Lemma~\ref{lemma:qr_backward_original}, there exist a perturbation $\mathbf E_Q$ and an orthogonal matrix $\mathbf Q'$ such that
$\hat{\mathbf R}_0+\mathbf E_Q=\mathbf Q'\hat{\mathbf R}$, with
$\|\mathbf E_Q[*,i]\|_2\leq\tilde{\gamma}_{d+n-2}\|\hat{\mathbf R}_0[*,i]\|_2$
for every column $i$. Since
\[\|\hat{\mathbf R}_0[*,i]\|_2\leq\|\mathbf R_0[*,i]\|_2+\|\mathbf E_F[*,i]\|_2\leq(1+\tilde{\gamma}_{d})\|\mathbf R_0[*,i]\|_2,\]
we obtain $\|\mathbf E_Q[*,i]\|_2\leq\tilde{\gamma}_{d+n-2}(1+\tilde{\gamma}_{d})\|\mathbf R_0[*,i]\|_2$.

Define the total perturbation in the FiGaRo-output space by $\mathbf E:=\mathbf E_F+\mathbf E_Q$.
Since $\hat{\mathbf R}_0=\mathbf R_0+\mathbf E_F$, we have
\[\mathbf R_0+\mathbf E=\hat{\mathbf R}_0+\mathbf E_Q=\mathbf Q'\hat{\mathbf R}.\]
Moreover, for every column $i$,
\[\|\mathbf E[*,i]\|_2\leq\|\mathbf E_F[*,i]\|_2+\|\mathbf E_Q[*,i]\|_2\leq\eta\|\mathbf R_0[*,i]\|_2,\]
where
$\eta=\tilde{\gamma}_{d}+\tilde{\gamma}_{d+n-2}(1+\tilde{\gamma}_{d})$.

Define \[\mathbf E_A:=\mathbf Q_0\begin{bmatrix}\mathbf E\\\mathbf 0\end{bmatrix}.\]
Because $\mathbf Q_0$ is orthogonal, for every column $i$,
\[\|\mathbf E_A[*,i]\|_2=\|\mathbf E[*,i]\|_2\leq\eta\|\mathbf R_0[*,i]\|_2=\eta\|\mathbf A[*,i]\|_2.\]
Furthermore,
\[\mathbf A+\mathbf E_A=\mathbf Q_0\begin{bmatrix}\mathbf R_0+\mathbf E\\\mathbf 0\end{bmatrix}=\mathbf Q_0\begin{bmatrix}\mathbf Q'\hat{\mathbf R}\\\mathbf 0\end{bmatrix}.\]
After padding $\mathbf Q'$ with an identity matrix on the omitted zero rows, the matrix multiplying $\hat{\mathbf R}$ on the right-hand side is orthogonal. Hence $\hat{\mathbf R}$ is the exact upper triangular QR factor of $\mathbf A+\mathbf E_A$, and therefore
$\hat{\mathbf R}=g(\mathbf A+\mathbf E_A)$.

Let $\mathbf P_{\mathcal Y_0}$ denote the orthogonal projector onto the join space $\mathcal Y_0$, and define
\[\tilde{\mathbf A}:=\mathbf P_{\mathcal Y_0}(\mathbf A+\mathbf E_A),\qquad \Delta\mathbf A^\perp:=(\mathbf{I}-\mathbf P_{\mathcal Y_0})(\mathbf A+\mathbf E_A).\]
Then $\tilde{\mathbf A}\in\mathcal Y_0$, $\Delta\mathbf A^\perp\in\mathcal Y_0^\perp$, and
$\tilde{\mathbf A}+\Delta\mathbf A^\perp=\mathbf A+\mathbf E_A$.
Consequently,
$\hat{\mathbf R}=g\bigl(\tilde{\mathbf A}+\Delta\mathbf A^\perp\bigr)$. Since $\mathbf A\in\mathcal Y_0$, we have
$\tilde{\mathbf A}-\mathbf A=\mathbf P_{\mathcal Y_0}\mathbf E_A$ and $\Delta\mathbf A^\perp=(\mathbf{I}-\mathbf P_{\mathcal Y_0})\mathbf E_A$.
Since the query has no self-joins, the columns of $\mathcal Y_0$ are parameterized independently and the projection acts independently on every column. Hence, by the Pythagorean theorem,
\[\|(\tilde{\mathbf A}-\mathbf A)[*,i]\|_2^2+\|\Delta\mathbf A^\perp[*,i]\|_2^2=\|\mathbf E_A[*,i]\|_2^2.\]
It follows that
\[\|(\tilde{\mathbf A}-\mathbf A)[*,i]\|_2\leq\|\mathbf E_A[*,i]\|_2\leq\eta\|\mathbf A[*,i]\|_2.\]
Since $\tilde{\mathbf A}\in\mathcal Y_0$, there exists a perturbation $\Delta\mathbf D=(\Delta\mathbf S_1,\ldots,\Delta\mathbf S_r)$ such that $f(\mathbf D+\Delta\mathbf D)=\tilde{\mathbf A}$.
Applying Proposition~\ref{proposition:condition_number_translation_with_dangling_tuples} to $\tilde{\mathbf A}-\mathbf A$, we obtain, for every relation $\mathbf S_j$ and every data column $k$ of $\mathbf S_j$, $\|\Delta\mathbf S_j[*,k]\|_2\leq\kappa_\mathrm{join}(\mathbf S_j)\eta\|\mathbf S_j[*,k]\|_2$.

It remains to bound the normal perturbation. Fix a column $i$ with $\mathbf A[*,i]\neq\mathbf 0$, and let
\[\lambda_i:=\frac{\|(\tilde{\mathbf A}-\mathbf A)[*,i]\|_2}{\|\mathbf A[*,i]\|_2}.\]
Then $0\leq\lambda_i\leq\eta<1$. By the triangle inequality,
\[\|\tilde{\mathbf A}[*,i]\|_2\geq\|\mathbf A[*,i]\|_2-\|(\tilde{\mathbf A}-\mathbf A)[*,i]\|_2=(1-\lambda_i)\|\mathbf A[*,i]\|_2.\]
The Pythagorean theorem and the bound on $\mathbf E_A$ give
\[\|\Delta\mathbf A^\perp[*,i]\|_2^2\leq(\eta^2-\lambda_i^2)\|\mathbf A[*,i]\|_2^2.\]
Notice that if $\|\tilde{\mathbf A}[*,i]\|_2=0$, the two inequalities above directly lead to $\|\mathbf A[*,i]\|_2=\|\Delta\mathbf A^\perp[*,i]\|_2=0$, and the bound on the normal perturbation trivially holds. Now if $\|\Delta\mathbf A^\perp[*,i]\|_2>0$, we have:
\[\frac{\|\Delta\mathbf A^\perp[*,i]\|_2}{\|\tilde{\mathbf A}[*,i]\|_2}\leq\frac{\sqrt{\eta^2-\lambda_i^2}}{1-\lambda_i}\leq\frac{\eta}{\sqrt{1-\eta^2}},\]
where the last inequality follows because the function $h(\lambda)=\frac{\sqrt{\eta^2-\lambda^2}}{1-\lambda}$ reaches its maximum over $0\leq\lambda\leq\eta$ at $\lambda=\eta^2$. 

Thus, $\hat{\mathbf R}=g(\tilde{\mathbf A}+\Delta\mathbf A^\perp)$, where $\tilde{\mathbf A}$ is the join matrix of a columnwise-small perturbation of the input database and $\Delta\mathbf A^\perp\in\mathcal Y_0^\perp$ satisfies the claimed normal-perturbation bound. Therefore, $\textsc{FiGaRo}$ followed by Givens rotations is column-wise $\Big(\kappa_\mathrm{join}(\mathbf{S}_j)\cdot\eta, \frac{\eta}{\sqrt{1-\eta^2}}\Big)$-projected backward stable for each input data column in $\vect S_j$ and $j\in[r]$. 
\end{proof}

\SVDprojectedstable*

\begin{proof}
Let $f(\mathbf{D})=\mathbf{A}$ be the join, and $g(\mathbf{A})=\mathbf{\Sigma}$ be the function mapping the join matrix to the diagonal matrix containing its singular values, and we use $\hat{\mathbf{\Sigma}}$ to denote the computed result from \textsc{FiGaRo} followed by \texttt{xGESVD}. Notice that this computation computes $g\circ f$ without ever materializing the join matrix $\mathbf{A}$. Let $\mathbf R_0$ and $\hat{\mathbf R}_0$ denote the exact and computed outputs of $\textsc{FiGaRo}$ on $\mathbf D$, respectively. Whenever a FiGaRo output is multiplied by an $m\times m$ matrix, we implicitly pad it with zero rows to obtain an $m\times n$ matrix. By Lemma~\ref{lemma:figaro_rotations_data_independent}, there exists an orthogonal matrix $\mathbf Q_0$, determined only by the fixed key columns and the join tree $\tau$, such that
\[\mathbf A=\mathbf Q_0\begin{bmatrix}\mathbf R_0\\\mathbf 0\end{bmatrix}.\]

Define the \textsc{FiGaRo} forward error by $\mathbf E_F:=\hat{\mathbf R}_0-\mathbf R_0$. Lemma~\ref{lemma:figaro_one_step_lemma} claims the forward error of \textsc{FiGaRo} in every column is upper-bounded by $\tilde{\gamma}_d$, which directly gives us a Frobenius-norm-wise forward error bound  $\tilde{\gamma}_d$ too. Namely, $\|\mathbf E_F\|_F\leq\tilde{\gamma}_{d}\|\mathbf R_0\|_F.$ The matrix $\hat{\mathbf R}_0$ has at most $d$ nonzero rows. Therefore, computing $\hat{\mathbf{\Sigma}}$ with \texttt{xGESVD} from $\hat{\mathbf R}_0$ means there exist a perturbation $\mathbf E_Q$ and orthogonal matrices $\mathbf{U}'$ and $\mathbf{V}'$ such that
$\hat{\mathbf R}_0+\mathbf E_Q=\mathbf {U'}\hat{\mathbf \Sigma}\mathbf{V'}^\top$, with
$\|\mathbf E_Q\|_F\leq\epsilon_\mathrm{svd}(d,n)\cdot\|\hat{\mathbf R}_0\|_F$. Since
\[\|\hat{\mathbf R}_0\|_F\leq\|\mathbf R_0\|_F+\|\mathbf E_F\|_F\leq(1+\tilde{\gamma}_{d})\|\mathbf R_0\|_F,\]
we obtain $\|\mathbf E_Q\|_F\leq\epsilon_\mathrm{svd}(d,n)\cdot(1+\tilde{\gamma}_{d})\|\mathbf R_0\|_F$.

Define the total perturbation in the FiGaRo-output space by $\mathbf E:=\mathbf E_F+\mathbf E_Q$.
Since $\hat{\mathbf R}_0=\mathbf R_0+\mathbf E_F$, we have
\[\mathbf R_0+\mathbf E=\hat{\mathbf R}_0+\mathbf E_Q=\mathbf U'\hat{\mathbf \Sigma}\mathbf{V}'^\top.\]
Moreover, $\|\mathbf E\|_F\leq\|\mathbf E_F\|_F+\|\mathbf E_Q\|_F\leq\eta\|\mathbf R_0\|_F$,
where
$\eta=\tilde{\gamma}_{d}+\epsilon_\mathrm{svd}(d,n)\cdot(1+\tilde{\gamma}_{d})$.

Define \[\mathbf E_A:=\mathbf Q_0\begin{bmatrix}\mathbf E\\\mathbf 0\end{bmatrix}.\]
Because $\mathbf Q_0$ is orthogonal,
\[\|\mathbf E_A\|_F=\|\mathbf E\|_F\leq\eta\|\mathbf R_0\|_F=\eta\|\mathbf A\|_F.\]
Furthermore,
\[\mathbf A+\mathbf E_A=\mathbf Q_0\begin{bmatrix}\mathbf R_0+\mathbf E\\\mathbf 0\end{bmatrix}=\mathbf Q_0\begin{bmatrix}\mathbf U'\hat{\mathbf \Sigma}\mathbf{V'}^\top\\\mathbf 0\end{bmatrix}.\]
After padding $\mathbf U'$ with an identity matrix on the omitted zero rows, the matrix multiplying $\hat{\mathbf \Sigma}\mathbf{V}'^\top$ on the right-hand side is orthogonal. Hence $\hat{\mathbf \Sigma}$ is the exact diagonal matrix containing the singular values of $\mathbf A+\mathbf E_A$, and therefore
$\hat{\mathbf \Sigma}=g(\mathbf A+\mathbf E_A)$.

Let $\mathbf P_{\mathcal Y_0}$ denote the orthogonal projector onto the join space $\mathcal Y_0$, and define
\[\tilde{\mathbf A}:=\mathbf P_{\mathcal Y_0}(\mathbf A+\mathbf E_A),\qquad \Delta\mathbf A^\perp:=(\mathbf{I}-\mathbf P_{\mathcal Y_0})(\mathbf A+\mathbf E_A).\]
Then $\tilde{\mathbf A}\in\mathcal Y_0$, $\Delta\mathbf A^\perp\in\mathcal Y_0^\perp$, and
$\tilde{\mathbf A}+\Delta\mathbf A^\perp=\mathbf A+\mathbf E_A$.
Consequently,
$\hat{\mathbf \Sigma}=g\bigl(\tilde{\mathbf A}+\Delta\mathbf A^\perp\bigr)$. Since $\mathbf A\in\mathcal Y_0$, we have
$\tilde{\mathbf A}-\mathbf A=\mathbf P_{\mathcal Y_0}\mathbf E_A$ and $\Delta\mathbf A^\perp=(\mathbf{I}-\mathbf P_{\mathcal Y_0})\mathbf E_A$.
By the Pythagorean theorem,
\[\|(\tilde{\mathbf A}-\mathbf A)\|_F^2+\|\Delta\mathbf A^\perp\|_F^2=\|\mathbf E_A\|_F^2.\]
It follows that
\[\|(\tilde{\mathbf A}-\mathbf A)\|_F\leq\|\mathbf E_A\|_F\leq\eta\|\mathbf A\|_F.\]
Since $\tilde{\mathbf A}\in\mathcal Y_0$, there exists a perturbation $\Delta\mathbf D=(\Delta\mathbf S_1,\ldots,\Delta\mathbf S_r)$ such that $f(\mathbf D+\Delta\mathbf D)=\tilde{\mathbf A}$.
Applying Proposition~\ref{proposition:condition_number_translation_with_dangling_tuples} to $\tilde{\mathbf A}-\mathbf A$, we obtain $\|\Delta\mathbf D\|_F\leq\kappa_\mathrm{join}(\mathbf D)\eta\|\mathbf D\|_F$.

It remains to bound the normal perturbation. Notice that if the join matrix is a zero matrix, then the computed singular values will be all zero and we trivially recover the input database with all zero entries with no perturbation. Now suppose $\mathbf A\neq\mathbf 0$, and let
\[\lambda:=\frac{\|(\tilde{\mathbf A}-\mathbf A)\|_F}{\|\mathbf A\|_F}.\]
Then $0\leq\lambda\leq\eta<1$. By the triangle inequality,
\[\|\tilde{\mathbf A}\|_F\geq\|\mathbf A\|_F-\|(\tilde{\mathbf A}-\mathbf A)\|_F=(1-\lambda)\|\mathbf A\|_F.\]
The Pythagorean theorem and the bound on $\mathbf E_A$ give
\[\|\Delta\mathbf A^\perp\|_F^2\leq(\eta^2-\lambda^2)\|\mathbf A\|_F^2.\]
Notice that if $\|\tilde{\mathbf A}\|_F=0$, the two inequalities above directly lead to $\|\mathbf A\|_F=\|\Delta\mathbf A^\perp\|_F=0$, and the bound on the normal perturbation trivially holds. Now if $\|\Delta\mathbf A^\perp\|_F>0$, we have:
\[\frac{\|\Delta\mathbf A^\perp\|_F}{\|\tilde{\mathbf A}\|_F}\leq\frac{\sqrt{\eta^2-\lambda^2}}{1-\lambda}\leq\frac{\eta}{\sqrt{1-\eta^2}},\]
where the last inequality follows because the function $h(\lambda)=\frac{\sqrt{\eta^2-\lambda^2}}{1-\lambda}$ reaches its maximum over $0\leq\lambda\leq\eta$ at $\lambda=\eta^2$. 

Thus, $\hat{\mathbf \Sigma}=g(\tilde{\mathbf A}+\Delta\mathbf A^\perp)$, where $\tilde{\mathbf A}$ is the join matrix of a Frobenius-norm-wise-small perturbation of the input database and $\Delta\mathbf A^\perp\in\mathcal Y_0^\perp$ satisfies the claimed normal-perturbation bound. Therefore, $\textsc{FiGaRo}$ followed by \texttt{xGESVD} is Frobenius-norm-wise \\$\Big(\kappa_\mathrm{join}(\mathbf{D})\cdot\eta, \frac{\eta}{\sqrt{1-\eta^2}}\Big)$-projected backward stable, where $\eta:=\tilde{\gamma}_{d}+\epsilon_\mathrm{svd}(d,n)\cdot\bigl(1+\tilde{\gamma}_{d}\bigr)$ and when $\eta<1$. 
\end{proof}

\end{document}